%% file: main.tex
\pdfoutput=1
\documentclass[12pt, reqno]{amsart}
\usepackage{amssymb, mathtools, bbm}
\usepackage[utf8]{inputenc}
\usepackage[margin=1.25in]{geometry}
\usepackage{times}
\usepackage[all=normal, bibbreaks=tight, floats=tight,
            mathdisplays=tight, bibnotes=tight]{savetrees}
\usepackage{graphicx}
\usepackage{caption, subcaption}
\usepackage[svgnames]{xcolor}
\usepackage{enumitem}
\usepackage{array, booktabs, threeparttable}
\usepackage{setspace}
\usepackage{float}
\usepackage{placeins}
\usepackage{etoolbox}
\usepackage{natbib}

\definecolor{LinkNavy}{RGB}{0,45,95}
\usepackage[colorlinks=true,
            linkcolor=LinkNavy,
            urlcolor=LinkNavy,
            citecolor=LinkNavy,
            filecolor=LinkNavy,
            pdfborder={0 0 0}
           ]{hyperref}
\hypersetup{pdftitle={Assumption-Lean Shrinkage and Model Averaging for Spatial Parameters},
            pdfauthor={Harvey Barnhard}}

\patchcmd{\section}{\scshape}{\bfseries}{}{}
\patchcmd{\subsection}{\bfseries\itshape}{\bfseries}{}{}
\makeatletter
\patchcmd{\@settitle}{\uppercasenonmath\@title}{\large}{}{}
\patchcmd{\@setauthors}{\MakeUppercase}{\vspace{-1em}\normalsize}{}{}
\makeatother
\newtheorem{theorem}{Theorem}[section]
\newtheorem{proposition}[theorem]{Proposition}
\newtheorem{lemma}[theorem]{Lemma}
\newtheorem{corollary}[theorem]{Corollary}
\theoremstyle{definition}
\newtheorem{definition}[theorem]{Definition}
\newtheorem{exampleinner}[theorem]{Example}
\newenvironment{example}
  {\pushQED{\qed}\exampleinner}
  {\popQED\endexampleinner}
\newtheorem{assumption}[theorem]{Assumption}
\theoremstyle{remark}
\newtheorem{remark}[theorem]{Remark}

\numberwithin{equation}{section}

\newcommand{\bbR}{\mathbb{R}}

\newcommand{\E}{\mathbb{E}}

\newcommand{\Cov}{\mathrm{Cov}}
\newcommand{\tr}{\mathrm{tr}}
\DeclareMathOperator*{\argmin}{arg\,min}
\DeclareMathOperator*{\argmax}{arg\,max}
\DeclareMathOperator{\diag}{diag}
\newcommand{\N}{\mathcal{N}}

\newcommand{\scmethod}[1]{\ifmmode\text{\textup{\textsc{#1}}}\else\textup{\textsc{#1}}\fi}
\newcommand{\sure}{\scmethod{sure}}
\newcommand{\mle}{\scmethod{mle}}
\newcommand{\nneb}{\scmethod{nn-eb}}
\newcommand{\close}{\scmethod{close}}
\newcommand{\closegauss}{\close-\scmethod{gauss}}
\newcommand{\gp}{\scmethod{gp}}

\newcommand{\bilat}{\scmethod{bilat}}
\newcommand{\gpbilat}{\gp-\bilat}

\input{tables/auto_oa_results}

\input{tables/auto_oa_proxy_gaps}

\input{tables/auto_oa_assure_welfare}

\input{tables/auto_oa_target_scatter}

\input{tables/auto_oa_ladder}

\begin{document}

\title{Assumption-Lean Shrinkage and Model Averaging for Spatial Parameters}
\author{Harvey Barnhard \\ Department of Economics, Harvard University}
\thanks{Email: \href{mailto:hbarnhard@g.harvard.edu}{hbarnhard@g.harvard.edu}. I am grateful to Isaiah Andrews, Jacob Carlson, Edward Glaeser, Neil Shephard, Rahul Singh, Elie Tamer, Davide Viviano, and participants in the Harvard econometrics workshop for helpful comments and suggestions. This material is based upon work supported by the National Science Foundation Graduate Research Fellowship Program under Grant No. DGE-2140743. Any opinions, findings, and conclusions or recommendations expressed in this material are those of the author and do not necessarily reflect the views of the National Science Foundation.}
\date{June 2026}

\begingroup
\onehalfspacing
\begin{abstract}
Economic decisions often depend on many noisy estimates of quantities such as neighborhood effects, school quality, and hospital performance. Shrinkage estimation can improve decisions by pooling information across related units, but geography, adjacency, and shared characteristics each define a different notion of relatedness, and each implies a different way of pooling. We treat the choice of relatedness as part of the estimation problem, using Stein's Unbiased Risk Estimate (\sure{}) to form a weighted average over a library of flexible shrinkage estimators. This comparison among the candidate estimators treats no prior or latent covariance structure as a correctly specified model for the parameters being estimated. Each candidate is judged by its \sure{} value. Under smoothness conditions on the estimators, the \sure{}-weighted average performs nearly as well as the best fixed weighted average of trained candidates, including nonlinear rules whose reported values use the full vector of noisy estimates. In an application to Opportunity Atlas economic mobility data from \oaNCzs{} commuting zones, the best individual spatial specification varies across zones, yet the \sure{}-weighted average tracks the best in each zone and reduces estimated mean squared error by about \oaReductionAggVsClose\% relative to the best-performing non-spatial empirical Bayes baseline in our library of estimators.
\end{abstract}

\maketitle
\endgroup
\newpage
\onehalfspacing

\section{Introduction}

When economic decisions rest on thousands of noisy estimates, how much should each estimate borrow from its neighbors, and which neighbors should count? The same issue arises for noisy estimates at the school \citep{kaneEstimatingTeacherImpacts2008, chettyMeasuringImpactsTeachersI2014}, hospital \citep{dimickRankingHospitalsReliability2010, hullEstimatingHospitalQuality2020}, firm \citep{klineSystemicDiscrimination2022}, or small-area \citep{fayEstimatesIncomeSmall1979} level whenever researchers have several plausible ways to define related units. The Opportunity Atlas \citep{chettyOpportunityAtlasMapping2018} estimates of neighborhood economic mobility helped identify high-opportunity neighborhoods in the Creating Moves to Opportunity program \citep{bergmanCreatingMovesOpportunity2024}, yet many of the underlying tract-level estimates are noisy due to few observations.\footnote{Title~I education funding under the Elementary and Secondary Education Act is allocated using the Census Bureau's Small Area Income and Poverty Estimates for over 13,000 school districts. Medicare's Hospital Readmissions Reduction Program adjusts payments to thousands of hospitals based on noisy risk-adjusted readmission rates \citep{guptaPerformancePayHospitals2021}. See \cite{waltersEmpiricalBayesMethods2024} for an overview of empirical Bayes methods in economics.} Shrinkage estimation can improve these estimates by borrowing strength from related units. But borrowing from \textit{which} units? Adjacent neighborhoods tend to have similar economic mobility, while neighborhoods separated by a highway or school district boundary may differ sharply. This paper treats that choice as part of the estimation problem: we build a library of shrinkage rules that encode different notions of relatedness, then use the data to average over the resulting estimates---with weights that can, and sometimes do, concentrate on a single rule.

We observe noisy estimates $Y=(Y_1,\ldots,Y_n)^\top$ of unobserved parameters $\theta = (\theta_1,\ldots,\theta_n)^\top$, written in vector form as
\[
Y = \theta + \varepsilon,
\]
where $\theta$ is the latent parameter vector to be estimated and $\varepsilon$ is sampling noise. The goal is to estimate $\theta$ accurately by pooling information across related units, and the question is how to pool. Empirical Bayes (EB) methods are a natural starting point. These methods estimate a prior distribution for $\theta$ and report the implied posterior means as the shrinkage estimates. Much of the nonparametric EB literature makes the prior flexible while retaining exchangeability across units:\footnote{A prior is \emph{exchangeable} if it is invariant to relabeling the units: it can encode how spread out the $\theta_i$ are, but not which particular units are related. Section~\ref{sec:estimators} states this formally.} the resulting rule does not use information about which units are close in space or otherwise linked \citep[e.g.,][]{kiefer_consistency_1956, jiang_general_2009, koenkerMizeraConvexOptimization2014, soloff_multivariate_2025}. Other EB approaches relax exchangeability by allowing the prior to vary with precision or covariates \citep{ignatiadisCovariatePoweredEmpiricalBayes2021, chenEmpiricalBayesWhen2024, luo_empirical_2025}. We treat such prior specifications as one way to construct candidate shrinkage maps. A prior specification or covariance structure can motivate a shrinkage map $Y\mapsto f(Y)$ that returns an estimate of $\theta$, but we ask which map performs better under squared-error loss, not which specification is the right model for $\theta$. In this sense our approach is \emph{assumption-lean} about the latent vector: it requires a model for the sampling noise in the estimates, but no model for the latent parameters themselves. The candidate library of shrinkage maps can therefore contain EB posterior-mean maps, spatial maps motivated by covariance models or adjacency structures, and maps whose tuning parameters are estimated from the same noisy estimates.

Candidate shrinkage maps can differ both in which other units enter each reported value and in how their tuning parameters are trained: one map may shrink through an estimated prior, another may give more weight to nearby or adjacent units, and another may estimate its smoothing weights from the observed vector $Y$. This flexibility can lower estimation error, but it creates an overfitting risk: the map that looks best for the observed vector $Y$ may be fitting the sampling noise $\varepsilon$ rather than estimating $\theta$ well. When the sampling noise is Gaussian with known covariance $\Sigma$, $\varepsilon \sim \N(0, \Sigma)$, Stein's Unbiased Risk Estimate (\sure{}; \citealp{steinEstimationMeanMultivariate1981}) provides an observable, unbiased estimate of each candidate's expected squared-error loss. Selecting the candidate with the smallest \sure{} therefore guards against overfitting, favoring the map expected to estimate $\theta$ best rather than the one that fits the observed $Y$ most closely. Such an estimate matters because cross-validation does not measure the right quantity here: with one noisy estimate per latent parameter $\theta_i$, holding out unit $i$ does not pin down the squared error relative to $\theta_i$. \sure{} avoids this validation problem without placing a distributional assumption on $\theta$, and lets us compare candidate shrinkage maps on the squared-error loss scale used to evaluate estimation of $\theta$.

Our procedure has two steps. The researcher first specifies candidate classes and how each is trained.  Applied to the observed data, these choices yield a finite library of trained maps $Y \mapsto f_k(Y)$, $k=1,\ldots,K$. \sure{} is then evaluated for each trained map, including the correction for parameters trained on the same data, and minimizing \sure{} over convex weights yields the \sure{}-weighted average $Y \mapsto \sum_{k=1}^K w_k f_k(Y)$. Choosing a single map is the special case in which the weights concentrate on it, so model selection is nested within this averaging. Section~\ref{sec:methodology} formalizes this sequence; Table~\ref{tab:workflow_summary} summarizes the workflow.

We make two methodological contributions. First, we give sufficient conditions under which \sure{} minimization can be used to choose within parameterized classes of shrinkage maps $Y\mapsto f_\gamma(Y)$ that use the full vector of noisy estimates. The resulting oracle inequalities show that the \sure{}-selected map performs nearly as well as the best map in the class, with performance measured by squared-error loss against $\theta$. A close point of comparison is \citet{kwonOptimalShrinkageEstimation2025}, who studies best-in-class shrinkage for panel fixed effects. That setting uses repeated observations over time and focuses on affine shrinkage rules. By contrast, the setting here is cross-sectional: one unit's estimate may enter another unit's fitted value through nonlinear rules based on geographic distance, spatial adjacency, or similarity in the observed estimates themselves. This within-class selection result also relates to work that uses estimated loss, \sure{}, or cross-validation to tune regularized many-parameter estimators \citep{abadieRiskMachineLearning2019, vives-i-bastidaSTRETCHINGNETMULTIDIMENSIONAL2023, adusumilliCrossValidationSURE2026}.

Second, we give a \sure{}-based model-averaging step for a finite library of trained shrinkage maps. This result combines the trained maps directly, taking each as a fixed building block. Each trained map must satisfy a regularity condition, checked separately for each candidate (Section~\ref{sec:theory}). The oracle comparison then applies whether the maps come from closed-form formulas or iterative optimization. Given a finite library of estimators, the procedure chooses convex weights by minimizing \sure{} for the weighted average, with the weights treated as fixed inside the criterion (Section~\ref{sec:model_averaging}). The oracle guarantee is stated for fixed weights. When the \sure{}-weighted average is reported, its \sure{} value is a separate evaluation of the map $Y\mapsto f_{\hat w(Y)}(Y)$. That evaluation accounts for the data-dependence of the selected weights and trained parameters. A close model-averaging comparison is \citet{hansenLeastSquaresModel2007}, who studies weights over linear least-squares fits. Here the weights range instead over trained nonlinear maps, including shrinkage rules whose tuning parameters are estimated from the data before averaging.

The empirical application uses Opportunity Atlas data to estimate tract-level economic mobility across \oaNCzs{} commuting zones. The tract-level estimates show strong spatial patterns: nearby tracts often have similar estimated mobility, but geography, adjacency, school-district boundaries, and historical segregation can make different forms of relatedness empirically relevant. In the main neighborhood mobility comparison, candidate rules differ by distance metric and incorporation of demographic covariates, and the empirical analysis reports the \sure{}-weighted average of candidate maps as the primary estimator. A separate Cook County comparison adds the value-similarity rule of Section~\ref{sec:estimators}, which lets large differences in the observed estimates reduce smoothing between nearby tracts. In this setting, the \sure{}-weighted average of candidate maps reduces \sure{}-estimated mean squared error (MSE) by about \oaReductionAggMle\% relative to the raw maximum-likelihood benchmark (\mle{}, the unshrunk tract estimates $Y$) and by about \oaReductionAggVsClose\% relative to \closegauss{}, the closed-form non-spatial EB benchmark---the Gaussian member of the conditional location-scale (CLOSE) family of \citet{chenEmpiricalBayesWhen2024}. The best individual spatial shrinkage rule varies across commuting zones. The broader lesson is that the relevant notion of relatedness is an empirical choice, and we recommend making that choice by \sure{}-weighted averaging rather than by committing to one form of spatial smoothing in advance.

The remainder of the paper proceeds as follows. Section~\ref{sec:methodology} sets up the Gaussian compound-decision problem for shrinkage estimation, introduces \sure{} as an observable estimate of squared-error loss, gives examples of candidate shrinkage maps, and then shows how \sure{} can be used to choose weighted averages of trained candidates. Section~\ref{sec:theory} states the two oracle inequalities, one for \sure{} selection within a class and one for \sure{}-weighted averaging across a library. Section~\ref{sec:empirical} applies the framework to Opportunity Atlas mobility estimates, first comparing estimators on \sure{}-estimated MSE for the latent mobility vector and then asking how the same shrinkage estimates affect a targeting exercise that selects high-mobility tracts. Section~\ref{sec:conclusion} concludes.

\section{Methodology}\label{sec:methodology}

\subsection{The Estimation Problem}

The researcher observes an $n$-dimensional vector $Y = \theta + \varepsilon$, where $\theta \in \bbR^n$ is a fixed but unknown parameter vector and $\varepsilon \sim \N(0, \Sigma)$ is Gaussian noise with known covariance matrix $\Sigma$. The Gaussian sampling model is a working approximation for microdata-derived estimates with reported precision, as in related empirical Bayes applications \citep[e.g.,][]{chenEmpiricalBayesWhen2024}. In such applications, each component $Y_i$ is itself an average or regression coefficient estimated from an underlying micro-sample, and the known covariance matrix $\Sigma$ reflects the sampling precision of those estimates. In the Opportunity Atlas application, $\Sigma$ is taken to be the diagonal matrix of reported marginal variances; Appendix~\ref{app:noise_covariance_misspecification} records what changes when \sure{} is computed with an approximate covariance matrix.

This is a compound decision problem in the sense of \citet{robbins1951asymptotically}: $\theta$ is fixed (no prior is placed on it), the researcher chooses a decision rule $f$ that returns a vector of actions $f(Y)$ with one action $f_i(Y)$ per unit, each allowed to depend on the entire vector $Y$, and the rule is judged by its average squared-error loss against the fixed $\theta$. Each such rule is a map $f\colon\bbR^n\to\bbR^n$, and we compare rules within a candidate class $\mathcal F$ by the realized loss
\[
L_n(f)=\frac{1}{n}\|f(Y)-\theta\|_2^2.
\]
The infeasible oracle in this class is
\[
f^* \in \argmin_{f\in\mathcal F} L_n(f).
\]
This oracle uses the unknown vector $\theta$ and is therefore only a benchmark. The statistical problem is to use the observed vector $Y$ to select, from the rules $f\in\mathcal F$, an estimate whose realized loss is close to this oracle benchmark.

Throughout this paper, expectations are over the sampling noise in $Y=\theta+\varepsilon$, treating $\theta$ as fixed. We reserve the term \emph{risk} for the expected realized loss, $R_n(f):=\E[L_n(f)]$. Because $\theta$ is unknown, neither $L_n(f)$ nor $R_n(f)$ can be evaluated directly, so effective estimation requires an observable criterion whose behavior tracks the unobserved loss. Appendix~\ref{app:mse_decision_relevance} gives a simple condition under which lower squared-error estimation error also reduces errors in downstream comparisons based on the estimated vector.

As a concrete example, each $Y_i$ is a tract-level estimate of economic mobility from the Opportunity Atlas \citep{chettyOpportunityAtlasMapping2018}, with known sampling variance $\Sigma_{ii} = \sigma_i^2$ reflecting the precision of the underlying microdata. The parameter vector $\theta \in \bbR^n$ represents true neighborhood-level mobility across hundreds to thousands of Census tracts in a commuting zone. Reported standard errors vary substantially across tracts because the underlying sample sizes differ. Because the units are neighborhoods, spatial structure is a natural source of pooling information.  The application therefore motivates shrinkage rules that can leverage geographic relationships rather than treating all tracts as exchangeable.

\subsection{Stein's Unbiased Risk Estimate (\sure{})}\label{sec:sure}

For any continuously differentiable estimator $f$ such that $\E[\|f(Y)-Y\|_2^2]<\infty$ and $\E[\sum_{i,j}|\Sigma_{ij}\partial_j f_i(Y)|]<\infty$, Stein's lemma gives an observable statistic $\sure_n(f)$ satisfying $\E[\sure_n(f)] = R_n(f)=\E[L_n(f)]$:\footnote{Continuous differentiability is used here as a convenient sufficient condition.  Unbiasedness of \sure{} also extends to weakly differentiable maps $Y\mapsto f(Y)$ satisfying the same integrability conditions. For example, $x\mapsto |x|$ has weak derivative $\operatorname{sign}(x)$.  The sufficient conditions used below are stated as continuous-differentiability conditions.}
\[
\sure_n(f) := \underbrace{\frac{1}{n}\|Y - f(Y)\|_2^2 - \frac{1}{n}\tr(\Sigma)}_{\text{noise-corrected in-sample MSE}} + \underbrace{\frac{2}{n}\tr\{\Sigma Df(Y)\}}_{\text{complexity correction}}.
\]
Here $Df(Y) = [\partial f_i(Y)/\partial Y_j]_{ij}$ is the $n \times n$ Jacobian matrix of $f$, and $\tr(\Sigma)=\sum_{i=1}^n \Sigma_{ii}$ is the trace of the noise covariance matrix, the total sampling variance in $Y$. The first term subtracts the irreducible noise variance $\tr(\Sigma)$ from the in-sample prediction error, converting it into an estimate of the estimation error $\|f(Y) - \theta\|_2^2$ rather than the prediction error $\|f(Y) - Y^{\mathrm{new}}\|_2^2$, where $Y^{\mathrm{new}}$ is an independent draw from the same model. But this estimate is generally biased downward when the rule uses the same $Y$ to form the reported values being evaluated.\footnote{Cross-fitting eliminates the downward bias issue by making the decision rule separable across sample splits \citep[e.g.,][]{ignatiadisCovariatePoweredEmpiricalBayes2021, chenCompoundSelectionDecisions2025}. \sure{} instead accounts for the dependence directly through the complexity correction, using the full sample without splitting.} The complexity correction measures the sensitivity of $f$ to the data---a penalty that is larger for more flexible estimators and corrects for this optimism.

For a linear smoother $f_S(Y)=SY$ with a data-independent matrix $S$, the same formula reduces to
\[
\sure_n(f_S)
=
\frac{1}{n}\|(I-S)Y\|_2^2
-\frac{1}{n}\tr(\Sigma)
+\frac{2}{n}\tr(\Sigma S).
\]
Once $S$ is fixed, every term on the right-hand side is computable from $Y$, $\Sigma$, and $S$.

The statistic $\sure_n(f)$ is the observable risk criterion used below to evaluate shrinkage estimators. It can also serve as a training criterion within a candidate class. Section~\ref{sec:theory} gives conditions under which $\sure_n(f)$ tracks $L_n(f)$ closely enough to justify those choices.

\subsection{Examples of Shrinkage Estimators}\label{sec:estimators}

Write $f_i(Y)$ for the $i$th coordinate of a shrinkage rule $f$, the value reported for unit $i$. The examples below supply the paper's candidate classes in the empirical application of Section~\ref{sec:empirical}. For a parameter space $\Gamma\subset \bbR^d$, each example represents a parameterized class of maps $\{f_\gamma:\gamma\in\Gamma\}$ encoding one notion of relatedness, and they are ordered in this section by the degree to which the full vector $Y$ enters each component rule $f_i(Y)$. In the normal--normal empirical Bayes (\nneb{}) example, $f_i(Y)$ depends on $Y$ only through unit $i$'s own estimate and indirectly through trained scalars shared by all units. The Gaussian-process (\gp{}) examples let other units' estimates enter $f_i(Y)$ directly, through weighted combinations $f_i(Y)=\sum_j s_{ij}Y_j$ whose smoothing weights $s_{ij}$ are determined by spatial relationships such as geographic distance. The value-similarity example lets the smoothing weights depend on the observed estimates themselves, $s_{ij}=s_{ij}(Y)$: among nearby units, those with similar values receive more weight. Each example is motivated by a working model for the latent vector $\theta$, but the working model serves only to construct the class $\{f_\gamma: \gamma\in \Gamma\}$. The \sure{} comparisons in the paper evaluate the resulting maps by squared-error loss for the fixed vector $\theta$, whether or not the working model is properly specified. When the tuning parameters of a class are themselves trained on $Y$, \sure{} must account for that dependence; Section~\ref{sec:learned_params} gives the correction.

\begin{example}[Normal--normal empirical Bayes (\nneb{})]\label{ex:normal_normal_eb}
The first candidate estimates a two-parameter prior for the latent parameters and reports the associated posterior means. The rule proceeds as if the parameters were exchangeable---as if relabeling units changed nothing, $(\theta_{\pi(1)},\ldots,\theta_{\pi(n)}) \overset{d}{=} (\theta_1,\ldots,\theta_n)$ for every permutation $\pi$ of $\{1,\ldots,n\}$---and adopts the normal working prior $\theta_i \stackrel{\mathrm{iid}}{\sim} \N(\mu, \tau^2)$ \citep{robbinsEmpiricalBayesApproach1956, efronSteinParadoxStatistics1977, morrisParametricEmpiricalBayes1983, xieSUREEstimatesHeteroscedastic2012a}. For fixed $\gamma=(\mu,\tau^2)$, the posterior mean shrinks each estimate toward $\mu$, more strongly when the sampling variance $\sigma_i^2$ is large relative to $\tau^2$:
\[
f_{\gamma,i}(Y)
=
\frac{\tau^2}{\sigma_i^2 + \tau^2}\, Y_i
+
\frac{\sigma_i^2}{\sigma_i^2 + \tau^2}\,\mu .
\]
The trained version replaces $\gamma$ by $\hat\gamma=(\hat\mu,\hat\tau^2)$, where $\hat\mu$ is the estimated global mean and $\hat\tau^2$ is the estimated prior variance. With these trained scalar parameters held fixed, the reported value for unit $i$ depends on its own estimate $Y_i$ and sampling variance $\sigma_i^2$; it does not use geography or adjacency to decide which other estimates enter unit $i$'s reported value.
\end{example}

One way to relax this working assumption keeps independence across units but drops identical distribution: conditioning on unit covariates $X_i$---which may include the standard error $\sigma_i$---replaces the two scalars with functions, so the working prior becomes $\theta_i \mid X_i \sim \N(\mu(X_i), \tau^2(X_i))$ and the trained rule shrinks $Y_i$ toward $\hat\mu(X_i)$ by an amount governed by $\hat\tau^2(X_i)$ and $\sigma_i^2$ \citep{ignatiadisCovariatePoweredEmpiricalBayes2021, chenEmpiricalBayesWhen2024}. Units with the same covariates are still treated symmetrically. Nothing in the rule links unit $i$ to specific other units.

\begin{example}[Gaussian-process shrinkage]\label{ex:gp_shrinkage}
\gp{} shrinkage starts from a covariance specification for the latent vector $\theta$ and uses the resulting posterior-mean formula as the shrinkage map. Let $K_\gamma\in\bbR^{n\times n}$ be a positive semidefinite covariance matrix whose $(i,j)$ entry records the covariance assigned to $\theta_i$ and $\theta_j$. Under a spatial specification, this entry is larger for units that are close under the chosen measure of distance, as in standard spatial covariance models \citep{steinInterpolationSpatialData1999, rasmussenGaussianProcessesMachine2006}. The parameter vector $\gamma$ indexes the covariance specification used to construct $K_\gamma$. The Gaussian prior specification $\theta \sim \N(0,K_\gamma)$ would deliver posterior mean $K_\gamma(K_\gamma+\Sigma)^{-1}Y$.\footnote{A prior mean can be included without changing the role of $K_\gamma$. If $\theta\sim\N(\mu,K_\gamma)$ for $\mu\in\bbR^n$, the posterior mean is $\mu+K_\gamma(K_\gamma+\Sigma)^{-1}(Y-\mu)$. The zero-mean display keeps the notation focused on the covariance structure.} We use this posterior-mean formula as a class of shrinkage maps over $\gamma\in \Gamma$, while continuing to treat $\theta$ as a fixed unknown vector:
\[
f_\gamma(Y) = S_\gamma\, Y, \qquad S_\gamma := K_\gamma(K_\gamma + \Sigma)^{-1}.
\]
For fixed $\gamma$, the map $f_\gamma(Y)=S_\gamma Y$ is the linear-smoother case above with $S=S_\gamma$.\footnote{The difference from \citet{kwonOptimalShrinkageEstimation2025} is which dimension the covariance matrix indexes. In Kwon's panel fixed-effect setting, the relevant covariance matrix is indexed by time periods within a unit, so the resulting smoother combines that unit's time-specific estimates. In this paper, $K_\gamma\in\bbR^{n\times n}$ is indexed by cross-sectional units, so $S_\gamma=K_\gamma(K_\gamma+\Sigma)^{-1}$ lets the reported value for one tract depend on estimates from other tracts through geography or adjacency.} The smoothing matrix $S_\gamma$ implements shrinkage by balancing cross-unit similarity against sampling precision: noisier units are moved more toward estimates from similar units, while precisely estimated units retain more of their own observation.

The covariance matrix $K_\gamma$ determines the smoothing matrix $S_\gamma$, and hence which other estimates enter each reported value. One common spatial covariance form, with $\gamma=(\sigma_{\mathrm{sp}}^2,\sigma_{\mathrm{nug}}^2,\ell)$, is
\[
K_{ij}(\gamma) := \sigma_{\mathrm{sp}}^2\, k(d_{ij};\ell) + \sigma_{\mathrm{nug}}^2\, \mathbf{1}\{i = j\},
\]
where $d_{ij}$ is the distance between units $i$ and $j$, $\sigma_{\mathrm{sp}}^2$ sets the variance scale of the shared spatial component, $\sigma_{\mathrm{nug}}^2$ is a nugget variance for idiosyncratic variation not explained by the spatial structure, and $\ell$ controls how quickly covariance decays with distance. The length scale $\ell$ sets the radius of effective pooling: small $\ell$ makes $K_\gamma$ nearly diagonal, so each estimate is shrunk on its own without local pooling; large $\ell$ pools over ever-wider neighborhoods, approaching a single common component shared by all units. Like the variance parameters, $\ell$ is not fixed in advance but trained on the data (Section~\ref{sec:learned_params}). The distance metric itself is a modeling choice. It might be geographic distance, road-network distance, or shortest-path distance on an adjacency graph. The Opportunity Atlas application uses the exponential kernel $k(d;\ell)=\exp(-d/\ell)$ with either geographic distance or contiguity distance.\footnote{This is the Mat\'ern-$\tfrac12$ correlation. More general Mat\'ern kernels add a smoothness parameter, but the empirical application fixes that parameter at $1/2$.} For any fixed $\gamma$, $K_{ij}$ depends only on relationships among units, not on their observed outcomes, so $f_\gamma$ is linear in $Y$.
\end{example}

Figure~\ref{fig:shrinkage_choropleth} illustrates the difference between global and spatial shrinkage on mobility estimates for Cook County tracts selected from the Chicago commuting zone. The raw estimates (panel~A) are visibly noisy. \nneb{} (panel~B) shrinks every tract toward the same global target $\bar Y$, the amount depending only on the tract's own precision: estimates above $\bar Y$ are pulled down and estimates below are pulled up, regardless of the values of neighboring tracts. The result both over-smooths and under-smooths: tracts in different community areas---Chicago's named groupings of Census tracts---are blurred toward one another, while a tract surrounded by similar neighbors is still dragged away from its local average toward the distant global mean $\bar Y$. The spatial \gp{} (panel~C) instead shrinks each tract toward a local neighborhood average, so the smoothed map retains more of the spatial pattern in the raw estimates while reducing tract-level noise.

\begin{figure}[t]
  \centering
  \includegraphics[width=\textwidth]{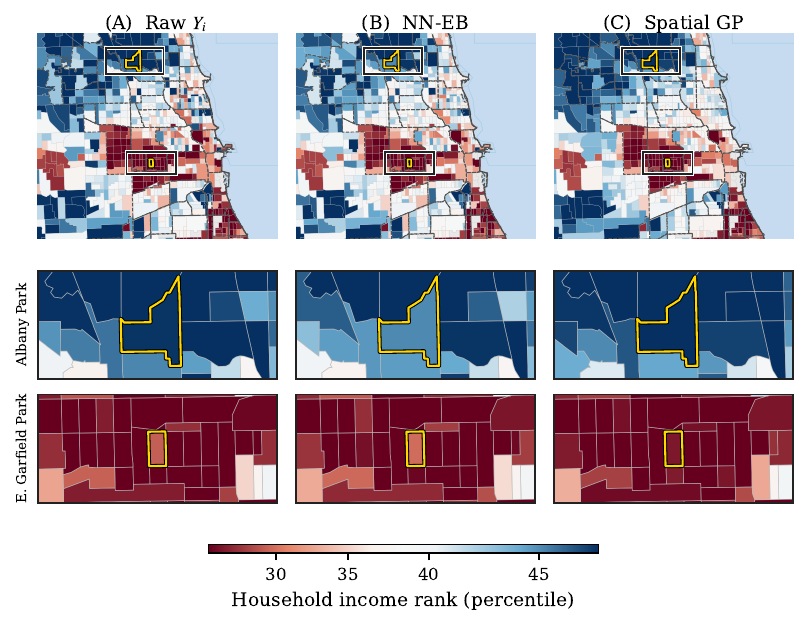}
  \caption{Choropleth maps of tract-level mobility estimates for Cook County tracts selected from the Chicago commuting zone. (A)~Raw \mle{}s. (B)~\nneb{} predictions, which shrink every tract toward the global mean $\bar Y$. (C)~Spatial \gp{} predictions, which smooth locally while retaining local geographic variation. Dashed lines mark community areas, official groupings of contiguous tracts. The bottom rows magnify the boxed regions: \nneb{} pulls a tract in Albany Park away from its uniformly high-mobility neighbors and a tract in East Garfield Park away from its uniformly low-mobility neighbors, while the spatial \gp{} holds each tract near its neighborhood. The same two tracts are highlighted in Figure~\ref{fig:target_scatter}. This figure and Figure~\ref{fig:target_scatter} use the pooled-male mobility outcome, whose tract estimates are noisier than the pooled outcome's.  The noisier estimates make the differences between the shrinkage rules easier to see. The empirical comparison in Section~\ref{sec:empirical} focuses on the pooled outcome, with additional outcomes in Appendix~\ref{app:oa_robustness}.}
  \label{fig:shrinkage_choropleth}
\end{figure}

Figure~\ref{fig:target_scatter} quantifies the competing shrinkage targets. Each tract's raw estimate $Y_i$ is plotted against its leave-one-out spatial \gp{} shrinkage target $\mu_i$, defined by writing the \gp{} prediction for tract $i$ as $f_i(Y)=S_{ii}Y_i+(1-S_{ii})\mu_i$. The spatial target $\mu_i$ plays the role that the global mean $\mu$ played in \nneb{}. The figure overlays the two targets: \nneb{} shrinks every tract toward $\bar{Y}$ (orange line), regardless of spatial context, while the spatial \gp{} shrinks toward $\mu_i$ (blue diagonal). Whenever a tract's raw estimate lies between the two targets---the shaded region of Figure~\ref{fig:target_scatter}, containing \oaWedgeSharePct\% of the \oaScatterNTracts{} tracts---the two procedures pull $Y_i$ in opposite directions. The tracts highlighted in Figure~\ref{fig:shrinkage_choropleth} sit on opposite lobes of this region: the spatial \gp{} pulls the tract in Albany Park up toward its high-mobility neighborhood while \nneb{} pulls it down toward $\bar Y$, and the mirror pattern holds for the tract in East Garfield Park.

\begin{figure}[t]
  \centering
  \includegraphics[width=0.7\textwidth]{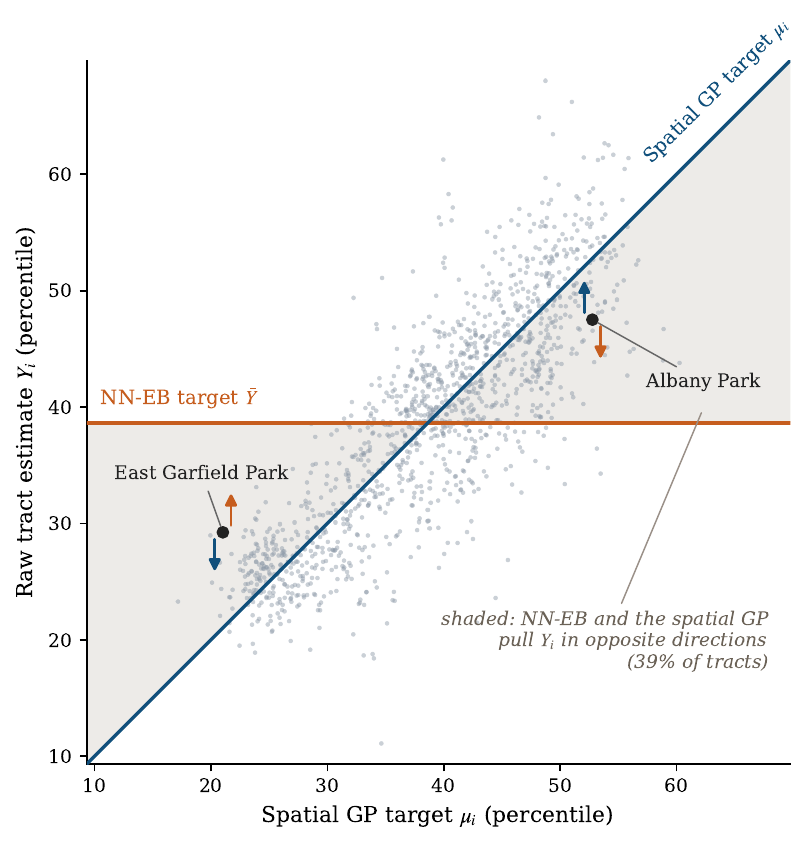}
  \caption{Competing shrinkage targets for the same Cook County tract-level mobility estimates from the Chicago commuting zone. Each tract's \mle{} $Y_i$ is plotted against its leave-one-out spatial \gp{} shrinkage target $\mu_i = (f_i(Y) - S_{ii} Y_i)/(1 - S_{ii})$, where $S = K_\gamma(K_\gamma + \Sigma)^{-1}$. The blue diagonal shows the spatial \gp{} target ($Y_i = \mu_i$); the orange horizontal line shows the \nneb{} target ($\bar{Y}$). Axes are in percentile units. The shaded wedge between the two target lines contains the tracts whose raw estimate lies between the targets, so \nneb{} and the spatial \gp{} pull $Y_i$ in opposite directions (\oaWedgeSharePct\% of tracts). Albany Park and East Garfield Park mark the two highlighted tracts from Figure~\ref{fig:shrinkage_choropleth}. Paired fixed-length arrows at each show the direction---not the magnitude---in which each procedure moves the raw estimate, in the color of the corresponding target line.}
  \label{fig:target_scatter}
\end{figure}

The \gp{} shrinkage maps above have the form $SY$ once the covariance matrix is fixed: for a given tract, the estimates that enter its reported value are determined by distance, not by the realized values. The next example keeps the spatial smoothing structure but lets differences in the observed estimates reduce the geography-based similarity between nearby tracts. The empirical motivation is that nearby tracts can be separated by highways, rivers, school-district boundaries, or boundaries associated with historical segregation.  In such cases, smoothing over geographic distance may average across places whose observed mobility estimates differ sharply. The construction is the bilateral filter, originally developed as an edge-preserving smoother in image processing \citep{tomasiBilateralFilteringGray1998}, adapted to the Gaussian shrinkage form above.

\begin{example}[Value-similarity shrinkage]\label{ex:bilateral_shrinkage}
A value-similarity rule starts from a geography-based covariance matrix $K^{\mathrm{geo}}_\gamma$ and reduces the entry for nearby tracts whose observed estimates are far apart:
\[
K_{ij}(Y):=K^{\mathrm{geo}}_{\gamma,ij}\exp\{-\lambda(Y_i-Y_j)^2\}.
\]
The exponential factor is close to one when the observed estimates $Y_i$ and $Y_j$ are similar and close to zero when they are far apart in value. Here $\lambda\ge 0$ is a tuning parameter that joins $\gamma$. At $\lambda=0$ the rule is the geography-only smoother of Example~\ref{ex:gp_shrinkage}, and larger $\lambda$ suppresses smoothing across large differences in observed values. The rule therefore smooths locally in geography while allowing large differences in the observed estimates to reduce cross-tract smoothing.
The corresponding shrinkage map has the same algebraic form as the \gp{} smoother, but now with a covariance matrix that depends on $Y$:
\[
f(Y)=S(Y)Y,\qquad S(Y):=K(Y)\{K(Y)+\Sigma\}^{-1}.
\]
Because the smoothing matrix now changes with $Y$, the shrinkage map is nonlinear in the data. The \sure{} complexity correction must therefore differentiate the whole map $Y\mapsto S(Y)Y$ using the product rule: it includes the fixed-smoother trace term $\tr\{\Sigma S(Y)\}$ plus additional terms from how $S(Y)$ changes with the observed estimates.
\end{example}

Figure~\ref{fig:bilateral_example} illustrates the difference between geography-only and value-similarity shrinkage targets for a set of central-Chicago tracts in Cook County. The figure displays the leave-one-out \emph{shrinkage targets} $\mu_i = (f_i(Y) - S_{ii} Y_i) / (1 - S_{ii})$. Panel~(A) shows the noisy estimates, including visible local contrasts across some community-area borders. Panel~(B) uses a geographic smoother, so nearby tracts enter the formula according to distance.  This attenuates some of those local contrasts. Panel~(C) also uses similarity in the observed estimates, so pairs of tracts with dissimilar values receive less weight and more of the visible contrast is retained. Compare panels~(B) and~(C) at communities \textcircled{1} and \textcircled{3}: geographic smoothing alone bleeds North Lawndale's low estimates across the boundary into the Lower West Side, while the value-similarity target preserves the contrast. Both smoothed panels use the same preliminary covariate adjustment (OLS residualization on demographic covariates), held fixed across panels, so the comparison isolates geography-only versus value-similarity smoothing.

\begin{figure}[!t]
    \centering
    \includegraphics[width=\textwidth]{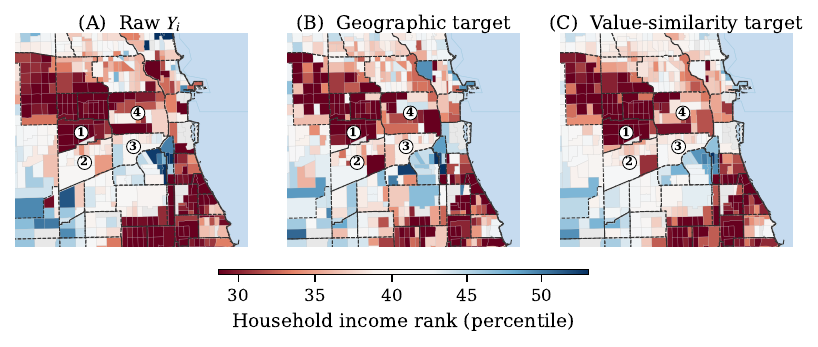}
    \caption{Leave-one-out shrinkage targets for geography-only and value-similarity smoothing: central-Chicago tracts in Cook County, economic mobility. (A)~Raw tract-level estimates $Y_i$. (B)~Shrinkage targets using geographic distance only. (C)~Shrinkage targets that also downweight tracts with dissimilar observed estimates, retaining more of the visible local contrast across nearby communities. Both smoothed panels use the same preliminary covariate adjustment, held fixed across panels. The target $\mu_i$ for tract~$i$ is the affine target satisfying $f_i(Y)=S_{ii}Y_i+(1-S_{ii})\mu_i$. Dashed lines mark community area boundaries. Numbered communities: \textcircled{1}~North Lawndale, \textcircled{2}~South Lawndale, \textcircled{3}~Lower West Side, \textcircled{4}~Near West Side.}
    \label{fig:bilateral_example}
\end{figure}

Each way of defining which estimates are related---through geographic distance, tract adjacency, observed-value similarity, or their combinations---yields a distinct candidate class with its own tuning parameters $\gamma$. A practitioner then faces two nested choices: \emph{within} each candidate class, how should $\gamma$ be trained? And when averaging \emph{across} classes, what weight should each trained candidate receive? The across-candidate comparison is posed in terms of \sure{}, while within-candidate training may use \sure{} or another criterion chosen by the researcher.  Either way, the resulting trained map is later evaluated by \sure{}. The value-similarity example shows one source of extra $Y$-dependence, through the matrix $S(Y)$. The next subsection turns to another: tuning parameters that are trained on $Y$.

\subsection{\sure{} for Trained Parameters}\label{sec:learned_params}

In practice, using a candidate class $\{f_\gamma:\gamma\in\Gamma\}$ requires a training rule chosen by the researcher. This rule maps the observed vector $Y$ to parameter values $\hat\gamma(Y)$. The tuning parameter $\gamma$ may be a length scale, a variance component, a regularization strength, or a value-similarity parameter. Write
\[
F(Y)=f_{\hat\gamma(Y)}(Y)
\]
for the trained map. For risk evaluation, \sure{} is applied to the full map $F$ using the chain rule, not to $f_\gamma$ with the realized value $\hat\gamma(Y)$ plugged in and treated as fixed. With output coordinates as rows, the chain rule gives
\[
DF(Y)
=
\underbrace{
D_y f_\gamma(Y)\big|_{\gamma=\hat\gamma(Y)}
}_{\text{direct sensitivity, holding $\gamma$ fixed}}
+
\underbrace{
D_\gamma f_\gamma(Y)\big|_{\gamma=\hat\gamma(Y)}D_Y\hat\gamma(Y)
}_{\text{sensitivity from training on $Y$}}.
\]
Therefore \sure{} for the trained map equals the fixed-parameter \sure{} formula evaluated at the realized $\hat\gamma(Y)$, plus an additional training correction:
\[
\begin{aligned}
\sure_n(F)
&=
\underbrace{
\frac{1}{n}\|Y-f_{\hat\gamma(Y)}(Y)\|_2^2
-\frac{1}{n}\tr(\Sigma)
+\frac{2}{n}
\tr\!\left[
\Sigma D_y f_\gamma(Y)\big|_{\gamma=\hat\gamma(Y)}
\right]
}_{\text{fixed-parameter \sure{} at the realized $\hat\gamma(Y)$}}
\\
&\quad+
\underbrace{
\frac{2}{n}
\tr\!\left[
\Sigma
D_\gamma f_\gamma(Y)\big|_{\gamma=\hat\gamma(Y)}
D_Y\hat\gamma(Y)
\right]
}_{\text{training correction}}.
\end{aligned}
\]
We call the first brace---the \sure{} formula with $\gamma$ treated as fixed---\emph{proxy} \sure{}. The second brace is the training correction.

The training rule can be a closed-form estimator or an iterative algorithm. Method-of-moments estimates, maximum-likelihood estimates, and fixed iterative optimization routines all produce a map $Y\mapsto\hat\gamma(Y)$.\footnote{Iterative optimization includes standard stochastic-gradient methods. The Opportunity Atlas implementation uses AdamW, a decoupled-weight-decay variant of Adam, for the trainable \gp{} candidates \citep{kingmaAdamMethodStochastic2017,loshchilovDecoupledWeightDecay2019}.} The training choice is part of candidate construction. For \sure{} evaluation, the relevant object is the resulting trained map $Y\mapsto f_{\hat\gamma(Y)}(Y)$, including the sensitivity of $\hat\gamma(Y)$ to the same data.

For a \gp{} candidate class indexed by covariance parameters $\gamma$, write the fixed-parameter zero-mean shrinkage map as $f_\gamma(Y)=K_\gamma(K_\gamma+\Sigma)^{-1}Y$. The conventional \gp{} training rule is maximum marginal likelihood. Given a working covariance family $\{K_\gamma:\gamma\in\Gamma\}$, this training rule chooses $\hat\gamma_{\mathrm{ML}}(Y)$ by maximizing the Gaussian marginal likelihood for $Y$ with covariance $K_\gamma+\Sigma$ \citep{rasmussenGaussianProcessesMachine2006}.\footnote{Under the auxiliary \gp{} prior specification $\theta\sim\N(0,K_\gamma)$ and $Y\mid\theta\sim\N(\theta,\Sigma)$, integrating out $\theta$ gives the marginal likelihood $Y\sim\N(0,K_\gamma+\Sigma)$. If this covariance specification is correct for some $\gamma_0\in\Gamma$, the posterior mean based on $K_{\gamma_0}$ is the squared-error Bayes rule, so likelihood-based covariance training and squared-error prediction are aligned. If no such $\gamma_0$ exists, the working covariance family is misspecified, and maximum marginal likelihood targets the $\gamma$ minimizing Kullback--Leibler divergence within $\{K_\gamma:\gamma\in\Gamma\}$, which need not minimize squared-error risk of the induced shrinkage map. \citet{bachocCrossValidationMaximum2013,bachocAsymptoticAnalysisCovariance2018} show this target mismatch for \gp{} covariance estimation under misspecification: maximum likelihood targets Kullback--Leibler divergence, whereas cross-validation targets prediction mean squared error. \sure{} plays the corresponding role here by estimating the mean squared error of the shrinkage rule.} A training rule need not minimize \sure{} for the resulting trained map to be evaluated using \sure{}.  A likelihood-trained \gp{} map could therefore be included as one trained candidate in the finite library considered for averaging in Section~\ref{sec:model_averaging}. In the empirical application, the trainable \gp{} candidates in Table~\ref{tab:oa_candidates} are trained by minimizing proxy \sure{} over $\gamma$.

In practice, the researcher does not need to derive $D_Y\hat\gamma(Y)$ by hand. Automatic differentiation can propagate sensitivities through an implemented training algorithm. When $\hat\gamma(Y)$ is characterized by first-order conditions, implicit-differentiation tools use those conditions to obtain the same sensitivity without deriving a new formula for each training problem \citep{blondel_efficient_2022}. The trace terms in \sure{} can then be computed efficiently using randomized trace estimation \citep{hutchinson_stochastic_1990,nobel_tractable_2023}.\footnote{If $v$ is a random vector with $\E[vv^\top]=\Sigma$, then $\E[v^\top DF(Y)v]=\tr\{\Sigma DF(Y)\}$. A randomized trace estimate averages $v^\top DF(Y)v$ over several independent draws of $v$. Each draw requires the Jacobian-vector product $DF(Y)v$, which automatic differentiation computes by propagating one direction through the chain rule. This avoids constructing $DF(Y)$ explicitly and avoids the matrix-matrix products that would arise from carrying the full Jacobian through the training rule.}

The chain-rule decomposition also clarifies how the theory treats trained candidates.  Once the training rule is fixed,
the composite map $Y\mapsto f_{\hat\gamma(Y)}(Y)$ is the estimator evaluated by \sure{}, and Appendix~\ref{app:regularity_verification} gives sufficient conditions under which this composite map satisfies the regularity condition used for averaging.

\subsection{\sure{} Model Averaging}\label{sec:model_averaging}

Suppose the preceding steps produce a finite library of $K$ candidate estimators, $f_1,\ldots,f_K$. The candidates may differ in the information they use, their preprocessing choices, or their training rules. Rather than committing in advance to one candidate class, \sure{} model averaging uses \sure{} itself to choose convex weights across the trained maps. When candidate rules make different errors, a convex combination can beat every single candidate. In the Opportunity Atlas application, the best single candidate differs across commuting zones. Averaged across them, however, the \sure{}-weighted average has lower estimated risk than any single candidate (Section~\ref{sec:empirical}). The oracle comparison below is therefore against the best fixed convex combination---a stronger benchmark than the best single candidate, since every single candidate is a vertex of the simplex.

Specifically, for weights in the simplex we form
\[
f_w(Y) = \sum_{k=1}^K w_k\, f_k(Y), \qquad w \in \Delta^{K-1} := \{w \in \bbR^K_+ : \textstyle\sum_k w_k = 1\},
\]
and choose weights by minimizing the fixed-weight \sure{} criterion:
\[
\hat{w} \in \argmin_{w \in \Delta^{K-1}} \sure_n(f_w).
\]
Within this minimization problem, each proposed weight vector $w$ is treated as fixed. This fixed-weight criterion is the object used for the oracle comparison in Section~\ref{sec:averaging_guarantee}. After the observed data select $\hat w(Y)$, the reported map is the \sure{}-weighted average,
\[
\tilde{f}(Y) := f_{\hat{w}(Y)}(Y) = \sum_{k=1}^K \hat{w}_k(Y)\, f_k(Y).
\]
As a function of the weights, $f_w$ is linear in $w$, and for each fixed $w$ its Jacobian satisfies $Df_w=\sum_k w_k Df_k$. Therefore the Jacobian trace $\tr(\Sigma\,Df_w) = \sum_k w_k \tr(\Sigma\,Df_k)$ is also linear in $w$, so the fixed-weight objective $\sure_n(f_w)$ is quadratic in $w$.\footnote{The simplex-constrained quadratic program can be solved with standard convex optimization tools. If \sure{} is reported for the final data-chosen average, differentiating the selected weights can be handled by differentiating the optimization conditions when the solution is locally stable.}

Before averaging over candidate estimators, each candidate's \sure{} value must be computed for the trained map actually produced from $Y$. When tuning parameters are trained on $Y$, evaluating \sure{} as if those parameters were fixed can understate risk.  This is the excess-optimism problem for \sure{}-tuned estimators studied by \citet{tibshiraniExcessOptimismBiased2019}. The same concern applies to weights chosen by minimizing \sure{}, so the reported \sure{} value for the \sure{}-weighted average evaluates the full map $\tilde f$, rather than the fixed-weight criterion used to choose $\hat w$.

\subsection{Workflow Summary}\label{sec:workflow_summary}

\begingroup
\newcommand{\workflowcell}[2]{\begin{minipage}[t]{\linewidth}\raggedright #1\\[-0.15em]#2\end{minipage}}
\begin{table}[!htbp]
\centering
\caption{Workflow for \sure{}-based selection, averaging, and final evaluation. Steps~3--4 compute $\sure_n(f_k)$ for the trained candidates and choose convex weights. Step~5 evaluates $\sure_n(\tilde f)$ for the \sure{}-weighted average $\tilde f$, whose data-dependent weights require their own regularity conditions (Appendix~\ref{app:averaging_adaptive_weights}).}
\label{tab:workflow_summary}
\footnotesize
\setlength{\tabcolsep}{3pt}
\begin{tabular}{@{}>{\raggedright\arraybackslash}p{0.17\textwidth}
                  >{\raggedright\arraybackslash}p{0.27\textwidth}
                  >{\raggedright\arraybackslash}p{0.21\textwidth}
                  >{\raggedright\arraybackslash}p{0.25\textwidth}@{}}
\toprule
Step & Researcher action & Resulting object & Notation \\
\midrule
\mbox{\textbf{1. Specify}} & Choose similarity rules & Candidate classes & \workflowcell{$\mathcal F_k=\{f_{k,\gamma}:\gamma\in\Gamma_k\}$}{$k=1,\ldots,K$} \\
\addlinespace[0.85em]
\mbox{\textbf{2. Train}} & Choose training rule & Trained maps & \workflowcell{$\hat\gamma_k(Y)$}{$f_k(Y)=f_{k,\hat\gamma_k(Y)}(Y)$} \\
\addlinespace[0.85em]
\mbox{\textbf{3. Evaluate}} & Compute \sure{} per candidate & Risk estimates & $\sure_n(f_k)$ \\
\addlinespace[0.85em]
\mbox{\textbf{4. Average}} & Choose convex weights & Averaged map & \workflowcell{$\hat w(Y)\in\Delta^{K-1}$}{$\tilde f(Y)=\sum_k \hat w_k(Y)f_k(Y)$} \\
\addlinespace[0.85em]
\mbox{\textbf{5. Report}} & Compute final-map \sure{} & Final risk estimate & $\sure_n(\tilde f)$ \\
\bottomrule
\end{tabular}
\end{table}
\endgroup

Table~\ref{tab:workflow_summary} summarizes the procedure and separates three uses of \sure{}. In Step~2, \sure{} may be used as a training criterion within a candidate class. In Steps~3 and~4, \sure{} gives a common squared-error criterion for trained candidates and fixed-weight averages, so the researcher can select the smallest-\sure{} candidate or choose convex weights across candidates. In Step~5, \sure{} is used to evaluate the map actually reported, including the sensitivity induced by trained parameters and, when relevant, by selected weights.

The training and averaging steps lead to two theoretical questions. The first question concerns selection within a fixed parameterized class of shrinkage maps $\{f_\gamma:\gamma\in\Gamma\}$: if the researcher chooses $\hat\gamma$ by minimizing $\sure_n(f_\gamma)$ over $\Gamma$, how close is the realized loss $L_n(f_{\hat\gamma})$ to the loss $\min_{\gamma\in\Gamma}L_n(f_\gamma)$ of the best parameter choice in the class? The second concerns averaging after a finite candidate library has been assembled: if the researcher chooses convex weights by minimizing the fixed-weight \sure{} criterion, how close is the resulting loss to the loss from the fixed convex combination with the smallest realized loss?
\section{Theoretical Guarantees}\label{sec:theory}

Throughout this section, the sampling experiment and the estimator classes are implicitly indexed by the dimension $n$. When the $n$-dependence needs to be explicit, we write $Y^{(n)},\theta^{(n)},\varepsilon^{(n)},\Sigma_n$ and $f_{n,\gamma}:\bbR^n\to\bbR^n$. Otherwise, we suppress the $n$-dependence and write $Y,\theta,\varepsilon,\Sigma$, $f_\gamma$, $\mathcal F=\{f_\gamma:\gamma\in\Gamma\}$, and $\Gamma$. Constants described as independent of $n$ are uniform over this sequence.

This section establishes two oracle inequalities for \sure{}-based selection among shrinkage maps. The within-class guarantee, Theorem~\ref{thm:uniform_conc}, applies to a compact parameterized class $\mathcal F=\{f_\gamma:\gamma\in\Gamma\}$ of maps that report vectors of shrinkage estimates. If $\hat\gamma\in\argmin_{\gamma\in\Gamma}\sure_n(f_\gamma)$ and $\gamma^*\in\argmin_{\gamma\in\Gamma}L_n(f_\gamma)$, then the theorem bounds the excess realized loss $L_n(f_{\hat\gamma})-L_n(f_{\gamma^*})$. The averaging guarantee, Proposition~\ref{prop:averaging}, applies after a finite library of trained candidate maps $f_1,\ldots,f_K$ has been assembled. For fixed weights $w\in\Delta^{K-1}$, write $f_w(Y)=\sum_{k=1}^K w_k f_k(Y)$. If $\hat w(Y)\in\argmin_{w\in\Delta^{K-1}}\sure_n(f_w)$ and $w^*\in\argmin_{w\in\Delta^{K-1}}L_n(f_w)$, then the proposition compares the \sure{}-weighted average $\tilde f(Y):=f_{\hat w(Y)}(Y)$ with the fixed-weight convex combination $f_{w^*}$ that has the smallest realized loss over the simplex.

The guarantees do not place a distribution on the latent vector $\theta$. For each dimension $n$, we fix $\theta\in\bbR^n$ and take expectations only over the sampling noise in $Y=\theta+\varepsilon$. Thus $\gamma^*$ and $w^*$ are realized-loss benchmarks for the candidate maps under consideration, not procedures derived from a correctly specified prior on the latent parameters. Although Bayesian or empirical Bayes specifications can motivate maps such as $f_\gamma$, the guarantees below evaluate the resulting maps directly. If every map in $\mathcal F$ or every convex average of $f_1,\ldots,f_K$ has high realized loss for the fixed vector $\theta$, the theory does not remove that approximation error; it controls the additional loss from choosing within the class or library using \sure{}.

Theorem~\ref{thm:uniform_conc} and Proposition~\ref{prop:averaging} impose different requirements because they apply at different points in the construction. Theorem~\ref{thm:uniform_conc} studies exact minimization of \sure{} over the full parameter set $\Gamma$, so its regularity condition is uniform over the parameterized class $\mathcal F$. Proposition~\ref{prop:averaging} begins after the finite candidate maps $f_1,\ldots,f_K$ have already been constructed. Those candidates may include trained parameters and may come from different training rules; Proposition~\ref{prop:averaging} does not require each $f_k$ to solve a within-class \sure{} minimization problem. For averaging, the optimization is over the weights $w$, and the regularity condition is imposed separately on each composite candidate map, including the dependence introduced when a candidate's tuning parameters are trained on the same data.

\subsection{Regularity for Within-Class \sure{} Minimization}\label{sec:regularity}

Theorem~\ref{thm:uniform_conc} studies the rule selected by minimizing $\sure_n(f_\gamma)$ over the compact parameter set $\Gamma$. For this oracle comparison, the observable criterion $\sure_n(f_\gamma)$ must track the realized loss $L_n(f_\gamma)$ uniformly over the candidate class, not only at a fixed map. Existing \sure{} selection guarantees cover finite collections of globally Lipschitz estimators \citep{bellecSecondOrderStein2021}; the result below extends that benchmark in two ways. First, it allows $\Gamma$ to be a compact continuum of tuning parameters, such as length scales and kernel variances. Second, it replaces global Lipschitzness with a broader regularity condition: the adjustment relative to the raw estimate, its derivative, and its variation with $\gamma$ are controlled by polynomial envelopes in the input vector of raw estimates.

The value-similarity rule shows why the broader regularity condition is needed. Changing one observed estimate can change both its own reported value and the weights assigned to other estimates, so the resulting map can fail to be globally Lipschitz even in a two-dimensional fixed-parameter case.\footnote{This issue is not specific to the value-similarity example. \citet{kimLipschitzConstantSelfAttention2021} show that standard dot-product self-attention is not globally Lipschitz in its input. Self-attention is not part of the empirical library in this paper, but it is another example of a flexible data-adaptive map that would not be covered by a theory taking global Lipschitzness as a primitive condition.} Appendix~\ref{sec:verification} gives this non-Lipschitz calculation, and Appendix~\ref{sec:bilateral_verification} verifies a fixed value-similarity building block as a single candidate for averaging. Under a bounded-row-sum condition on the fixed geographic factor of the kernel, Appendix~\ref{sec:bilateral_verification} shows that this fixed value-similarity building block is no more costly, in the regularity order used below, than fixed linear smoothers. Training of tuning parameters raises a related verification issue: regularity of the fixed maps $f_\gamma$ does not by itself establish regularity of the reported map $Y\mapsto f_{\hat\gamma(Y)}(Y)$, because the trained parameters $\hat{\gamma}$ are also functions of $Y$. Appendix~\ref{app:learned_parameter_regularity} gives primitive conditions under which trained maps still satisfy the per-candidate regularity condition used for averaging.

For symmetric matrices, write $A\succ0$ for positive definite and $A\succeq0$ for positive semidefinite. Let $\lambda_{\max}(A)$ denote the largest eigenvalue of a symmetric matrix. Norms $\|\cdot\|_2$, $\|\cdot\|_{\mathrm{op}}$, and $\|\cdot\|_F$ denote Euclidean, operator, and Frobenius norms, respectively. The first assumption restates the Gaussian sampling model of Section~\ref{sec:methodology} along the sequence of experiments and adds two uniform bounds: on the noise scale and on the average magnitude of the latent vector.

\begin{assumption}[Sampling array]\label{asm:sampling_array}
For each dimension $n$,
\[
Y^{(n)}=\theta^{(n)}+\varepsilon^{(n)},\qquad
\theta^{(n)}\in\bbR^n,\qquad
\varepsilon^{(n)}\sim\N(0,\Sigma_n).
\]
The noise covariance $\Sigma_n$ is positive definite\footnote{Positive definiteness is used only to reduce the Gaussian noise to a standard normal vector in the proof. If $\Sigma_n$ is positive semidefinite, the same argument applies after restricting the Gaussian experiment to the support of $\Sigma_n$.} with $\lambda_{\max}(\Sigma_n)\leq\bar\sigma^2$, and the latent vector satisfies $\|\theta^{(n)}\|_2/\sqrt n\leq C_\theta$, for constants $\bar\sigma^2,C_\theta<\infty$ not depending on $n$.
\end{assumption}

The bound on $\theta^{(n)}$ is an average-magnitude condition: $\|\theta^{(n)}\|_2^2/n$ remains bounded, so the fixed latent vectors do not grow in average squared size along the sequence. This condition is implied by putting each coordinate of $\theta^{(n)}$ in a fixed compact set, but it is slightly more general: individual coordinates may exceed any fixed bound as long as their squared magnitudes remain controlled on average.

For a generic input $y\in\bbR^n$, write $g_\gamma(y) := f_\gamma(y) - y$ for the adjustment relative to the raw estimate.\footnote{Under Tweedie's formula \citep{efronTweediesFormulaSelection2011}, the posterior mean of the normal location model satisfies $\E[\theta | Y] = Y + \Sigma \nabla \log p(Y)$, where $p$ is the marginal density of $Y$ and $\nabla\log p(Y)$ is the score of that marginal density. Thus, when a candidate map is motivated by a posterior-mean formula, $g_\gamma$ can be read as an estimate of this score term. Assumption~\ref{asm:regularity} places regularity conditions on the candidate maps themselves; it does not require a correctly specified marginal density for $Y$. See \citet{ghoshSteinSUREScoreMatching2025} for a unified treatment of \sure{} and Hyv\"arinen score matching.} Write $\|g(y)\|_{W} := \|g(y)\|_2 + \|Dg(y)\|_F$ for the combined function--Jacobian norm. Assumption~\ref{asm:regularity} formalizes polynomial-envelope regularity by controlling $\|g_\gamma(y)\|_W$ and the corresponding variation with $\gamma$ for every $y\in\bbR^n$, rather than only at the realized random vector $Y$.

\begin{assumption}[Pointwise-envelope regularity]\label{asm:regularity}
The candidate class is $\mathcal F=\{f_\gamma:\gamma\in\Gamma\}$, where $\Gamma\subset\bbR^{d_\Gamma}$ is compact and
\[
\operatorname{diam}(\Gamma)
:=
\sup_{\gamma,\gamma'\in\Gamma}\|\gamma-\gamma'\|_2
\leq D_\Gamma
\]
for a constant $D_\Gamma<\infty$ not depending on $n$. There exist $\beta \geq 0$, a scaling sequence $\nu_n > 0$, and a reference point $\gamma_0\in\Gamma$ such that, for all $y \in \bbR^n$,
\[
\|g_{\gamma_0}(y)\|_W
+
\sup_{\gamma\ne\gamma'}
\frac{\|g_\gamma(y)-g_{\gamma'}(y)\|_W}{\|\gamma-\gamma'\|_2}
\leq
\nu_n \left(1 + \frac{\|y\|_2}{\sqrt{n}}\right)^{2\beta}.
\]
The maps $g_{\gamma_0}$ and $g_\gamma-g_{\gamma'}$, $\gamma\ne\gamma'$, have continuous first partial derivatives. If $\Gamma$ is a singleton, the supremum is interpreted as zero.
\end{assumption}

The reference point $\gamma_0$ corresponds to the map $f_{\gamma_0}$ in the class. The first term in Assumption~\ref{asm:regularity} controls the adjustment $g_{\gamma_0}$ at this reference point, while the supremum controls how both $g_\gamma$ and $Dg_\gamma$ vary with $\gamma$. If the class contains the identity map, then $g_{\gamma_0}\equiv0$ is the natural reference choice. Otherwise, any fixed member of the class satisfying the displayed envelope can serve as the reference. The exponent $\beta$ controls how the envelope may grow with the normalized input magnitude $\|y\|_2/\sqrt n$: when $\beta=0$, the bound is uniform in $y$, while $\beta>0$ permits polynomial growth. For fixed linear smoothers $f(y)=Sy$, Lemma~\ref{lem:row_norm_linear_smoothers} in Appendix~\ref{app:general_regularity} shows that the pointwise envelope holds with $\beta=1/2$ and $\nu_n=O(\sqrt n)$ when the smoothing matrix $S$ has bounded operator norm and bounded maximum row Euclidean norm, $\max_{i\leq n}\|S_{i\cdot}\|_2$, where $S_{i\cdot}$ denotes the $i$th row of $S$. In spatial applications, the row bound captures a bounded-sensitivity form of local borrowing: the reported estimate for one spatial unit can average information from nearby units, but the Euclidean norm of the corresponding row of $S$ does not grow with $n$.

The continuous-differentiability requirement excludes some familiar nonsmooth shrinkage rules. For example, the one-dimensional hard-thresholding rule $f(y)=y\,1\{|y|>\tau\}$, $\tau>0$, has jumps at $y=\pm\tau$, so its derivative is not defined there and the rule is not covered by Assumption~\ref{asm:regularity}. The pointwise polynomial envelope is a sufficient condition for the concentration theorem below.

With these regularity conditions in hand, we state the main result: \sure{} tracks the realized loss uniformly over $\mathcal F$, and the expected excess realized loss of the \sure{}-selected rule is of order $\nu_n \max\{d_\Gamma,1\}^{4+\beta}/n$.

\begin{theorem}[Concentration and oracle inequality]\label{thm:uniform_conc}
Suppose Assumptions~\ref{asm:sampling_array} and~\ref{asm:regularity} hold.
Assume the \sure{} error process $f\mapsto\sure_n(f)-L_n(f)$ is separable and the minimizers
\[
\hat\gamma\in\argmin_{\gamma\in\Gamma}\sure_n(f_\gamma),
\qquad
\gamma^*\in\argmin_{\gamma\in\Gamma}L_n(f_\gamma)
\]
are measurable in $Y$. Set $\hat f:=f_{\hat\gamma}$ and $f^*:=f_{\gamma^*}$, with $f^*$ the realized-loss oracle. Then
\[
\E\left[\sup_{f\in \mathcal{F}}|\sure_n(f) - L_n(f)|\right] \lesssim \frac{1}{\sqrt{n}} + \frac{\nu_n \, \max\{d_\Gamma,1\}^{4+\beta}}{n}.
\]
The \sure{}-selected estimator also satisfies the oracle comparison
\[
\E\left[L_n(\hat{f}) - L_n(f^*)\right] \lesssim \frac{\nu_n\,\max\{d_\Gamma,1\}^{4+\beta}}{n}.
\]
\end{theorem}

The two displays control different quantities: the first is an uncentered uniform approximation bound for $\sure_n(f)$ over $\mathcal F$, while the second is the excess realized loss from choosing $\hat f$ by minimizing $\sure_n$. In the oracle comparison, the component of $\sure_n(f)-L_n(f)$ that does not depend on $f$ is common to $\hat f$ and $f^*$, so the $n^{-1/2}$ term from the first display does not enter. The proof is deferred to Appendix~\ref{app:general_regularity}, which first establishes a general Sobolev-moment concentration result and then applies Assumption~\ref{asm:regularity} to obtain the displayed bound.

The factor $\max\{d_\Gamma,1\}^{4+\beta}$ is a convention for including singleton classes in the same rate display. When $d_\Gamma\geq1$, this factor is $d_\Gamma^{4+\beta}$; when $\mathcal F$ is a singleton, with $d_\Gamma=0$, the uniform concentration bound still contains the $n^{-1/2}$ term and the $\nu_n/n$ contribution from the reference map $f_{\gamma_0}$, even though there is no variation over $\gamma$ to control. The constants hidden by $\lesssim$ depend only on the fixed bounds $C_\theta$, $\bar\sigma$, the regularity exponent $\beta$, and the diameter bound $D_\Gamma$. Apart from the displayed factors, these constants are uniform in $n$, $d_\Gamma$, and $\nu_n$.

\begin{remark}[Interpreting the rate]
The rate depends on three quantities: the envelope scale $\nu_n$, the parameter dimension $d_\Gamma$, and the permitted growth in the input vector, summarized by $\beta$. When $\beta=0$, neither the adjustment $g_\gamma$ nor its variation across $\gamma$ grows with the input $y$. The excess-loss bound is then $\nu_n\max\{d_\Gamma,1\}^4/n$; the separate uniform-approximation bound also contains the common $n^{-1/2}$ fluctuation term. This bounded-envelope case is different from global Lipschitzness. For example, a fixed linear smoother $f(y)=Sy$ has adjustment $g(y)=(S-I)y$, which can grow with $\|y\|_2$. Lemma~\ref{lem:row_norm_linear_smoothers} shows that bounded operator norm and bounded row norms are enough to cover such smoothers with $\beta=1/2$ and $\nu_n=O(\sqrt n)$.
\end{remark}

\begin{remark}[Relation to Bellec and Zhang (2021)]
The closest antecedent is \citet{bellecSecondOrderStein2021}, who analyze \sure{} selection from a finite collection of globally Lipschitz candidates. Theorem~\ref{thm:uniform_conc} studies \sure{} minimization over a compact $d_\Gamma$-dimensional continuum under polynomial-envelope regularity. Because the candidate classes and regularity conditions differ, the displayed rates answer different questions. We do not claim a sharper rate in their setting. The price of moving from a finite collection to a continuum is the polynomial factor in $d_\Gamma$.
\end{remark}

\begin{remark}[Kernel and prior misspecification]
The oracle inequality evaluates the maps in $\mathcal F=\{f_\gamma:\gamma\in\Gamma\}$ under realized squared-error loss for the fixed vector $\theta$. It does not require the \gp{} prior, spatial kernel, distance metric, or prior covariance that motivates those maps to be correctly specified. If the resulting class has large oracle loss $\inf_{\gamma\in\Gamma}L_n(f_\gamma)$ for the fixed vector $\theta$, that approximation error remains in the oracle benchmark. The theorem controls only the additional loss from selecting $\gamma$ with the observed data. This candidate-class issue is distinct from noise-covariance misspecification: if the covariance matrix used inside \sure{} differs from the true sampling covariance, Appendix~\ref{app:noise_covariance_misspecification} records the resulting bias term.
\end{remark}

\begin{remark}[Sobolev refinements]
Assumption~\ref{asm:regularity} is a convenient sufficient condition for the concentration argument: it gives an envelope that holds for every input $y\in\bbR^n$. Appendix~\ref{app:general_regularity} gives a more general formulation using moment bounds on derivatives under the law $P_Y=\N(\theta,\Sigma)$ of $Y$, rather than pointwise bounds in $y$. In that formulation, $k=0$ requires moment control of the adjustment and its first derivative, as implied by Assumption~\ref{asm:regularity}; $k=1$ adds second-derivative bounds, and larger $k$ adds bounds on the corresponding higher derivatives. Under the $k$th Sobolev moment condition, Theorem~\ref{thm:general_conc} gives dimension exponent $1+3\cdot2^{-k}+\beta$ in place of the $4+\beta$ exponent from the $k=0$ case. This exponent equals $4+\beta$ at $k=0$ and approaches $1+\beta$ as $k$ increases, provided the corresponding higher-order envelope can be verified with a scale $\nu_n$ of the same order, with an implied constant that may grow with $k$.
\end{remark}

\subsection{\sure{} Model Averaging over Trained Candidates}\label{sec:averaging_guarantee}

The averaging result returns to the finite library of trained maps $f_1,\ldots,f_K$ from Section~\ref{sec:model_averaging} and imposes regularity on each final trained map $f_k$ separately, rather than uniformly over the full training family for each candidate.  For averages of affine estimators $f_k(Y)=S_kY$ with fixed matrices $S_k$, sharp oracle inequalities based on unbiased risk estimates are available \citep[and references therein]{dalalyanSharpOracle2012}; Proposition~\ref{prop:averaging} covers finite libraries of nonlinear trained maps whose parameters are learned from the same data.

\begin{assumption}[Regularity for model averaging]\label{asm:averaging_regularity}
Let $f_1,\ldots,f_K$ be the trained candidate maps used for averaging, and write $g_k(y):=f_k(y)-y$. Each $g_k$ is continuously differentiable. There exist $\beta_k\geq0$ and $\nu_n^{(k)}>0$ such that
\[
\left(\E\left[\|g_k(Y)\|_W^p\right]\right)^{1/p}
\leq
\nu_n^{(k)}\,p^{\beta_k},
\qquad p\geq2,\qquad k=1,\ldots,K ,
\]
where $\|g(y)\|_W=\|g(y)\|_2+\|Dg(y)\|_F$ is the function--Jacobian norm of Section~\ref{sec:regularity}.
\end{assumption}

Assumption~\ref{asm:averaging_regularity} is a per-candidate moment condition on each final trained map $f_k$, including derivative contributions from any training rule used to construct that map. Because the averaging library is finite, the assumption does not require uniform increment bounds over a parameter set. A convenient sufficient condition is a pointwise polynomial envelope: if, for all $y\in\bbR^n$ and $k=1,\ldots,K$,
\[
\|g_k(y)\|_2+\|Dg_k(y)\|_F
\leq
\nu_n^{(k)}\left(1+\frac{\|y\|_2}{\sqrt n}\right)^{2\beta_k},
\]
then, under Assumption~\ref{asm:sampling_array}, the Gaussian moment bound in the proof of Lemma~\ref{lem:pointwise_regular_to_sobolev} gives Assumption~\ref{asm:averaging_regularity} with the same exponent $\beta_k$ and with $\nu_n^{(k)}$ inflated by a constant depending only on $\beta_k$, $C_\theta$, and $\bar\sigma$. Appendix~\ref{app:general_regularity} gives the corresponding class-level Sobolev-moment formulation, and Appendix~\ref{app:learned_parameter_regularity} gives sufficient conditions for the composite map $Y\mapsto f_{\hat\gamma_k(Y)}^{(k)}(Y)$ to satisfy Assumption~\ref{asm:averaging_regularity} when the parameter estimate $\hat\gamma_k(Y)$ is trained on the same data.

The proposition below applies the fixed-weight \sure{} criterion to this finite library and shows that, with $\bar\beta := \max_k \beta_k$ and $\bar\nu_n := \max_k \nu_n^{(k)}$, the expected regret of the \sure{}-weighted average, relative to the best fixed convex combination, is of order $(\log(eK))^{4+\bar\beta}\,\bar\nu_n/n$. For fixed $w\in\Delta^{K-1}$, write
\[
f_w(Y)=\sum_{k=1}^K w_k f_k(Y),
\]
where the Jacobian of $f_w$ treats the entries of $w$ as constants.

\begin{proposition}[Oracle inequality for model averaging]\label{prop:averaging}
Suppose Assumptions~\ref{asm:sampling_array} and~\ref{asm:averaging_regularity} hold and $\bar\beta\leq B<\infty$ for a fixed constant $B$. Let
\[
\hat w(Y)\in\argmin_{w\in\Delta^{K-1}}\sure_n(f_w)
\]
be a measurable global minimizer of the fixed-weight \sure{} criterion, and set $\tilde f(Y):=f_{\hat w(Y)}(Y)$. Then
\[
\E\left[L_n(\tilde f) - \min_{w \in \Delta^{K-1}} L_n(f_w)\right] \lesssim (\log(eK))^{4+\bar\beta}\, \frac{\bar\nu_n}{n}.
\]
The constants hidden by $\lesssim$ may depend on the fixed sampling bounds in Assumption~\ref{asm:sampling_array} and on $B$, but not on $n$, $K$, or $\bar\nu_n$ except through the displayed terms.
\end{proposition}

The proof, in Appendix~\ref{app:model_averaging_proof}, decomposes the fixed-weight \sure{} error into a noise term common to all weights plus a weighted average of candidate-specific terms. The maximum over the library yields the $(\log(eK))^{4+\bar\beta}$ factor.

The fixed-weight \sure{} criterion is a quadratic function of $w$ and can be computed exactly. Since $Df_w(Y)=\sum_k w_k Df_k(Y)$ when the weights are held fixed, define $A_{k\ell}(Y):=n^{-1}f_k(Y)^\top f_\ell(Y)$, $b_k(Y):=n^{-1}Y^\top f_k(Y)$, and $c_k(Y):=n^{-1}\tr\{\Sigma Df_k(Y)\}$. Then the fixed-weight criterion can be written as
\[
\sure_n(f_w)
=
\underbrace{
w^\top A(Y)w
}_{\text{quadratic in }w}
+
\underbrace{
2\{c(Y)-b(Y)\}^\top w
}_{\text{linear in }w}
+
\underbrace{
\frac{1}{n}\|Y\|_2^2-
\frac{1}{n}\tr(\Sigma)
}_{\text{constant in }w}.
\]
The matrix $A(Y)$ is positive semidefinite because $a^\top A(Y)a=n^{-1}\|\sum_k a_k f_k(Y)\|_2^2$ for any $a\in\bbR^K$. The display therefore shows that minimizing $\sure_n(f_w)$ over the simplex is a convex quadratic program (QP). The selected weights $\hat w(Y)$ are the minimizer of this QP.

\begin{remark}[Fixed weights and final evaluation]
The comparator in Proposition~\ref{prop:averaging} is the best fixed-weight convex combination for the realized $(Y,\theta)$. Since the individual candidates are vertices of the simplex,
\[
\min_{w \in \Delta^{K-1}} L_n(f_w) \leq \min_{1 \leq k \leq K} L_n(f_k),
\]
so the averaging oracle benchmark is weakly no worse than the best individual candidate. The proposition is stated for the trained maps $f_1,\ldots,f_K$ and does not require the candidates themselves to solve within-class optimization problems; any construction is allowed once the resulting maps satisfy Assumption~\ref{asm:averaging_regularity}. Evaluating the \sure{}-weighted average $\tilde f(Y)=f_{\hat w(Y)}(Y)$ is a different \sure{} calculation: it treats $\hat w(Y)$ as part of the estimator and therefore differentiates through the weight map $Y\mapsto\hat w(Y)$. Appendix~\ref{app:averaging_adaptive_weights} gives sufficient conditions under which this full-map \sure{} calculation remains unbiased for the risk of $\tilde f$.
\end{remark}

\begin{remark}[Uniqueness of the fixed-weight QP]
At the realized value of $Y$, the QP has a unique minimizer $\hat w(Y)$ whenever different simplex weights produce different averaged prediction vectors:
\[
w,w'\in\Delta^{K-1},\quad w\neq w'
\quad\Longrightarrow\quad
\sum_{k=1}^K w_k f_k(Y)\neq \sum_{k=1}^K w'_k f_k(Y).
\]
Uniqueness here concerns the selected weights $\hat w(Y)$, not the infeasible oracle problem $\min_{w\in\Delta^{K-1}}L_n(f_w)$. The one-to-one condition implies $(w-w')^\top A(Y)(w-w')>0$ for any distinct $w,w'\in\Delta^{K-1}$. Hence the QP has a unique minimizer $\hat w(Y)$, and the \sure{}-weighted average $\tilde f(Y)=f_{\hat w(Y)}(Y)$ is unique. The one-to-one condition is sufficient but not necessary for uniqueness of the weights. Even when the quadratic part is flat along some simplex direction, the full QP can still have a unique solution because the linear component $2\{c(Y)-b(Y)\}^\top w$ may favor one weight vector. For example, if $f_j(Y)=f_k(Y)$ for some $j\neq k$, shifting weight between candidates $j$ and $k$ leaves the averaged prediction vector unchanged, so the one-to-one condition fails.
\end{remark}

\begin{remark}[Spread of realized losses across candidates]
For any fixed weights $w$, the realized loss of the average can be decomposed by expanding squared norms:
\[
L_n(f_w) = \sum_k w_k \, L_n(f_k) \;-\; \underbrace{\frac{1}{n}\sum_k w_k \|f_k(Y) - f_w(Y)\|_2^2}_{\text{dispersion} \;\geq\; 0}.
\]
The decomposition shows why the realized-loss oracle over convex combinations can be below the best single candidate: averaging subtracts a nonnegative dispersion term from the weighted average of individual realized losses. The dispersion term is positive whenever $w_j,w_k>0$ for some candidates $j\neq k$ with $f_j(Y)\neq f_k(Y)$. A positive dispersion term is not, by itself, a finite-sample guarantee that the \sure{}-selected convex average has lower realized loss than the best individual candidate. Proposition~\ref{prop:averaging} instead controls regret relative to the infeasible best convex average, $\min_{w\in\Delta^{K-1}}L_n(f_w)$. A fixed average beats the best single candidate exactly when
\[
\frac{1}{n}\sum_k w_k \|f_k(Y)-f_w(Y)\|_2^2
>
\sum_k w_k L_n(f_k)-\min_j L_n(f_j).
\]
\end{remark}

\begin{remark}[Envelope scale and averaging rate]
The averaging bound is driven by the largest per-candidate envelope scale $\bar\nu_n=\max_k\nu_n^{(k)}$. For fixed spatial smoothers $f(y)=Sy$, the scale $\bar\nu_n=O(\sqrt n)$ corresponds to bounded per-unit sensitivity: the row-norm condition in Lemma~\ref{lem:row_norm_linear_smoothers} requires each row of $S-I$ to remain bounded as $n$ grows.\footnote{The lemma's pointwise envelope implies Assumption~\ref{asm:averaging_regularity} via the sufficient condition stated after that assumption.} Under this envelope scale, Proposition~\ref{prop:averaging} gives
\[
\E\left[
L_n(\tilde f)-\min_{w\in\Delta^{K-1}}L_n(f_w)
\right]
\lesssim
\frac{(\log(eK))^{4+\bar\beta}}{\sqrt n}.
\]
Thus the regret vanishes for fixed $K$, and more generally whenever $(\log(eK))^{4+\bar\beta}=o(\sqrt n)$.
\end{remark}

Appendix~\ref{app:regularity_verification} verifies Assumption~\ref{asm:averaging_regularity} for the main estimator building blocks used in the Opportunity Atlas application, and gives trained-parameter and closure tools for assembling trained candidates from those pieces. These sufficient-condition checks connect Proposition~\ref{prop:averaging} to the empirical candidate library: the application uses \sure{} to average over the trained library. The within-class result, Theorem~\ref{thm:uniform_conc}, is the separate guarantee for exact \sure{} training over parameterized shrinkage classes satisfying the stronger uniform regularity condition; Appendix~\ref{app:implementation} records how the reported \sure{} values are computed. We now turn to the Opportunity Atlas application, where the candidate maps differ in how they define which tracts are neighbors and the empirical question is whether shrinkage and averaging reduce estimated squared-error loss across commuting zones.

\section{Economic Mobility in the Opportunity Atlas}\label{sec:empirical}

Does spatial shrinkage improve estimates of neighborhood economic mobility, and does averaging over spatial specifications reduce sensitivity to that choice?  The Opportunity Atlas \citep{chettyOpportunityAtlasMapping2018} estimates tract-level intergenerational economic mobility for over 70,000 Census tracts in the United States.  The application is motivated by evidence that economic mobility varies substantially across places and that childhood exposure to neighborhoods can affect adult outcomes \citep{chettyWhereLandOpportunity2014, chettyImpactsNeighborhoodsIntergenerational2018}.  The target here is narrower: estimating the latent tract-level mean of the released Opportunity Atlas outcome, not re-estimating causal exposure effects.  Opportunity Atlas estimates are used to rank neighborhoods in settings such as housing mobility programs \citep{bergmanCreatingMovesOpportunity2024}, so reducing estimation error can change which places are identified as high-opportunity.  The released tract-level estimates are noisy measurements of latent neighborhood mobility, with reported standard errors and pronounced spatial patterning across nearby tracts.  Because each tract estimate is observed only once, the empirical comparison cannot be organized around holdout performance.  We therefore use \sure{}, with sampling variances implied by the reported standard errors, as the common risk scale for these comparisons.

The main empirical comparison yields two findings.  First, spatial shrinkage substantially improves \sure{}-estimated MSE relative to non-spatial empirical Bayes baselines.  Second, multiple plausible spatial specifications compete across commuting zones (CZs): geographic distance has lower \sure{}-estimated MSE in some CZs, while contiguity (tract adjacency) has lower \sure{}-estimated MSE in others, and using OLS to residualize tract estimates on demographic covariates before spatial smoothing can change rankings within each distance family.  This heterogeneity is the reason the application reports a \sure{}-weighted average rather than choosing a single spatial specification for all CZs.  Candidates are trained, evaluated, and averaged separately within each CZ. National summaries average each CZ's \emph{\sure{} ratio}, the CZ-level \sure{}-estimated MSE divided by the corresponding raw-\mle{} benchmark, weighting each CZ by its tract count. Selected averaging weights are also averaged across CZs using tract-count weights.  The \sure{}-weighted average matches or improves on the best individual candidate's \sure{} ratio in \oaAggBeatsBest{} of \oaAggTotal{} CZs.

\subsection{Data}\label{sec:oa_data}

We study \oaNCzs{} CZs spanning a range of sizes ($n$ from \oaMinN{} to \oaMaxN{} tracts, median \oaMedianN{}; \oaTotalTracts{} tracts in total).  Together, these CZs cover more than one-third of the U.S. population.  The main outcome is pooled household income rank in adulthood for children with parents at the 25th percentile of the national income distribution (\texttt{kfr\_pooled\_pooled}; we refer to this family of children's household income-rank outcomes as KFR); in the Opportunity Atlas Table~1 data, this rank is measured in 2014--2015 for the 1978--1983 birth cohorts.  Each tract $i$ has an estimate $Y_i$ and reported standard error $\mathrm{se}_i$, with variance $\sigma_i^2=\mathrm{se}_i^2$.  These are the raw, unshrunk tract estimates and sampling standard errors, so $f(Y)=Y$ is the maximum-likelihood benchmark.\footnote{The empirical analysis treats the reported standard errors as fixed known sampling standard errors and does not account for uncertainty in the standard-error estimates themselves.}  The Gaussian location model is $Y_i = \theta_i + \varepsilon_i$, $\varepsilon_i \sim \N(0, \sigma_i^2)$, with variances varying by a factor of $10$--$100$ across tracts within a CZ.  This heteroskedasticity reflects differences in tract-level effective sample size and in how precisely the underlying Opportunity Atlas regressions estimate outcomes at the 25th percentile of parent income.

The empirical comparison uses two metrics to encode spatial proximity.  \emph{Geographic distance} is the Euclidean distance between tract centroids in longitude--latitude coordinates.  \emph{Contiguity distance} is the shortest-path distance on the tract adjacency graph, where two tracts are adjacent if they share a boundary or vertex, so that distance 1 means direct neighbors, distance 2 means neighbors-of-neighbors, and so on.  These metrics capture different notions of spatial relatedness: contiguity can reflect administrative and social boundaries that may not align with physical distance, such as tracts separated by a river or highway.

\subsection{Candidate Estimators}\label{sec:oa_models}

The main comparison across CZs uses a fixed library of $K = 7$ candidate maps $Y\mapsto f_k(Y)$, summarized in Table~\ref{tab:oa_candidates}; additional variants appear only in supporting analyses.  The library of candidate maps is designed to vary three empirical choices: non-spatial versus spatial pooling, geographic versus contiguity distance, and spatial smoothing with versus without covariate residualization.  The candidates include non-spatial baselines (\mle{}, \nneb{}, \closegauss{}) and spatial \gp{} candidates that vary the distance metric and preprocessing.  The OLS-preprocessed spatial candidates residualize tract estimates on four tract-level demographic covariates: percent White, percent Black, percent Hispanic, and median age.  This covariate-residualization step has the same motivation as the small-area-estimation use of auxiliary covariates with noisy area-level estimates \citep{fayEstimatesIncomeSmall1979}. It is used only for the OLS-labeled spatial candidates.  The \closegauss{} benchmark is a Gaussian, precision-dependent EB rule in the spirit of \citet{chenEmpiricalBayesWhen2024}: estimates are locally centered and scaled using weights formed from log reported variance, and the standardized values are then shrunk by the same heteroskedastic normal--normal posterior-mean formula as \nneb{}.

\begin{table}[!htbp]
  \centering
  \begin{threeparttable}
  \caption{Candidate estimators in the main Opportunity Atlas comparison across CZs.}\label{tab:oa_candidates}
  \begin{tabular}{lllll}
    \toprule
    Estimator & GP kernel & Distance metric & Preprocessing & Training rule \\
    \midrule
    \mle{} & --- & --- & --- & --- \\
    \nneb{} & --- & --- & --- & Closed-form \\
    \closegauss{} & --- & --- & Local NW & Closed-form \\
    \addlinespace
    \gp{} Geo & Mat\'ern-$\tfrac{1}{2}$ & Euclidean & Local NW & AdamW--\sure{} \\
    \gp{} Contig & Mat\'ern-$\tfrac{1}{2}$ & Contiguity & Local NW & AdamW--\sure{} \\
    \gp{} Geo OLS & Mat\'ern-$\tfrac{1}{2}$ & Euclidean & Local NW, OLS & AdamW--\sure{} \\
    \gp{} Contig OLS & Mat\'ern-$\tfrac{1}{2}$ & Contiguity & Local NW, OLS & AdamW--\sure{} \\
    \addlinespace
    \bottomrule
  \end{tabular}
  \begin{tablenotes}[flushleft]\footnotesize
  \item \textit{Notes:} Euclidean is coordinate distance between tract centroids in longitude-latitude coordinates; Contiguity is shortest-path distance on the tract adjacency graph (Section~\ref{sec:oa_data}).  Mat\'ern-$\tfrac{1}{2}$ is the exponential covariance.  Local NW denotes Nadaraya--Watson centering and scaling with weights based on log reported variance.  OLS denotes residualization on demographic covariates before the Nadaraya--Watson standardization and spatial \gp{} smoothing steps.  \nneb{} and \closegauss{} are closed-form heteroskedastic normal--normal EB rules with method-of-moments parameter estimates; \nneb{} uses the global rule, while \closegauss{} applies the same shrinkage formula after Nadaraya--Watson preprocessing.  AdamW--\sure{} means that each \gp{} candidate trains its tuning parameters with AdamW on the proxy \sure{} criterion.  After training, reported \sure{} values are recomputed for the implemented map, including trained-parameter dependence (Section~\ref{sec:oa_risk_eval}; Appendix~\ref{app:implementation}).
  \end{tablenotes}
  \end{threeparttable}
\end{table}

The preprocessing labels describe transformations that are part of the candidate map \(Y\mapsto f_k(Y)\).  For rows labeled Local NW, Nadaraya--Watson weights based on log reported variance define a local mean and scale; shrinkage is applied to the standardized estimates, and predictions are then transformed back to the original rank scale.  For rows also labeled OLS, the estimates are first residualized on demographic covariates; the same Nadaraya--Watson standardization and spatial \gp{} shrinkage are then applied to the residualized estimates, and the fitted covariate component is added back afterward.  For \sure{} evaluation, automatic differentiation tracks the full map \(Y\mapsto f_k(Y)\) for each candidate, including residualization, standardization, shrinkage in the transformed space, and transformations back to ranks, conditional on the fixed Nadaraya--Watson weights and the scale floors (small constants that bound the local scale estimates away from zero).  The value-similarity rule in Example~\ref{ex:bilateral_shrinkage} appears in the Cook County comparison in Section~\ref{sec:oa_value_similarity_ladder}.  It is not part of the main seven-candidate average across CZs.  Rather than select a single row of the table, we take as the primary empirical estimator the convex average whose weights are chosen by minimizing \sure{} over the seven candidate maps, as in Section~\ref{sec:model_averaging}; the reported \sure{} for this average evaluates the full map $\tilde f$, including both trained-parameter and selected-weight dependence.

\subsection{Risk Evaluation}\label{sec:oa_risk_eval}

All main empirical comparisons evaluate the trained candidate maps \(Y\mapsto f_k(Y)\), including their preprocessing and training steps, and the final \sure{}-weighted convex average of those maps.  For a map \(f\), the loss of interest is \(L_n(f)=n^{-1}\|f(Y)-\theta\|_2^2\), and the corresponding risk is \(R_n(f)=\E[L_n(f)]\).  Because the latent tract-level vector \(\theta\) is unobserved, the empirical tables use \sure{} as the observable risk estimate.  Under the Gaussian location model, if the covariance matrix used in the \sure{} formula equals the true sampling covariance of \(Y\), then \(\E[\sure_n(f)]=R_n(f)\).  Thus \sure{} puts candidate maps on a common mean-squared-error scale under the stated sampling model.

The Opportunity Atlas reports a standard error for each tract-level regression estimate.  The empirical evaluation uses these standard errors to form
\[
\Sigma=\diag(\mathrm{se}_1^2,\ldots,\mathrm{se}_n^2),
\]
which treats sampling errors across tract estimates as uncorrelated. This covariance assumption concerns the estimation noise in the released tract estimates, not the spatial dependence in the latent mobility vector.  We interpret each \(Y_i\) as a direct estimate of tract \(i\)'s latent mobility mean \(\theta_i\); under separate tract-level estimation with disjoint underlying observations, the reported marginal standard errors are a natural working choice for the covariance input in the Gaussian location approximation.  The main remaining source of off-diagonal sampling covariance would be overlap in the underlying children contributing to multiple tract estimates, for example among movers.  Appendix~\ref{app:noise_covariance_misspecification} gives the omitted-covariance bias decomposition and shows that comparisons are most affected when candidate maps differ in how they smooth across pairs of tracts with correlated estimation errors.

The reported \sure{} values are also distinct from the values of proxy \sure{} used to train the \gp{} candidates.  After training, \sure{} is recomputed for the implemented map \(Y\mapsto f_k(Y)\), accounting for trained-parameter dependence.  The same issue arises for the final convex average: the fixed-weight \sure{} criterion selects the weights, while the reported \sure{} value evaluates the full map \(\tilde f\colon Y\mapsto f_{\hat w(Y)}(Y)\), including the derivative of the selected weights with respect to the data.  Under the conditions in Appendix~\ref{app:averaging_adaptive_weights}, \sure{} remains unbiased for the risk of \(\tilde f\).  Appendix~\ref{app:proxy_exact_case} compares the reported \sure{} values with fixed-parameter proxies that treat trained parameters as constants. For individual candidates, this gap is the training correction, and its sign shows the optimism from ignoring trained-parameter dependence.  For the \sure{}-weighted average, the gap also reflects the correction for the data-selected weights.  The coupled-bootstrap procedure of \citet{oliveiraUnbiasedRiskEstimation2024} is a practical derivative-free alternative: by refitting on one perturbed sample and evaluating on its coupled counterpart, it gives an unbiased estimate of the risk for the rule trained on a variance-inflated input, approaching the original-risk target as the perturbation level shrinks.  Appendix~\ref{app:cb} reports this comparison for one CZ.  The main evaluation across CZs uses \sure{} because the candidate maps are differentiable and the automatic differentiation runs efficiently. Appendix~\ref{app:implementation} gives the implementation details.

\subsection{Results}\label{sec:oa_results}

\subsubsection{Heterogeneity Across Commuting Zones}\label{sec:oa_heterogeneity}

The central empirical finding is that the spatial candidates have lower \sure{}-estimated MSE than the non-spatial empirical Bayes baselines in every CZ, while the best spatial specification varies across CZs.  Figure~\ref{fig:oa_heterogeneity} plots each highlighted candidate's \sure{} ratio in each CZ.  The highlighted series are the geographic and contiguity \gp{} variants, with and without OLS preprocessing, together with the \sure{}-weighted average; the gray dashed line traces the best individual candidate in each CZ.

The best geographic-distance candidate wins in \oaGeoWins{} CZs while the best contiguity-distance candidate wins in \oaContigWins{}.  The pattern does not follow CZ size or an obvious geographic rule.  Neither distance metric systematically dominates.  This heterogeneity extends beyond distance metrics: within a given CZ, OLS preprocessing can also change the ranking of spatial candidates.  The geographic \gp{} has lower \sure{}-estimated MSE than its contiguity counterpart in \oaGeoPlainWins{} CZs without OLS preprocessing, while the geographic OLS-preprocessed \gp{} has lower \sure{}-estimated MSE than the contiguity OLS-preprocessed version in \oaGeoOlsWins{} CZs.

The pattern in Figure~\ref{fig:oa_heterogeneity} is the empirical case for averaging.  The highlighted spatial candidates do not move in parallel across CZs: a distance metric or preprocessing choice that performs well in one CZ can lie well above the lower envelope---the best individual candidate's ratio in each CZ---in another.  The \sure{}-weighted average uses the same \sure{}-estimated MSE scale within each CZ to choose convex weights over the candidate maps, rather than imposing a national choice between geographic distance, contiguity distance, and OLS preprocessing.  The resulting \sure{}-weighted average closely tracks the lower envelope, matching or improving on the best individual candidate's \sure{} ratio in \oaAggBeatsBest{}/\oaAggTotal{} CZs.  Thus the role of \sure{}-weighted averaging in the application is to retain the gains from spatial shrinkage while reducing sensitivity to which spatial specification is best in a given CZ.

\begin{figure}[!t]
  \centering
  \includegraphics[height=0.74\textheight,keepaspectratio]{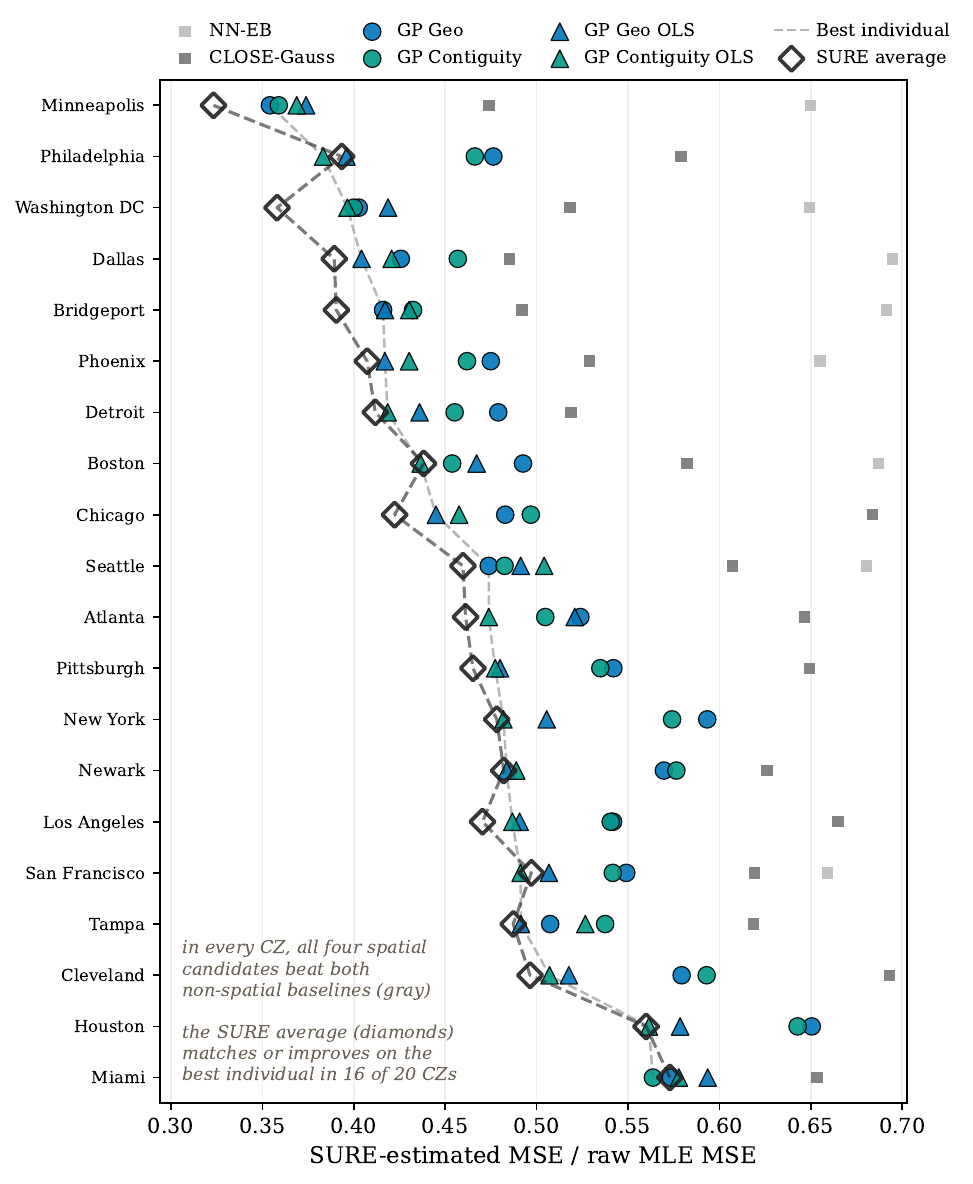}
  \caption{Heterogeneity across commuting zones.  The horizontal axis reports each candidate's \sure{} ratio, its \sure{}-estimated MSE divided by the MSE of the raw \mle{} estimator; lower values are better, and a value of $0.5$ corresponds to a 50\% lower estimated MSE than the raw \mle{}.  CZs are ordered vertically by the best individual candidate's value of the same ratio, and the gray dashed line traces that best individual candidate.  Highlighted series show geographic-distance and contiguity-distance \gp{} variants, with and without OLS preprocessing, and the \sure{}-weighted average.  Geographic distance (blue) and contiguity distance (teal) alternate as the better choice across CZs, and the \sure{}-weighted average (diamonds) tracks the lower envelope, matching or improving on the best individual candidate's \sure{} ratio in \oaAggBeatsBest{}/\oaAggTotal{} CZs.  The horizontal axis is truncated at the worst spatial candidate's ratio; some non-spatial values extend further right.}
  \label{fig:oa_heterogeneity}
\end{figure}

\subsubsection{\sure{} Averaging Across Commuting Zones}\label{sec:oa_averaging}

Table~\ref{tab:oa_results} reports tract-weighted average performance across all \oaNCzs{} CZs.  All four spatial candidates have substantially lower \sure{}-estimated MSE than the non-spatial baselines: the reduction relative to raw \mle{} ranges from \oaReductionGpGeo\% for the geographic \gp{} to \oaReductionGpContigOls\% for the OLS-preprocessed contiguity \gp{}, compared to \oaReductionLsClose\% for \closegauss{}.
\input{tables/auto_oa_table}

The \sure{}-weighted average has a tract-weighted average \sure{} ratio of \oaAggReportedSure{}, lower than every individual candidate in Table~\ref{tab:oa_results}. The improvement illustrates the averaging logic of Section~\ref{sec:averaging_guarantee}: candidate rankings differ across CZs, so minimizing \sure{} over convex combinations can lower \sure{}-estimated MSE without committing to a single spatial specification. The reported table value is the \sure{} evaluation of the \sure{}-weighted average $\tilde f$. This evaluation differentiates through both the trained candidate parameters and the selected weights; Appendix~\ref{app:averaging_adaptive_weights} gives smoothness conditions under which it is unbiased, and the combined correction relative to the fixed-weight, fixed-parameter proxy average is $+\oaAggFullMinusFixed{}$ on the \mle{}-normalized scale (Table~\ref{tab:oa_proxy_exact}). The selected averaging weights concentrate on the spatial candidates: OLS-preprocessed contiguity \gp{} receives \oaWeightGpContigOls\%, OLS-preprocessed geographic \gp{} receives \oaWeightGpGeoOls\%, and geographic \gp{} receives \oaWeightGpGeo\%, with the remaining weight spread across the other candidates.

\subsubsection{A Value-Similarity Comparison for Cook County Tracts in the Chicago Commuting Zone}\label{sec:oa_value_similarity_ladder}

As a supporting comparison, Figure~\ref{fig:oa_geo_bilateral_ladder} reports what happens when value-similarity smoothing is added in one setting where nearby neighborhoods differ sharply: Cook County tracts selected from the Chicago CZ.  The stepwise comparison begins with non-spatial baselines, then adds the geographic \gp{}, the OLS-preprocessed geographic \gp{}, and finally \gpbilat{}.  The \gpbilat{} candidate implements the value-similarity idea from Example~\ref{ex:bilateral_shrinkage}: starting from OLS-preprocessed geographic smoothing, it gives more weight to nearby tracts whose observed mobility estimates are similar.  On the same \sure{}-ratio scale, adding \gpbilat{} lowers both ratios: the \sure{}-weighted average falls from $\oaLadderAggBeforeBilat$ to $\oaLadderAggAfterBilat$, and the best individual ratio falls from $\oaLadderBestBeforeBilat$ for the OLS-preprocessed geographic \gp{} to $\oaLadderBestAfterBilat$ once \gpbilat{} is added.  The \sure{}-weighted average places weight $\oaLadderBilatWeight$ on \gpbilat{} in the final comparison step.  This focused comparison shows that, in a setting where nearby neighborhoods differ sharply, adding value similarity lowers \sure{}-estimated MSE further.

\begin{figure}[!t]
    \centering
    \includegraphics[width=0.84\textwidth]{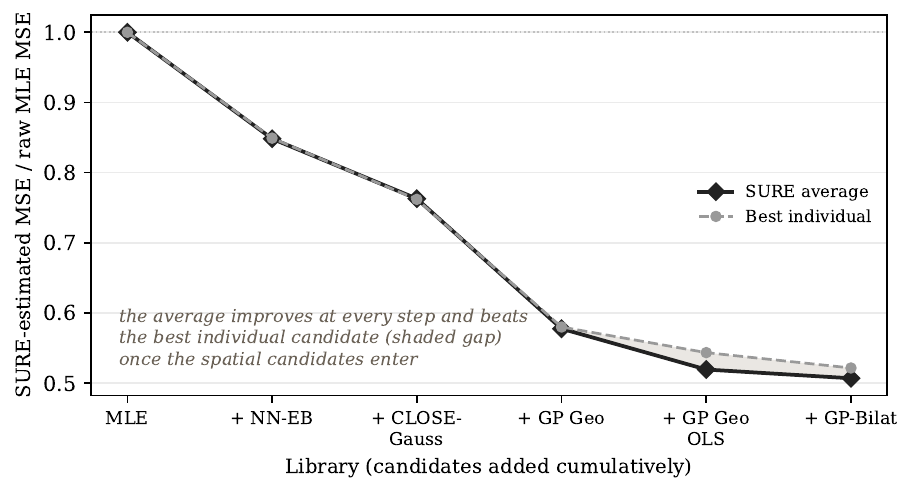}
    \caption{Value-similarity comparison for pooled economic mobility on Cook County tracts within the Chicago commuting zone.  Each step adds one estimator to this comparison-specific library and recomputes the \sure{}-weighted average for the same Cook County sample.  The shaded region marks the gap between the \sure{}-weighted average and the best individual candidate at each step.  The final step adds \gpbilat{}, the OLS-preprocessed geographic \gp{} with an additional value-similarity (bilateral-filter) kernel.  Adding \gpbilat{} lowers both the best individual candidate's \sure{} ratio and the \sure{}-weighted average's ratio.}
    \label{fig:oa_geo_bilateral_ladder}
\end{figure}

\subsubsection{Targeting High-Mobility Tracts}\label{sec:oa_targeting_application}

The \sure{}-estimated-MSE comparisons evaluate accuracy of the estimated mobility vector, but these estimates are often used as inputs into ranking and selection decisions, a compound-decision setting studied by \citet{guKoenkerInvidiousComparisons2023}.  Related work studies inference for ranks and selected high-opportunity neighborhoods from noisy Opportunity Atlas estimates \citep{mogstadInferenceRanksApplications2024, andrewsInferenceWinners2024}.  We use the shrinkage estimates in a targeting exercise under a top-third rule, comparing which tracts each rule selects and the selected group's average latent mobility rank.  Within each CZ, each rule ranks tracts by the estimated outcome, selects the top third, and estimates that group's average latent mobility rank using an evaluation strategy motivated by the coupled-bootstrap procedure of \citet{oliveiraUnbiasedRiskEstimation2024}; Appendix~\ref{app:oa_robustness} describes the implementation and Appendix~\ref{app:assure} gives the corresponding unbiasedness calculation.  The targeting library contains four trained maps: \mle{}, \nneb{}, \closegauss{}, and the geographic-distance \gp{} candidate from Table~\ref{tab:oa_candidates}, which uses Local NW preprocessing and no OLS residualization.  Table~\ref{tab:oa_targeting_welfare} reports the three non-\mle{} maps and the \sure{}-weighted average formed from all four, with estimated gains measured relative to the raw-\mle{} targeting rule for the pooled outcome and three subgroup outcomes.  \mle{} enters as the zero benchmark rather than as a separate row.  Dollar-equivalent gains use the official Opportunity Atlas 2015 percentile-dollar crosswalk, so they should be read as an interpretation of rank gains rather than as a separately estimated dollar outcome.

\input{tables/auto_oa_targeting_welfare_table}

Appendix~\ref{app:oa_robustness} reports the same four-candidate comparison with the geographic-distance \gp{} on the \sure{}-ratio scale for these related KFR outcomes.  In Table~\ref{tab:oa_targeting_welfare}, the largest absolute targeting gain occurs for the Black-male outcome: the \sure{}-weighted average improves top-third targeting by an estimated 1.21 rank points, or about \$1,265 in dollar-equivalent terms.  The table also separates overall shrinkage gains from the incremental gain of geographic smoothing over the non-spatial \closegauss{} benchmark: relative to \closegauss{}, the geographic-distance \gp{} adds more for the pooled-male and White-male outcomes than for the Black-male outcome.
\FloatBarrier

\section{Conclusion}\label{sec:conclusion}

We develop \sure{}-based model averaging, and the selection guarantees that underpin it, for shrinkage maps that exploit spatial structure. The risk comparisons are among the resulting maps, so the prior distribution, prior covariance structure, or similarity rule used to motivate a candidate map need not be correctly specified as a model for \(\theta\). The theory gives sufficient conditions for two uses of \sure{} with nonlinear shrinkage maps that have cross-unit dependence. One result covers selection within a parameterized class by minimizing \sure{}. The model-averaging result shows that once a candidate library has been assembled, the fixed-weight \sure{} criterion can average its members with a finite-library oracle guarantee.

Empirically, in the main pooled-outcome comparison, every spatial candidate has a lower \sure{} ratio than both non-spatial empirical Bayes baselines in every one of the \oaNCzs{} CZs. The best individual spatial specification varies with local geography, and the \sure{}-weighted average of candidate maps has a \sure{} ratio of \oaAggReportedSure{}.  A supporting top-third targeting exercise finds that the \sure{}-weighted average selects tracts with higher estimated average income ranks than the raw-\mle{} targeting rule across the reported outcomes.

For applications with several plausible notions of similarity, the practical lesson is to build a diverse library of shrinkage estimators and let \sure{} evaluate and average them, rather than committing to a single exchangeable model \emph{ex ante}. The same issue arises for noisy area, school, hospital, or firm-level estimates whenever researchers have several credible ways to pool information across units. Targeting and other decision-focused analyses remain important \citep[e.g.,][]{chenCompoundSelectionDecisions2025}. This paper instead applies \sure{} to average shrinkage maps for estimating the latent vector, with targeting treated as a downstream application of those estimates. \sure{}'s unbiasedness leans on the Gaussian sampling model for the noise in $Y$. For tract-level regression predictions built from many observations this is an approximation we consider mild. In the application, \sure{} is computed treating the reported standard errors as known and the sampling errors across tracts as uncorrelated. The main threat to that covariance assumption is overlap in the underlying children contributing to multiple tract estimates, and Appendix~\ref{app:noise_covariance_misspecification} characterizes what that misspecification costs.

\input{main.bbl}
\clearpage
\appendix

\section*{Appendix Roadmap}
The appendices are organized around four supporting components: risk and averaging guarantees, regularity verification for estimator forms used in the application, empirical implementation details, and supplementary empirical comparisons.
\begin{itemize}[leftmargin=*,nosep]
  \item \hyperref[app:theory_proofs]{Appendix~\ref*{app:theory_proofs}: Theory and Proofs}
  \begin{itemize}[leftmargin=1.5em,nosep]
    \item \hyperref[app:mse_decision_relevance]{MSE and stable plug-in decisions} states a stable plug-in condition under which lower latent squared-error loss controls downstream ranking or selection errors.
    \item \hyperref[app:general_regularity]{Regularity framework and Sobolev extensions} states the general concentration result and derives the within-class oracle inequality.
    \item \hyperref[app:model_averaging_proof]{Model-averaging proof} proves the fixed-weight model-averaging oracle inequality.
    \item \hyperref[app:noise_covariance_misspecification]{Noise-covariance misspecification} records the bias term that appears when \sure{} is computed with an approximate covariance matrix.
    \item \hyperref[app:averaging_adaptive_weights]{\sure{} evaluation of the weighted average} states sufficient conditions under which \sure{} remains unbiased when the averaging weights are smooth functions of $Y$.
  \end{itemize}
  \item \hyperref[app:regularity_verification]{Appendix~\ref*{app:regularity_verification}: Regularity Verification for Estimators}
  \begin{itemize}[leftmargin=1.5em,nosep]
    \item \hyperref[app:learned_parameter_regularity]{Trained-parameter regularity} gives sufficient conditions for trained composites, and the \hyperref[app:regularity_closure]{closure tools} provide the affine-map, product, composition, and standardization steps used to verify candidate constructions.
    \item \hyperref[sec:verification]{Estimator-form verifications} cover normal--normal EB shrinkage, \close{}-style Nadaraya--Watson centering-and-scaling constructions, and fixed-kernel \gp{} smoothing, and show that value-similarity smoothing requires a separate moment-envelope argument because a basic value-similarity map need not be globally Lipschitz.
    \item \hyperref[sec:bilateral_verification]{Fixed-parameter value-similarity regularity} verifies the per-candidate moment condition for a single fixed value-similarity map used as a building block in the Cook County comparison; under a bounded-row-sum condition on the fixed geographic covariance factor, the envelope scale is \(\nu_n=O(\sqrt n)\).
  \end{itemize}
  \item \hyperref[app:implementation_computation]{Appendix~\ref*{app:implementation_computation}: Implementation and Computation}
  \begin{itemize}[leftmargin=1.5em,nosep]
    \item \hyperref[app:implementation]{Opportunity Atlas implementation details} summarize diagonal-covariance \sure{} evaluation, training, tangent tracking, Hutchinson trace estimation, preprocessing, distance construction, and \sure{}-based model averaging.
  \end{itemize}
  \item \hyperref[app:empirical_diagnostics]{Appendix~\ref*{app:empirical_diagnostics}: Supplementary Empirical Analyses and Additional Outcomes}
  \begin{itemize}[leftmargin=1.5em,nosep]
    \item Supplementary analyses include \hyperref[app:proxy_exact_case]{fixed-parameter proxy comparisons}, \hyperref[app:oa_robustness]{additional KFR outcomes and the top-third targeting implementation}, \hyperref[app:cb]{the coupled-bootstrap comparison}, and \hyperref[app:assure]{the ASSURE targeting-welfare comparison}.
  \end{itemize}
\end{itemize}

\section{Theory and Proofs}\label{app:theory_proofs}

\subsection{MSE and Stable Plug-In Decisions}\label{app:mse_decision_relevance}

We use latent squared-error loss as an estimation criterion.  This does
not require squared-error loss to coincide with a policymaker's welfare function.
The following bound makes this precise: for stable plug-in objectives, accurate latent estimation controls the downstream welfare approximation error.  Empirical Bayes work makes the same bridge
in more specialized problems; for example,
\citet[Theorem~4]{chenEmpiricalBayesWhen2024} shows that squared-error
accuracy of posterior-mean estimates controls regret for plug-in ranking and
selection decisions.

\begin{lemma}[Stable welfare calculations from latent estimation error]\label{lem:plugin_regret_mse}
Let $\mathcal D$ be a decision set and let $W_n(d,\vartheta)$ be a welfare
criterion defined for $d\in\mathcal D$ and latent vectors
$\vartheta\in\Theta\subseteq\bbR^n$.  Suppose the true vector
$\theta$ and the estimate $\hat\theta$ belong to $\Theta$, and suppose
there exists $L<\infty$ such that, for all $d\in\mathcal D$ and all
$\vartheta,\vartheta'\in\Theta$,
\[
|W_n(d,\vartheta)-W_n(d,\vartheta')|
\leq
L\,\frac{\|\vartheta-\vartheta'\|_2}{\sqrt n}.
\]
Then
\[
\sup_{d\in\mathcal D}
\left|W_n(d,\hat\theta)-W_n(d,\theta)\right|
\leq
L\,\frac{\|\hat\theta-\theta\|_2}{\sqrt n}.
\]
Consequently, if $\hat\theta=f(Y)$ and $f(Y)\in\Theta$ almost surely,
then
\[
\E\!\left[
\sup_{d\in\mathcal D}
\left|W_n(d,f(Y))-W_n(d,\theta)\right|
\right]
\leq
L\,\sqrt{R_n(f)}.
\]
If, in addition, $d^*(\theta)$ is any maximizer of $W_n(d,\theta)$ and
$\hat d$ is any maximizer of $W_n(d,\hat\theta)$, then the plug-in regret satisfies
\[
W_n(d^*(\theta),\theta)-W_n(\hat d,\theta)
\leq
2L\,\frac{\|\hat\theta-\theta\|_2}{\sqrt n}.
\]
\end{lemma}

\begin{proof}[Proof of Lemma~\ref{lem:plugin_regret_mse}]
The first display follows directly by taking the supremum over $d$ in the
Lipschitz condition with $\vartheta=\hat\theta$ and
$\vartheta'=\theta$.  Taking expectations and applying Jensen's inequality
gives
\[
\E\frac{\|f(Y)-\theta\|_2}{\sqrt n}
\leq
\left(\E\frac{\|f(Y)-\theta\|_2^2}{n}\right)^{1/2}
=
\sqrt{R_n(f)}.
\]
For the plug-in regret statement, add and subtract welfare evaluated at
$\hat\theta$:
\begin{align*}
W_n(d^*(\theta),\theta)-W_n(\hat d,\theta)
&=
\{W_n(d^*(\theta),\theta)-W_n(d^*(\theta),\hat\theta)\} \\
&\quad+
\{W_n(d^*(\theta),\hat\theta)-W_n(\hat d,\hat\theta)\} \\
&\quad+
\{W_n(\hat d,\hat\theta)-W_n(\hat d,\theta)\}.
\end{align*}
The middle term is nonpositive because $\hat d$ maximizes
$W_n(d,\hat\theta)$.  The first and third terms are each bounded by
$L\|\hat\theta-\theta\|_2/\sqrt n$ by the Lipschitz condition.
\end{proof}

\begin{remark}[Generality of the MSE-to-welfare bridge]
Lemma~\ref{lem:plugin_regret_mse} is stated at an abstract level.  For stable plug-in welfare objectives, latent MSE controls expected welfare approximation error, while decision-specific welfare analysis remains application-specific.  Related
empirical Bayes work makes this bridge in more specialized decision problems,
including downstream regret in CLOSE \citep{chenEmpiricalBayesWhen2024} and
direct welfare optimization for compound selection in ASSURE
\citep{chenCompoundSelectionDecisions2025}.  The present paper keeps MSE as the
common estimation target because the same latent mobility estimates can feed
many downstream analyses.
\end{remark}

\subsection{Regularity Framework and Sobolev Extensions}\label{app:general_regularity}

The pointwise regularity conditions in Assumption~\ref{asm:regularity} can be replaced by moment-based derivative conditions.
Assumption~\ref{asm:regularity} requires uniform-in-$y$ polynomial-growth bounds for the shrinkage adjustment
$g(y)=f(y)-y$, its Jacobian $Dg$, and the pairwise differences
$g_\gamma-g_{\gamma'}$ that control variation across the parameterized class.
For a probability measure $P$ on $\bbR^n$, the \emph{Sobolev space}
$W^{k,p}(P)$ consists of functions $g \colon \bbR^n \to \bbR^n$ whose
derivative arrays up to order $k$ have finite $L^p(P)$ norm:
\[
\|g\|_{W^{k,p}(P)} :=
\sum_{m=0}^k \|D^m g\|_{L^p(P)} .
\]
Here $D^0g=g$, measured in Euclidean norm; $D^1g=Dg$, measured in Frobenius norm; and for $m\geq2$, $D^m g$ is the array of all $m$th-order partial derivatives, measured in Hilbert--Schmidt norm, meaning the square root of the sum of squared entries of the derivative array. At $k=1$, this gives $\|g\|_{W^{1,p}(P)} = \|g\|_{L^p(P)} + \|Dg\|_{L^p(P)}$. The pointwise norm $\|g(y)\|_W = \|g(y)\|_2 + \|Dg(y)\|_F$ from Section~\ref{sec:regularity} satisfies $\|g\|_{W^{1,p}(P)} \leq 2\big\|\|g(\cdot)\|_W\big\|_{L^p(P)}$, so a pointwise bound on $\|g(y)\|_W$ immediately implies $W^{1,p}$ membership. The following lemma records a common $\sqrt n$ scaling case for fixed linear smoothers.

\begin{lemma}[Row-norm scaling for linear smoothers]\label{lem:row_norm_linear_smoothers}
Let $f(y)=Sy$ and $g(y)=f(y)-y=(S-I)y$. Suppose that, for constants $C_{\mathrm{op}},C_{\mathrm{row}}<\infty$ not depending on $n$,
\[
\|S\|_{\mathrm{op}}\leq C_{\mathrm{op}},
\qquad
\max_{i\leq n}\|S_{i\cdot}\|_2\leq C_{\mathrm{row}} .
\]
Then the singleton class $\mathcal F=\{f\}$ satisfies Assumption~\ref{asm:regularity} with $\beta=1/2$ and $\nu_n=C\sqrt n$ for a constant $C<\infty$ not depending on $n$.
\end{lemma}

\begin{proof}[Proof of Lemma~\ref{lem:row_norm_linear_smoothers}]
Take $\Gamma$ to be a singleton, so the parameter-increment term in Assumption~\ref{asm:regularity} is zero. The map $g$ is linear and continuously differentiable with $Dg=S-I$. The operator-norm bound gives
\[
\|g(y)\|_2\leq\|S-I\|_{\mathrm{op}}\|y\|_2
\leq
(C_{\mathrm{op}}+1)\|y\|_2
\]
for all $y\in\bbR^n$. The row-norm bound gives, writing $e_i$ for the $i$th coordinate vector,
\[
\|(S-I)_{i\cdot}\|_2\leq \|S_{i\cdot}\|_2+\|e_i\|_2\leq C_{\mathrm{row}}+1.
\]
Therefore
\[
\|Dg(y)\|_F^2
=
\sum_{i=1}^n \|(S-I)_{i\cdot}\|_2^2
\leq
n(C_{\mathrm{row}}+1)^2 .
\]
Thus
\[
\|g(y)\|_2+\|Dg(y)\|_F
\leq
(C_{\mathrm{op}}+1)\|y\|_2+(C_{\mathrm{row}}+1)\sqrt n
\leq
(C_{\mathrm{op}}+C_{\mathrm{row}}+2)\sqrt n
\left(1+\frac{\|y\|_2}{\sqrt n}\right),
\]
so the pointwise envelope in Assumption~\ref{asm:regularity} holds with $\beta=1/2$ and $\nu_n=(C_{\mathrm{op}}+C_{\mathrm{row}}+2)\sqrt n$.
\end{proof}

\begin{definition}[Sobolev moment envelope]\label{def:sobolev_moment_envelope}
Under Assumption~\ref{asm:sampling_array}, write $P_Y=\N(\theta,\Sigma)$ for the law of $Y$. For $g:\bbR^n\to\bbR^n$, $\beta\geq0$, and a nonnegative integer $k$, define
\[
M_{P_Y, k}^\beta(g) := \sup_{p \geq 2} p^{-\beta}\|g\|_{W^{k+1,p}(P_Y)},
\]
with the convention that $M_{P_Y,k}^{\beta}(g)=\infty$ if
$g\notin W^{k+1,p}(P_Y)$ for some $p\geq2$.
\end{definition}

Definition~\ref{def:sobolev_moment_envelope} packages the derivative-moment
bounds the proof needs.  Its index $k$ is shifted by one derivative: $k=0$
corresponds to $W^{1,p}$ control of $g$ and $Dg$, while $k\geq1$ adds
$L^p(P_Y)$ control of higher derivative arrays.  The general concentration
theorem, Theorem~\ref{thm:general_conc}, is stated in terms of this envelope
because the proof needs $L^p(P_Y)$ derivative bounds, not necessarily uniform
pointwise bounds in $y$.  This Sobolev formulation is more general than the
pointwise polynomial envelope because the needed derivative-moment bounds may
hold even when no convenient uniform polynomial bound is available.
Lemma~\ref{lem:pointwise_regular_to_sobolev} shows that
Assumption~\ref{asm:regularity} implies the $k=0$ Sobolev moment condition,
and this $k=0$ implication is the route used to prove
Theorem~\ref{thm:uniform_conc} through Theorem~\ref{thm:general_conc}.

\begin{assumption}[Sobolev moment regularity]\label{asm:sobolev_reg}
The candidate class is $\mathcal F=\{f_\gamma:\gamma\in\Gamma\}$, where $\Gamma\subset\bbR^{d_\Gamma}$ is compact and
\[
\operatorname{diam}(\Gamma)
:=
\sup_{\gamma,\gamma'\in\Gamma}\|\gamma-\gamma'\|_2
\leq D_\Gamma
\]
for a constant $D_\Gamma<\infty$ not depending on $n$. Write $g_\gamma(y):=f_\gamma(y)-y$. With $M_{P_Y,k}^{\beta}$ as in Definition~\ref{def:sobolev_moment_envelope}, there exist $\beta \geq 0$, a nonnegative integer $k$, a scaling sequence $\nu_n > 0$, and a point $\gamma_0\in\Gamma$ such that
\[
M_{P_Y,k}^{\beta}(g_{\gamma_0})
+
\sup_{\gamma\ne\gamma'}
\frac{M_{P_Y,k}^{\beta}(g_\gamma-g_{\gamma'})}{\|\gamma-\gamma'\|_2}
\leq
\nu_n .
\]
The maps $g_{\gamma_0}$ and $g_\gamma-g_{\gamma'}$, $\gamma\ne\gamma'$, have continuous partial derivatives through order $k+1$. If $\Gamma$ is a singleton, the supremum is interpreted as zero.
\end{assumption}

\begin{lemma}[Pointwise polynomial regularity implies Sobolev moment regularity]\label{lem:pointwise_regular_to_sobolev}
Under Assumption~\ref{asm:sampling_array}, Assumption~\ref{asm:regularity}
implies Assumption~\ref{asm:sobolev_reg} at $k=0$, up to constants depending
only on fixed sampling and envelope constants.
\end{lemma}

\begin{proof}[Proof of Lemma~\ref{lem:pointwise_regular_to_sobolev}]
Throughout this proof, $\lesssim$ hides constants that may depend on
$\beta$, $C_\theta$, and $\bar\sigma$, but not on
$p$, $n$, $\gamma,\gamma'$, $d_\Gamma$, or $\nu_n$.  By
Assumption~\ref{asm:regularity}, for every $y\in\bbR^n$ and
$\gamma\ne\gamma'$,
\[
\|g_\gamma(y)-g_{\gamma'}(y)\|_W
\leq
\nu_n\left(1+\frac{\|y\|_2}{\sqrt n}\right)^{2\beta}
\|\gamma-\gamma'\|_2 .
\]

The only probabilistic input is a Gaussian moment bound for the polynomial
factor.  Write $Y=\theta+\Sigma^{1/2}Z$, where $Z\sim\N(0,I_n)$.  By
Minkowski's inequality and the triangle inequality,
\[
\left\|1+\frac{\|Y\|_2}{\sqrt n}\right\|_{L^q(P_Y)}
\leq
1+\frac{\|\theta\|_2}{\sqrt n}
+\frac{\|\Sigma^{1/2}\|_{\mathrm{op}}}{\sqrt n}
\|\,\|Z\|_2\,\|_{L^q}.
\]
Assumption~\ref{asm:sampling_array} gives
$\|\theta\|_2/\sqrt n\leq C_\theta$ and
$\|\Sigma^{1/2}\|_{\mathrm{op}}\leq\bar\sigma$.  The Euclidean norm is
one-Lipschitz by the reverse triangle inequality: for any $z,z'\in\bbR^n$,
\[
\big|\|z\|_2-\|z'\|_2\big|
\leq
\|z-z'\|_2 .
\]
Also, by Cauchy--Schwarz,
\[
\E[\|Z\|_2]
\leq
\left(\E[\|Z\|_2^2]\right)^{1/2}
=
\sqrt n,
\]
where the equality uses $\|Z\|_2^2\sim\chi_n^2$, so
$\E[\|Z\|_2^2]=n$.

Apply the Gaussian concentration inequality
\citep[Theorem~5.2.2]{vershyninHighdimensionalProbabilityIntroduction2018} to
the one-Lipschitz function $z\mapsto\|z\|_2$.  This gives a universal
sub-Gaussian bound on
$\|Z\|_2-\E[\|Z\|_2]$.  By the tail characterization in
\citet[Proposition~2.5.2]{vershyninHighdimensionalProbabilityIntroduction2018},
for every $t\geq0$,
\[
\Pr\left(\big|\|Z\|_2-\E[\|Z\|_2]\big|\geq t\right)
\leq
2\exp(-ct^2)
\]
for a universal constant $c>0$.  The moment characterization in the same
proposition implies
\[
\left\|\|Z\|_2-\E[\|Z\|_2]\right\|_{L^q}
\lesssim
\sqrt q,\qquad q\geq1.
\]
Therefore, by the triangle inequality in $L^q$, applied to
\[
\|Z\|_2
=
\E[\|Z\|_2]+\{\|Z\|_2-\E[\|Z\|_2]\},
\]
we have
\[
\|\,\|Z\|_2\,\|_{L^q}
\leq
\|\E[\|Z\|_2]\|_{L^q}
+
\left\|\|Z\|_2-\E[\|Z\|_2]\right\|_{L^q}.
\]
The first term is just the constant $\E[\|Z\|_2]$, so
\[
\|\,\|Z\|_2\,\|_{L^q}
\lesssim
\sqrt n+\sqrt q,\qquad q\geq1.
\]
Therefore
\[
\left\|1+\frac{\|Y\|_2}{\sqrt n}\right\|_{L^q(P_Y)}
\lesssim
1+\sqrt q,\qquad q\geq1.
\]
Using this bound with $q=\max\{1,2\beta p\}$ and monotonicity of $L^q$
norms gives, for $p\geq2$,
\[
\left\|\left(1+\frac{\|Y\|_2}{\sqrt n}\right)^{2\beta}\right\|_{L^p(P_Y)}
\lesssim
p^\beta ,
\]
with the case $\beta=0$ understood as the constant-one bound.

Taking $L^p(P_Y)$ norms of the pointwise increment bound therefore gives,
for every $p\geq2$,
\[
\|g_\gamma-g_{\gamma'}\|_{W^{1,p}(P_Y)}
\lesssim
\nu_n p^\beta\|\gamma-\gamma'\|_2 .
\]
Thus
\[
M_{P_Y,0}^{\beta}(g_\gamma-g_{\gamma'})
\lesssim
\nu_n\|\gamma-\gamma'\|_2,
\]
and $g_\gamma-g_{\gamma'}\in W^{1,p}(P_Y)$ for all $p\geq2$.  The case
$\gamma=\gamma'$ is the zero map.

The same argument applied to the reference-point bound
\[
\|g_{\gamma_0}(y)\|_W
\leq
\nu_n\left(1+\frac{\|y\|_2}{\sqrt n}\right)^{2\beta}
\]
shows that $g_{\gamma_0}\in W^{1,p}(P_Y)$ for all $p\geq2$ and
\[
M_{P_Y,0}^{\beta}(g_{\gamma_0})\lesssim \nu_n.
\]
Combining the reference-point and increment bounds gives
\[
M_{P_Y,0}^{\beta}(g_{\gamma_0})
+
\sup_{\gamma\ne\gamma'}
\frac{M_{P_Y,0}^{\beta}(g_\gamma-g_{\gamma'})}{\|\gamma-\gamma'\|_2}
\lesssim
\nu_n.
\]
This is Assumption~\ref{asm:sobolev_reg} at $k=0$, up to constants depending
only on fixed sampling and envelope constants.
\end{proof}

For fixed $k$, once the corresponding higher-derivative moment bounds are
verified with scale $\nu_n$, the dimension exponent in
Theorem~\ref{thm:general_conc} is $1+3\cdot 2^{-k}+\beta$. This equals
$4+\beta$ at $k=0$ (the case delivered by
Lemma~\ref{lem:pointwise_regular_to_sobolev} and used to prove
Theorem~\ref{thm:uniform_conc}) and decreases toward $1+\beta$ as $k$ grows.

\begin{theorem}[General concentration]\label{thm:general_conc}
Suppose Assumptions~\ref{asm:sampling_array} and~\ref{asm:sobolev_reg} hold.
Assume the \sure{} error process $f\mapsto\sure_n(f)-L_n(f)$ is separable and the minimizers
\[
\hat\gamma\in\argmin_{\gamma\in\Gamma}\sure_n(f_\gamma),
\qquad
\gamma^*\in\argmin_{\gamma\in\Gamma}L_n(f_\gamma)
\]
are measurable in $Y$. Set $\hat f:=f_{\hat\gamma}$ and $f^*:=f_{\gamma^*}$, with $f^*$ the realized-loss oracle. Then
\[
\E\left[\sup_{f\in \mathcal{F}}|\sure_n(f) - L_n(f)|\right] \lesssim \frac{1}{\sqrt{n}} + \frac{\nu_n \, \max\{d_\Gamma,1\}^{1+3\cdot 2^{-k}+\beta}}{n}.
\]
The corresponding oracle comparison satisfies
\[
\E[L_n(\hat{f}) - L_n(f^*)]
\lesssim
\frac{\nu_n \, \max\{d_\Gamma,1\}^{1+3\cdot 2^{-k}+\beta}}{n}.
\]
The implicit constants may depend on fixed $k$, $\beta$, $C_\theta$,
$\bar\sigma$, and the diameter bound $D_\Gamma$, but not on $n$,
$d_\Gamma$, $\nu_n$, or the particular candidate class.
\end{theorem}

\noindent The proof, together with its supporting lemmas, occupies the remainder of this subsection; the \hyperref[proof:general_conc]{main argument} concludes it.

\subsubsection{Proof of General Concentration}\label{app:proofs}

The proof rewrites the \sure{} error as a Gaussian divergence functional.  The key analytic input is a continuity bound for this divergence operator: if the shrinkage adjustment and its derivatives have controlled Sobolev moments, then the \sure{} error has controlled sub-Weibull tails \citep{vladimirovaSubWeibullDistributionsGeneralizing2020}.  The remaining steps are a covering argument over the compact parameter space and a centering argument at the reference point.

\begin{lemma}[Tail-to-expectation conversion]\label{lem:tail_to_expectation}
Let $Z \geq 0$ satisfy $\Pr(Z \geq m + t) \leq 2\exp(-(t/s)^\alpha)$ for all $t \geq 0$, some $0\leq m < \infty$, $s > 0$, and $\alpha \in (0,1]$. Then $\E[Z] \leq m + 2s\,\Gamma(1 + 1/\alpha)$.
\end{lemma}

\begin{proof}[Proof of Lemma~\ref{lem:tail_to_expectation}]
Writing the expectation as an integral of tail probabilities above the level $m$,
\[
\E[Z]
=
\int_0^\infty \Pr(Z\geq u)\,du
\leq
m+\int_0^\infty \Pr(Z\geq m+t)\,dt .
\]
Using the assumed tail bound,
\[
\E[Z]\leq m+2\int_0^\infty e^{-(t/s)^\alpha}\,dt.
\]
With the change of variables $u=(t/s)^\alpha$, so
$dt=s\alpha^{-1}u^{1/\alpha-1}\,du$,
\[
\int_0^\infty e^{-(t/s)^\alpha}\,dt
=
\frac{s}{\alpha}\int_0^\infty e^{-u}u^{1/\alpha-1}\,du
=
s\,\Gamma(1+1/\alpha),
\]
where the last equality uses $\Gamma(1+x)=x\Gamma(x)$.  This gives the stated
bound.
\end{proof}

For $\alpha\in(0,1]$, use the $\psi_\alpha$ Orlicz quasi-norm
\citep{vershyninHighdimensionalProbabilityIntroduction2018}:
\[
\lVert X\rVert_{\psi_\alpha}
:=
\inf\left\{
s>0:
\E\exp\left[\left(\frac{|X|}{s}\right)^\alpha\right]\leq 2
\right\}.
\]

\begin{proposition}[Chaining bound for $\psi_\alpha$ processes]\label{prop:chaining}
Let $\{X_\gamma : \gamma \in \Gamma\}$ be a stochastic process indexed by $\Gamma \subset \bbR^{d_\Gamma}$ with $\mathrm{diam}(\Gamma) < \infty$, containing a reference point $\gamma_0$ with $X_{\gamma_0} = 0$ a.s. Suppose there exist $\alpha \in (0,1]$ (the $\psi_\alpha$ tail exponent) and $L > 0$ such that
\[
\|X_\gamma - X_{\gamma'}\|_{\psi_\alpha} \leq L\,\|\gamma - \gamma'\|_2 \quad \text{for all } \gamma, \gamma' \in \Gamma.
\]
Assume the process is separable \citep[\S2.3.3]{vaartWeakConvergenceEmpirical2023}: there is a countable $\Gamma_0\subseteq\Gamma$ such that
$\sup_{\gamma\in\Gamma}|X_\gamma|=\sup_{\gamma\in\Gamma_0}|X_\gamma|$ and
$\sup_{\gamma,\gamma'\in\Gamma}|X_\gamma-X_{\gamma'}|=\sup_{\gamma,\gamma'\in\Gamma_0}|X_\gamma-X_{\gamma'}|$
almost surely. The two suprema are therefore measurable. Then
\[
\E[\sup_\gamma |X_\gamma|]
\leq
C_\alpha\,L\,\mathrm{diam}(\Gamma)\,\max\{d_\Gamma,1\}^{1/\alpha},
\]
where $C_\alpha$ depends only on $\alpha$.
\end{proposition}

\begin{proof}[Proof of Proposition~\ref{prop:chaining}]
The $\psi_\alpha$-Lipschitz condition defines the metric $\rho(\gamma,\gamma') := L\|\gamma - \gamma'\|_2$ with diameter $\Delta := L\,\mathrm{diam}(\Gamma)$. If $\Delta=0$, then $X_\gamma=0$ for all $\gamma\in\Gamma$ and the result is immediate. Hence take $\Delta>0$.

For $\alpha=1$ the function $\psi_1(x)=e^x-1$ is convex; set $\tilde\psi_1:=\psi_1$. For $\alpha\in(0,1)$, $\psi_\alpha$ is not convex near the origin and $\|\cdot\|_{\psi_\alpha}$ is only a quasi-norm; however, there is a convex, nondecreasing $\tilde\psi_\alpha$ with $\tilde\psi_\alpha(0)=0$ and
\[
\|X\|_{\tilde\psi_\alpha}\leq\|X\|_{\psi_\alpha}\leq \kappa_\alpha\|X\|_{\tilde\psi_\alpha}
\quad\text{for every random variable }X
\]
\citep[Problem~2.14.1]{vaartWeakConvergenceEmpirical2023}, so the two norms may be used interchangeably up to $\alpha$-dependent constants. In particular, since $\tilde\psi_\alpha\leq\psi_\alpha$ pointwise, the first inequality holds with constant one, and the increments satisfy $\|X_\gamma-X_{\gamma'}\|_{\tilde\psi_\alpha}\leq\rho(\gamma,\gamma')$.

The function $\tilde\psi_\alpha$ also satisfies the growth condition required below:
\[
\limsup_{x,y\to\infty}\tilde\psi_\alpha(x)\,\tilde\psi_\alpha(y)\big/\tilde\psi_\alpha\bigl(2^{1/\alpha}xy\bigr)<\infty.
\]
Indeed, the construction of \citet[Problem~2.14.1]{vaartWeakConvergenceEmpirical2023} gives $\tilde\psi_\alpha(x)=e^{x^\alpha-c_\alpha^\alpha}-1$ for $x\geq c_\alpha:=((1-\alpha)/\alpha)^{1/\alpha}$ and $\tilde\psi_\alpha\equiv0$ on $[0,c_\alpha]$ ($c_\alpha$ is the inflection point of $e^{x^\alpha}$, so $\tilde\psi_\alpha$ is convex), so for $x,y\geq\max\{c_\alpha,1\}$,
\[
\tilde\psi_\alpha(x)\,\tilde\psi_\alpha(y)
\leq
e^{x^\alpha+y^\alpha-2c_\alpha^\alpha}
\leq
e^{x^\alpha+y^\alpha},
\]
while $(2^{1/\alpha}xy)^\alpha=2(xy)^\alpha$ exactly, so that, once $2(xy)^\alpha-c_\alpha^\alpha\geq\log2$ (using $e^z-1\geq\tfrac12 e^z$ for $z\geq\log2$),
\[
\tilde\psi_\alpha\bigl(2^{1/\alpha}xy\bigr)
=
e^{2(xy)^\alpha-c_\alpha^\alpha}-1
\geq
\tfrac12\,e^{-c_\alpha^\alpha}\,e^{2(xy)^\alpha}.
\]
Because $x^\alpha+y^\alpha\leq2(xy)^\alpha$ for $x,y\geq1$ (as $y\geq1$ implies $x^\alpha\leq(xy)^\alpha$ and symmetrically), the exponent in the ratio of the two displays is nonpositive, whence
\[
\limsup_{x,y\to\infty}
\frac{\tilde\psi_\alpha(x)\,\tilde\psi_\alpha(y)}{\tilde\psi_\alpha(2^{1/\alpha}xy)}
\leq
2\,e^{c_\alpha^\alpha}<\infty .
\]
For $\alpha=1$ the same computation applies with $c_1:=0$.

\citet[Corollary~4.2.13]{vershyninHighdimensionalProbabilityIntroduction2018} bounds the covering numbers of a Euclidean ball, giving $N(u;\,\Gamma,\,\rho) \leq (C\Delta/u)^{d_\Gamma}$ for $u \in (0, \Delta]$, hence packing numbers $D(u;\Gamma,\rho)\leq N(u/2;\Gamma,\rho)\leq(2C\Delta/u)^{d_\Gamma}$. The maximal inequality for processes with Orlicz-Lipschitz increments \citep[Theorem~2.2.4 and Corollary~2.2.5]{vaartWeakConvergenceEmpirical2023}, applied with the convex function $\tilde\psi_\alpha$ (the process is separable by hypothesis), gives
\[
\Bigl\|\sup_{\gamma,\gamma'}|X_\gamma-X_{\gamma'}|\Bigr\|_{\tilde\psi_\alpha}
\leq
K_\alpha\int_0^{\Delta}\tilde\psi_\alpha^{-1}\bigl(D(u;\Gamma,\rho)\bigr)\,du ,
\]
where $K_\alpha$ is the constant of \citet[Theorem~2.2.4]{vaartWeakConvergenceEmpirical2023}, which depends only on the Orlicz function and the Lipschitz constant of the increments --- here $\tilde\psi_\alpha$ and constant one, since $L$ is absorbed into $\rho$ --- hence only on $\alpha$. From the explicit form $\tilde\psi_\alpha(t)=e^{t^\alpha-c_\alpha^\alpha}-1$ for $t\geq c_\alpha$, the inverse is $\tilde\psi_\alpha^{-1}(x)=\{\log(1+x)+c_\alpha^\alpha\}^{1/\alpha}$ for $x>0$; since $\log(1+x)\geq\log2$ for $x\geq1$,
\[
\tilde\psi_\alpha^{-1}(x)
\leq
C_\alpha\{\log(1+x)\}^{1/\alpha},
\qquad x\geq1,
\qquad
C_\alpha:=\Bigl(1+\frac{c_\alpha^\alpha}{\log2}\Bigr)^{1/\alpha}.
\]
Substituting $u = \Delta v$ and bounding $\log(1 + (2C/v)^{d_\Gamma}) \leq \max\{d_\Gamma,1\}\log(4C/v)$ for $v \in (0,1]$ yields
\begin{align*}
\int_0^{\Delta}\tilde\psi_\alpha^{-1}\bigl(D(u;\Gamma,\rho)\bigr)\,du
&\leq
C_\alpha\,\Delta\,\max\{d_\Gamma,1\}^{1/\alpha}\int_0^1\{\log(4C/v)\}^{1/\alpha}\,dv \\
&\leq
C_\alpha'\,\Delta\,\max\{d_\Gamma,1\}^{1/\alpha},
\end{align*}
the last integral converging and depending only on $\alpha$.

It remains to convert the $\tilde\psi_\alpha$-norm bound on the supremum of the \emph{increments} into the expectation bound on $\sup_\gamma|X_\gamma|$ asserted in the proposition. Write $W:=\sup_{\gamma,\gamma'}|X_\gamma-X_{\gamma'}|$ and $s:=\|W\|_{\tilde\psi_\alpha}$; combining the two preceding displays,
\[
s\leq K_\alpha C_\alpha'\,\Delta\,\max\{d_\Gamma,1\}^{1/\alpha}.
\]
First, the Orlicz norm controls the mean. By definition of the norm, $\E[\tilde\psi_\alpha(W/s)]\leq1$; Jensen's inequality for the convex function $\tilde\psi_\alpha$ gives $\tilde\psi_\alpha(\E[W]/s)\leq\E[\tilde\psi_\alpha(W/s)]\leq1$, and applying the increasing function $\tilde\psi_\alpha^{-1}$ to both sides yields $\E[W]\leq\tilde\psi_\alpha^{-1}(1)\,s$. Second, the supremum of the process is dominated by the supremum of its increments: the proposition assumes a reference point $\gamma_0\in\Gamma$ with $X_{\gamma_0}=0$ almost surely, so
$\sup_\gamma|X_\gamma|=\sup_\gamma|X_\gamma-X_{\gamma_0}|\leq W$ almost surely. Combining the three bounds,
\[
\E\Bigl[\sup_\gamma|X_\gamma|\Bigr]
\leq
\E[W]
\leq
\tilde\psi_\alpha^{-1}(1)\,K_\alpha C_\alpha'\,\Delta\,\max\{d_\Gamma,1\}^{1/\alpha}
=
C_\alpha''\,L\,\mathrm{diam}(\Gamma)\,\max\{d_\Gamma,1\}^{1/\alpha},
\]
where $C_\alpha'':=\tilde\psi_\alpha^{-1}(1)\,K_\alpha C_\alpha'$ depends only on $\alpha$ and the final equality substitutes $\Delta=L\,\mathrm{diam}(\Gamma)$. This is the claimed bound, with the statement's constant $C_\alpha:=C_\alpha''$.
\end{proof}

\begin{lemma}[Representation of $\sure_n(f) - L_n(f)$]\label{lem:representation}
Let $f(y)=y+g(y)$, where $g$ is weakly differentiable,
$g(Y)\in L^2(P_Y)$, and
$\E[\sum_{i,j}|\Sigma_{ij}\partial_j g_i(Y)|]<\infty$, where
$\partial_j g_i$ denotes the weak derivative.  Writing $\varepsilon := Y - \theta$:
\[
\sure_n(f) - L_n(f) = \frac{1}{n}\bigl[\tr(\Sigma) - \|\varepsilon\|_2^2\bigr] + \frac{2}{n}\,\Psi(g),
\]
where $\Psi(g) := \tr\{\Sigma\,Dg(Y)\} - \langle \varepsilon,\, g(Y)\rangle$. Moreover, $\E[\Psi(g)] = 0$ by Stein's lemma.
\end{lemma}

\begin{proof}[Proof of Lemma~\ref{lem:representation}]
Write $f(y) = y + g(y)$, so $Y - f(Y) = -g(Y)$ and $\theta - f(Y) = -\varepsilon - g(Y)$. Expanding:
\begin{align*}
n\,\sure_n(f) &= \|g(Y)\|_2^2 - \tr(\Sigma) + 2\tr(\Sigma(I + Dg(Y))) \\
&= \|g(Y)\|_2^2 + \tr(\Sigma) + 2\tr\{\Sigma Dg(Y)\}, \\
n\,L_n(f) &= \|\varepsilon + g(Y)\|_2^2 = \|\varepsilon\|_2^2 + 2\langle \varepsilon, g(Y)\rangle + \|g(Y)\|_2^2.
\end{align*}
The $\|g(Y)\|_2^2$ terms cancel in the difference, giving the representation.
The derivative term in $\Psi(g)$ is integrable by assumption, and
$\langle\varepsilon,g(Y)\rangle$ is integrable by Cauchy--Schwarz because
$\varepsilon$ is Gaussian and $g(Y)\in L^2(P_Y)$.  Stein's lemma for weak
derivatives gives
\[
\E[\varepsilon_j g_j(Y)]
=
\sum_l \Sigma_{jl}\,\E[\partial_l g_j(Y)].
\]
Summing over $j$ and using symmetry of $\Sigma$ and the row-Jacobian convention gives
$\E[\langle\varepsilon,g(Y)\rangle]=\E[\tr\{\Sigma Dg(Y)\}]$, and hence
$\E[\Psi(g)]=0$.
\end{proof}

For the remaining lemmas, let $Z\sim \N(0,I_n)$, and let $P_Z$ denote its
law.  For a fixed $p\geq2$, write
$\|\cdot\|_p=\|\cdot\|_{L^p(P_Z)}$.  All derivative arrays are vectorized before
taking Euclidean norms.  In particular, matrix norms below are Frobenius norms unless an operator norm is explicitly indicated.

Write $D$ for differentiation with respect to $z$, and write $D^m h$ for
the full array of $m$th-order partial derivatives of $h$.  For $m\geq1$,
let
\[
\ell=(i_1,\ldots,i_m)\in\{1,\ldots,n\}^m
\]
denote a derivative label.  For $m=0$, use one empty derivative label
$\ell$, so $(D^0h)_{\ell,j}=h_j$.  For $m\geq1$,
\[
(D^m h)_{\ell,j}(z)
=
\frac{\partial^m h_j(z)}
{\partial z_{i_1}\cdots \partial z_{i_m}}.
\]
When differentiating $(D^m h)_{\ell,j}$ once more with respect to $z_i$,
the corresponding derivative label for $D^{m+1}h$ is
$(i,i_1,\ldots,i_m)$.
Define $\delta(D^m h)$ by applying the divergence over the original output
coordinate $j$, while carrying the derivative label $\ell$ as an array
entry:
\[
\{\delta(D^m h)(z)\}_\ell
=
\sum_{j=1}^n z_j(D^m h)_{\ell,j}(z)
-
\sum_{j=1}^n
\frac{\partial}{\partial z_j}(D^m h)_{\ell,j}(z).
\]
Thus $\delta(D^0h)=\delta(h)$.  For array-valued derivatives, the Sobolev
norm is
\[
\|D^m h\|_{W^{r,p}(P_Z)}
=
\sum_{s=0}^{r}
\|D^{m+s}h\|_{L^p(P_Z)},
\]
where each $L^p(P_Z)$ norm uses the Euclidean norm after vectorizing the
corresponding derivative array.

Throughout the following lemmas we also use the Wiener-chaos decomposition
$L^2(P_Z)=\bigoplus_{r\geq0}\mathcal H_r$, writing $J_r$ for the orthogonal
projection onto chaos order $r$ and $T_t=\sum_{r\geq0}e^{-rt}J_r$ for the
Ornstein--Uhlenbeck semigroup \citep[Ch.~1]{nualartMalliavinCalculusRelated2006};
these conventions apply entrywise to array-valued functions.

\begin{lemma}[Nualart constant bound for $N=2$]
\label{lem:nualart_K_N2_bound}
Fix $p>1$.  Let $K(p,2)$ denote the constant in
\citet[Lemma~1.4.1]{nualartMalliavinCalculusRelated2006} for exponent $p$
and $N=2$.  Then
\[
K(p,2)\leq c p^2
\qquad\text{if }p\geq2,
\]
and
\[
K(p,2)
\leq
c\left(\frac{p}{p-1}\right)^2
\qquad\text{if }1<p<2,
\]
where $c$ is a universal constant.
\end{lemma}

\begin{proof}
First suppose $p\geq2$.  In the proof of Nualart's Lemma~1.4.1, $t_0$ is
chosen so that
\[
p=e^{2t_0}+1,
\qquad\text{hence}\qquad
e^{2t_0}=p-1.
\]
For $t\geq t_0$, the proof gives
\[
K=e^{Nt_0}.
\]
For $t<t_0$, the proof gives
\[
K=Ne^{2Nt_0}+e^{Nt_0}.
\]
With $N=2$, this means
\[
\begin{aligned}
K(p,2)
&\leq
\max\{e^{2t_0},\,2e^{4t_0}+e^{2t_0}\}  \\
&\leq
2e^{4t_0}+e^{2t_0}  \\
&=
2(p-1)^2+(p-1)  \\
&\leq
c p^2.
\end{aligned}
\]
The last inequality absorbs $2(p-1)^2+(p-1)$ into $c p^2$ for a universal
constant $c$, using $p\geq2$.

Now suppose $1<p<2$, and set $q=p/(p-1)>2$.  \citet{nualartMalliavinCalculusRelated2006} obtains this case by
duality.  The second equality below is the duality step used in
\citet[Exercise~1.4.5]{nualartMalliavinCalculusRelated2006}.  Explicitly,
\[
\begin{aligned}
\|T_t(I-J_0-J_1)G\|_p
&=
\sup_{\|H\|_q\leq1}
\left|
\E\langle T_t(I-J_0-J_1)G,H\rangle
\right|  \\
&=
\sup_{\|H\|_q\leq1}
\left|
\E\langle G,T_t(I-J_0-J_1)H\rangle
\right|  \\
&\leq
\sup_{\|H\|_q\leq1}
\|G\|_p
\|T_t(I-J_0-J_1)H\|_q  \\
&\leq
K(q,2)e^{-2t}\|G\|_p .
\end{aligned}
\]
By the $q>2$ case already proved,
\[
K(q,2)\leq c q^2.
\]
Therefore
\[
K(p,2)
\leq
c q^2
=
c\left(\frac{p}{p-1}\right)^2.
\]
Although the proof is stated for polynomial random variables, Exercise~1.4.6 of
\citet{nualartMalliavinCalculusRelated2006} extends the argument to
Hilbert-valued random variables, which covers the finite-dimensional Euclidean
arrays used here.
\end{proof}

\begin{lemma}[Boundedness of the multiplier operator $\sqrt R$]
\label{lem:sqrt_R_bounded}
Fix $p\geq2$, and set $q=p/(p-1)$.  Let $V$ be a finite-dimensional
Euclidean array space, and let $\widetilde G\in L^q(P_Z;V)$.  On polynomial
inputs, let $R$ denote the multiplier operator with multiplier $r/(r-1)$ on
chaos order $r\geq2$, and with multiplier zero on chaos orders $0$ and $1$.
Then $\sqrt R$ extends to a bounded operator on $L^q(P_Z;V)$, and this
extension satisfies
\[
\|\sqrt R\,\widetilde G\|_q
=
\|\sqrt R(I-J_0-J_1)\widetilde G\|_q
\leq
c p^2\|\widetilde G\|_q.
\]
\end{lemma}

\begin{proof}
We use the framework of \citet[Theorem~1.4.2]{nualartMalliavinCalculusRelated2006}
to bound $\sqrt R$.  On chaos order $r\geq2$, the operator $\sqrt R$
corresponds to the multiplier sequence
\[
\phi(r) = \sqrt{\frac{r}{r-1}} = \left(1-\frac{1}{r}\right)^{-1/2}.
\]
To verify the multiplier conditions in \citet[Theorem~1.4.2]{nualartMalliavinCalculusRelated2006},
we expand the scalar function
$\varphi(z) = (1-z)^{-1/2}$ in a Taylor series around $z=0$:
\[
(1-z)^{-1/2} = 1 + \frac{1}{2}z + \frac{3}{8}z^2 + \cdots = \sum_{k=0}^\infty a_k z^k.
\]
Substituting $z=1/r$,
\[
\phi(r) = \sum_{k=0}^\infty a_k r^{-k}.
\]

We first prove the bound for polynomial array-valued $\widetilde G$.  Since
$R$ has multiplier zero on chaos orders $0$ and $1$,
\[
\sqrt R\,\widetilde G
=
\sqrt R(I-J_0-J_1)\widetilde G.
\]
Define
\[
S_0\widetilde G
:=
(I-J_0-J_1)\widetilde G.
\]
Since $p\geq2$, $q=p/(p-1)\leq2$.  Using \citet[Lemma~1.4.1]{nualartMalliavinCalculusRelated2006}
with $N=2$ and Lemma~\ref{lem:nualart_K_N2_bound}, with the endpoint
$q=2$ covered by its $p\geq2$ case,
\begin{align*}
    \|S_0\widetilde G\|_q
    &= \|(I - J_0 - J_1) \widetilde G\|_q\\
    &\leq K(q, 2)\|\widetilde G\|_q\\
    &\leq c\left(\frac{q}{q-1}\right)^2 \|\widetilde G\|_q\\
    &= c p^2 \|\widetilde G\|_q,
\end{align*}
where $c$ is a universal constant.  This gives the bound for $k=0$.
For $k\geq1$, define
\[
S_k\widetilde G
=
\frac{1}{(k-1)!}\int_0^\infty
t^{k-1}T_t(I-J_0-J_1)\widetilde G\,dt .
\]
The proof of \citet[Theorem~1.4.2]{nualartMalliavinCalculusRelated2006} gives
this formula for the operator with multiplier $r^{-k}$ on chaos order
$r\geq2$.  Thus, for polynomial $\widetilde G$,
\[
S_k\widetilde G
=
\sum_{r=2}^\infty r^{-k}J_r\widetilde G,
\]
where the sum is finite.
Using \citet[Lemma~1.4.1]{nualartMalliavinCalculusRelated2006} with $N=2$ and
Lemma~\ref{lem:nualart_K_N2_bound}, with the endpoint $q=2$ covered by
its $p\geq2$ case,
\[
\begin{aligned}
\|T_t(I-J_0-J_1)\widetilde G\|_q
&\leq
K(q,2)e^{-2t}\|\widetilde G\|_q \\
&\leq
c\left(\frac{q}{q-1}\right)^2e^{-2t}\|\widetilde G\|_q \\
&=
c p^2e^{-2t}\|\widetilde G\|_q.
\end{aligned}
\]
Using the integral expression for the operator $S_k$ and
Minkowski's inequality for integrals,
\[
\begin{aligned}
\|S_k\widetilde G\|_q
&=
\left\|
\frac{1}{(k-1)!}
\int_0^\infty t^{k-1}T_t(I-J_0-J_1)\widetilde G\,dt
\right\|_q  \\
&\leq
\frac{1}{(k-1)!}
\int_0^\infty
t^{k-1}\|T_t(I-J_0-J_1)\widetilde G\|_q\,dt  \\
&\leq
\frac{c p^2}{(k-1)!}
\int_0^\infty t^{k-1}e^{-2t}\,dt\,\|\widetilde G\|_q .
\end{aligned}
\]
The integral is a Gamma function integral.  Using the change of variables
$u=2t$, so $t=u/2$ and $dt=du/2$, gives
\[
\int_0^\infty t^{k-1}e^{-2t}dt = \frac{1}{2^k}\int_0^\infty u^{k-1}e^{-u}du = \frac{1}{2^k}\Gamma(k) = \frac{(k-1)!}{2^k},
\]
so
\[
\|S_k\widetilde G\|_q
\leq
\frac{c p^2}{(k-1)!}\frac{(k-1)!}{2^k}\|\widetilde G\|_q
=
c p^2 2^{-k}\|\widetilde G\|_q.
\]
Together with the $k=0$ calculation above, this gives
\[
\|S_k\widetilde G\|_q\leq c p^2 2^{-k}\|\widetilde G\|_q
\qquad\text{for all }k\geq0.
\]
For polynomial $\widetilde G$, write
\[
(I-J_0-J_1)\widetilde G
=
\sum_{r=2}^M J_r\widetilde G
\]
for some finite $M$.  Then
\[
\sqrt R\,\widetilde G
=
\sum_{r=2}^M
\sqrt{\frac r{r-1}}\,J_r\widetilde G
=
\sum_{r=2}^M
\left(\sum_{k=0}^\infty a_kr^{-k}\right)J_r\widetilde G
=
\sum_{k=0}^\infty a_kS_k\widetilde G,
\]
because, for each $r\geq2$,
\[
\sum_{k=0}^\infty a_kr^{-k}
=
\left(1-\frac1r\right)^{-1/2}
=
\sqrt{\frac r{r-1}},
\]
and both $\sqrt R$ and $S_k$ act as zero on chaos orders $0$ and $1$.
Since $a_k\geq0$, the triangle inequality and the preceding bound give
\[
\|\sqrt{R}\,\widetilde G\|_q
\leq
c p^2
\left(\sum_{k=0}^\infty \frac{a_k}{2^k}\right)\|\widetilde G\|_q.
\]
The sum in the parentheses is exactly the value of the original
function $\varphi(z)$ evaluated at $z=1/2$, so
\[
\sum_{k=0}^\infty \frac{a_k}{2^k}
=
\varphi(1/2)
=
\phi(2)
=
\sqrt{\frac{2}{2-1}}
=
\sqrt{2}.
\]
Thus, for polynomial array-valued $\widetilde G$,
\[
\|\sqrt{R}\,\widetilde G\|_q
\leq
c p^2\sqrt2\,\|\widetilde G\|_q .
\]
Theorem~1.4.2 of \citet{nualartMalliavinCalculusRelated2006} gives the
corresponding $L^q$ multiplier bound, and
\citet[Exercise~1.4.6]{nualartMalliavinCalculusRelated2006} extends the
multiplier theorem to Hilbert-valued random variables.  Since $V$ is
finite-dimensional Euclidean, this applies to $V$-valued $\widetilde G$ and
gives
\[
\|\sqrt{R}\,\widetilde G\|_q
\leq
c p^2\sqrt2\,\|\widetilde G\|_q .
\]
Since this extension of $\sqrt R$ has multiplier zero on chaos orders $0$
and $1$, and absorbing $\sqrt2$ into the numerical constant $c$,
\[
\|\sqrt R\,\widetilde G\|_q
=
\|\sqrt R(I-J_0-J_1)\widetilde G\|_q
\leq
c p^2\|\widetilde G\|_q .
\]
\end{proof}

\begin{lemma}[First-order Gaussian divergence-continuity bound]
\label{lem:first_order_divergence_continuity}
Fix $k\geq0$ and $p\geq2$.  Suppose $h:\bbR^n\to\bbR^n$ belongs to the
Gaussian Sobolev space with $k+1$ derivatives in $L^r(P_Z)$ for every
finite $r\geq2$, written
\[
h\in W^{k+1,r}(P_Z;\bbR^n)
\quad\text{for all }r\in[2,\infty),
\]
and has continuous partial derivatives through order $k+1$.
Then, for every integer $0\leq m\leq k$,
\[
\|\delta(D^m h)\|_{L^p(P_Z)}
\leq
c_0 p^4
\left\{
\|D^m h\|_{L^p(P_Z)}
+
\|D^{m+1}h\|_{L^p(P_Z)}
\right\}.
\]
\end{lemma}

\begin{proof}
Let $q=p/(p-1)$.  We prove the base bound in two steps.  First,
we follow the divergence-continuity proof in
\citet[Proposition~1.5.4]{nualartMalliavinCalculusRelated2006}, keeping explicit
the constant $K_q$ in Meyer's inequality
\[
\|DF\|_q\leq K_q\|CF\|_q,
\]
for mean-zero $F$, where $D$ denotes differentiation with respect to $z$,
$L$ denotes the Ornstein--Uhlenbeck generator ($Lf=\Delta f-z\cdot\nabla f$
on smooth functions, so $Lf=-\sum_{r\geq0}r\,J_rf$ on the chaos
decomposition), and $C=(-L)^{1/2}$ is the corresponding square-root
operator.  Second,
we use the dimension-free operator norm bound in
\citet{arcozziRieszTransformsCompact1998} to show that $K_q=O(p)$.

Following the proof of
\citet[Proposition~1.5.4]{nualartMalliavinCalculusRelated2006}, first take
$G$ to be a polynomial array-valued test function with
$\|G\|_q\leq1$, taking values in the same finite-dimensional Euclidean space
as $\delta(D^m h)$.  For such
$G$, Stein's Gaussian integration-by-parts identity gives
\[
\left|
\E\left[
\left\langle \delta(D^m h),G\right\rangle
\right]
\right|
=
\left|
\E\left[
\left\langle D^m h,DG\right\rangle
\right]
\right|.
\]
Here $DG$ has entries
$\partial G_\ell/\partial z_j$, so the inner product sums over both the
divergence coordinate $j$ and the derivative-array label $\ell$.

Apply the decomposition used in
\citet[Proposition~1.5.4]{nualartMalliavinCalculusRelated2006}, with the
derivative array treated as a finite-dimensional Euclidean vector under the
norm convention above.  Set $u=D^m h$, and write
\[
u=\E u+\widetilde u,
\qquad
G=\E G+\widetilde G.
\]
Then
\[
\left|
\E\left[
\left\langle u,DG\right\rangle
\right]
\right|
\leq
\left|
\E\left[
\left\langle \E u,DG\right\rangle
\right]
\right|
+
\left|
\E\left[
\left\langle \widetilde u,D\widetilde G\right\rangle
\right]
\right|.
\]
For the term involving $\E u$, note that $\E u$ is deterministic.  Applying
Gaussian integration by parts to this constant array gives
\[
\E\left[
\left\langle \E u,DG\right\rangle
\right]
=
\E\left[
\left\langle \delta(\E u),G\right\rangle
\right].
\]
Since the derivative of $\E u$ with respect to $z$ is zero,
$\delta(\E u)$ is linear in $Z$.  Hence H\"older's inequality gives
\[
\left|
\E\left[
\left\langle \E u,DG\right\rangle
\right]
\right|
\leq
\|\delta(\E u)\|_p\|G\|_q.
\]
Because $\delta(\E u)$ is linear in the underlying Gaussian vector, its
$L^p$ norm grows as $O(\sqrt p)$
\citep[Proposition~2.5.2]{vershyninHighdimensionalProbabilityIntroduction2018}:
$\|\delta(\E u)\|_p$ is bounded by $c\sqrt p$ times the Euclidean norm of
the deterministic array $\E u$, with $c$ independent of $n$ and of the
array dimension.  Therefore, since $\|G\|_q\leq1$ and $u=D^m h$,
\[
\left|
\E\left[
\left\langle \E u,DG\right\rangle
\right]
\right|
\leq
c\sqrt p\,|\E u|
\leq
c\sqrt p\,\|u\|_p
=
c\sqrt p\,\|D^m h\|_p,
\]
where the second inequality is Jensen's inequality.
For the second term, the first-chaos part of $\widetilde G$ does not
contribute.  Indeed,
\[
\E\left[
\left\langle \widetilde u,\E(D\widetilde G)\right\rangle
\right]
=
\left\langle \E\widetilde u,\E(D\widetilde G)\right\rangle
=0.
\]
Therefore
\[
\E\left[
\left\langle \widetilde u,D\widetilde G\right\rangle
\right]
=
\E\left[
\left\langle
\widetilde u,\,
D\widetilde G-\E(D\widetilde G)
\right\rangle
\right].
\]
Thus the first-chaos part of $\widetilde G$ does not contribute to the
centered term.  In the operator notation below, this removal is handled by
$R$, which has multiplier zero on chaos orders $0$ and $1$.  Thus we keep
$\widetilde G=G-\E G$ in the notation.  The identity used in
\citet[Proposition~1.5.4]{nualartMalliavinCalculusRelated2006} gives
\[
\left|
\E\left[
\left\langle \widetilde u,D\widetilde G\right\rangle
\right]
\right|
=
\left|
\E\left[
\left\langle
D\widetilde u,\,
D C^{-2}D\widetilde G
\right\rangle
\right]
\right|,
\]
where $C=(-L)^{1/2}$ is as above.  The expectation in $\E u$ is over
$Z$, so $\E u$ is a fixed array and has derivative zero.  Since $u=D^m h$,
\[
D\widetilde u
=
D(u-\E u)
=
D^{m+1}h.
\]
H\"older's inequality gives
\[
\left|
\E\left[
\left\langle
D\widetilde u,\,
D C^{-2}D\widetilde G
\right\rangle
\right]
\right|
\leq
\|D^{m+1}h\|_p\,
\|D C^{-2}D\widetilde G\|_q.
\]
The next displays follow the operator step in
\citet[Proposition~1.5.4]{nualartMalliavinCalculusRelated2006}. We record them
only to track the $p$-dependence of the constant.  Let $R$ denote the
multiplier operator in Nualart's proof.  The operator identity below is applied
to the centered derivative term $D\widetilde G-\E(D\widetilde G)$ justified
above. To keep Nualart's notation, we continue writing $D\widetilde G$.
Meyer's inequality first gives
\[
\|D C^{-2}D\widetilde G\|_q
\leq
K_q\|C^{-1}D\widetilde G\|_q.
\]
The multiplier identity used in
\citet[Proposition~1.5.4]{nualartMalliavinCalculusRelated2006} gives
\[
\|C^{-1}D\widetilde G\|_q
=
\|D C^{-1}\sqrt R\,\widetilde G\|_q.
\]
Applying Meyer's inequality a second time gives
\[
\|D C^{-1}\sqrt R\,\widetilde G\|_q
\leq
K_q\|\sqrt R\,\widetilde G\|_q.
\]
Lemma~\ref{lem:sqrt_R_bounded} gives
\[
\|\sqrt R\,\widetilde G\|_q
=
\|\sqrt R(I-J_0-J_1)\widetilde G\|_q
\leq
c p^2\|\widetilde G\|_q.
\]
Finally, $\widetilde G=G-\E G$, so
\[
\|\widetilde G\|_q\leq \|G\|_q+|\E G|\leq 2\|G\|_q,
\]
where the last step uses Jensen's inequality, $|\E G|\leq\E|G|\leq\|G\|_q$.
Combining the two Meyer-inequality steps, the multiplier identity, the
multiplier bound, and the centering bound,
\[
\|D C^{-2}D\widetilde G\|_q
\leq
c p^2K_q^2\|G\|_q,
\]
and therefore
\[
\left|
\E\left[
\left\langle \delta(D^m h),G\right\rangle
\right]
\right|
\leq
c\sqrt p\,\|D^m h\|_p
+
c p^2K_q^2\,\|D^{m+1}h\|_p.
\]
The numerical constants in this display do not depend on $n$ or on the number
of derivative-array entries.

To bound $K_q$, use the dimension-free bound from
\citet[Theorem~4]{arcozziRieszTransformsCompact1998} for the real-valued
Gaussian Riesz transform $DC^{-1}$:
\[
\|DC^{-1}\|_{q\to q}
\leq
2(q^*-1),
\qquad
q^*=\max\{q,q/(q-1)\},
\]
on mean-zero Gaussian $L^q$ functions.  \citet[Section~3.2]{banuelosFoundationalInequalities2010}
records that real-valued Riesz-transform bounds extend to Hilbert-valued
functions.  Since the derivative arrays here take values in finite-dimensional
Euclidean spaces, the Hilbert-valued form covers the array-valued $G$ used
above, with the same dimension-free constant.

\citet[Section~3.4]{banuelosFoundationalInequalities2010} summarizes this
Gaussian Riesz-transform inequality and the relevant constant bounds: the
underlying dimension-free boundedness is due to
\citet{meyerTransformationsRieszLois1984};
\citet{pisierRieszTransformsSimpler1988} gave an earlier explicit control of
the growth of the constant; and the Arcozzi bound above gives the sharper
constant $2(q^*-1)$ used here.  This is the same operator that enters the
displayed Meyer inequality above.  Indeed, for any mean-zero $F$,
\[
\|DF\|_q
=
\|DC^{-1}(CF)\|_q
\leq
\|DC^{-1}\|_{q\to q}\|CF\|_q
\leq
2(q^*-1)\|CF\|_q.
\]
Thus Meyer's inequality holds with
\[
K_q
\leq
\|DC^{-1}\|_{q\to q}
\leq
2(q^*-1).
\]
In this proof $q=p/(p-1)$, so $q/(q-1)=p$ and $q\leq p$ for $p\geq2$.
Thus $q^*=p$, and
\[
K_q\leq 2(p-1)=O(p).
\]
Consequently $p^2K_q^2=O(p^4)$, which dominates the $O(\sqrt p)$ term for
$p\geq2$.  Thus, for every polynomial array-valued $G$ with
$\|G\|_q\leq1$, we have shown
\[
\left|
\E\left[
\left\langle \delta(D^m h),G\right\rangle
\right]
\right|
\leq
c_0p^4
\left\{
\|D^m h\|_{L^p(P_Z)}
+
\|D^{m+1}h\|_{L^p(P_Z)}
\right\}.
\]
This is the testing argument used in
\citet[Proposition~1.5.4]{nualartMalliavinCalculusRelated2006}.  We now
deduce the stated $L^p$ bound from the testing bound by duality.  The
duality identity
$\|X\|_{L^p(P_Z)}=\sup\{\E[\langle X,G\rangle]\,:\,\|G\|_q\leq1\}$ is valid
only when $X\in L^p(P_Z)$, so we first verify this membership for
$X=\delta(D^m h)$.  Coordinatewise, the Cauchy--Schwarz inequality bounds the two sums in the
defining expression of $\delta(D^m h)$,
\begin{align*}
\Bigl|\sum_{j=1}^n z_j\,(D^m h)_{\ell,j}(z)\Bigr|
&\leq
\|z\|_2\,\Bigl(\sum_{j=1}^n (D^m h)_{\ell,j}(z)^2\Bigr)^{1/2},
\\
\Bigl|\sum_{j=1}^n \partial_j(D^m h)_{\ell,j}(z)\Bigr|
&\leq
\sqrt n\,\Bigl(\sum_{j=1}^n \bigl(\partial_j(D^m h)_{\ell,j}(z)\bigr)^2\Bigr)^{1/2},
\end{align*}
and combining the derivative labels $\ell$ by the triangle inequality in
$\ell^2$ gives the pointwise bound
\[
\|\delta(D^m h)(z)\|_2
\leq
\|z\|_2\,\|D^m h(z)\|_2+\sqrt n\,\|D^{m+1}h(z)\|_2 .
\]
The second term on the right has finite $p$th moment directly from the
hypothesis with $r=p$.  The first term has finite $p$th moment by the
Cauchy--Schwarz inequality in $L^2(P_Z)$,
\[
\E\bigl[\,\|Z\|_2^p\,\|D^m h(Z)\|_2^p\,\bigr]
\leq
\bigl(\E\|Z\|_2^{2p}\bigr)^{1/2}
\bigl(\E\|D^m h(Z)\|_2^{2p}\bigr)^{1/2}
<\infty,
\]
which uses the hypothesis with $r=2p$. This is the only place the
hypothesis is used beyond $r=p$.  Hence $\delta(D^m h)\in L^p(P_Z)$.
Finally, since polynomials are dense in $L^q(P_Z)$, the supremum in the
duality identity may be restricted to polynomial $G$, where the testing
bound applies.  It follows that
\[
\|\delta(D^m h)\|_{L^p(P_Z)}
\leq
c_0 p^4
\left\{
\|D^m h\|_{L^p(P_Z)}
+
\|D^{m+1}h\|_{L^p(P_Z)}
\right\}.
\]
\end{proof}

\begin{lemma}[Admissibility of the test function]
\label{lem:test_function_admissible}
Let $Z\sim \N(0,I_n)$, and let $P_Z$ denote its law.  Fix $p\geq2$ and
set $q=p/(p-1)$.  Let $A:\bbR^n\to\bbR^M$ be continuously differentiable
for some finite $M$.  Throughout this lemma, $|\cdot|$ denotes the
Euclidean norm after vectorizing finite-dimensional arrays, and $L^p$ norms
are taken after applying this pointwise norm.  In particular, $|DA|$ is the
Frobenius norm of the Jacobian of $A$.  Suppose
\[
A\in L^p(P_Z),
\qquad
DA\in L^p(P_Z),
\]
and define
\[
F:=A|A|^{p-2}.
\]
Then $F\in W^{1,q}(P_Z;\bbR^M)$.  Moreover,
\[
|DF|
\leq
(p-1)|A|^{p-2}|DA|.
\]
\end{lemma}

\begin{proof}
The Sobolev claim has two parts: $F\in L^q(P_Z)$ and
$DF\in L^q(P_Z)$.  Since $|F|=|A|^{p-1}$ and $q=p/(p-1)$,
\[
\E[|F|^q]
=
\E[|A|^{(p-1)q}]
=
\E[|A|^p]
<\infty.
\]
Thus $F\in L^q(P_Z)$.

It remains to check the derivative.  If $p=2$, then $F=A$, so
\[
DF=DA\in L^2(P_Z)=L^q(P_Z),
\]
and the displayed derivative bound is immediate.  Now suppose $p>2$.  Define
\[
\Phi(x):=x|x|^{p-2},
\qquad x\in\bbR^M.
\]
For $x\ne0$, the Jacobian of $\Phi$ is
\[
D\Phi(x)
=
|x|^{p-2}I_M
+
(p-2)|x|^{p-4}xx^\top .
\]
At $x=0$, the derivative is zero.  Indeed,
\[
\frac{|\Phi(x)-\Phi(0)|}{|x|}
=
\frac{|x|^{p-1}}{|x|}
=
|x|^{p-2}
\to0
\qquad\text{as }x\to0.
\]
Thus $D\Phi(0)=0$.  The formula for $D\Phi(x)$ when $x\ne0$ is
compatible with this value at zero.  The second term satisfies
\[
\left\|(p-2)|x|^{p-4}xx^\top\right\|_{\mathrm{op}}
=
(p-2)|x|^{p-2}
\to0
\qquad\text{as }x\to0,
\]
and the first term $|x|^{p-2}I_M$ also converges to zero in operator norm.
Therefore $\Phi$ is continuously differentiable.

By the chain rule, $F=\Phi(A)$ satisfies
\[
DF=D\Phi(A)DA.
\]
We now bound the operator norm of $D\Phi(x)$.  For the first term in the
Jacobian,
\[
\left\||x|^{p-2}I_M\right\|_{\mathrm{op}}
=
|x|^{p-2}.
\]
For the second term, use $\|xx^\top\|_{\mathrm{op}}=|x|^2$.  Thus
\[
\left\|(p-2)|x|^{p-4}xx^\top\right\|_{\mathrm{op}}
=
(p-2)|x|^{p-4}\|xx^\top\|_{\mathrm{op}}
=
(p-2)|x|^{p-2}.
\]
Since $D\Phi(x)$ is the sum of these two terms, the triangle inequality for
operator norms gives, for every $x\in\bbR^M$,
\[
\begin{aligned}
\|D\Phi(x)\|_{\mathrm{op}}
&\leq
\left\||x|^{p-2}I_M\right\|_{\mathrm{op}}
+
\left\|(p-2)|x|^{p-4}xx^\top\right\|_{\mathrm{op}} \\
&=
|x|^{p-2}
+
(p-2)|x|^{p-2}
=
(p-1)|x|^{p-2}.
\end{aligned}
\]
Hence
\[
|DF|
\leq
\|D\Phi(A)\|_{\mathrm{op}}|DA|
\leq
(p-1)|A|^{p-2}|DA|.
\]
H\"older's inequality, with exponents $p/(p-2)$ and $p$, gives
\[
\|DF\|_{L^q(P_Z)}
\leq
(p-1)
\left\||A|^{p-2}\right\|_{L^{p/(p-2)}(P_Z)}
\|DA\|_{L^p(P_Z)}.
\]
The first factor is
\[
\left\||A|^{p-2}\right\|_{L^{p/(p-2)}(P_Z)}
=
\left(\E[|A|^p]\right)^{(p-2)/p}
=
\|A\|_{L^p(P_Z)}^{p-2}.
\]
Therefore
\[
\|DF\|_{L^q(P_Z)}
\leq
(p-1)\|A\|_{L^p(P_Z)}^{p-2}\|DA\|_{L^p(P_Z)}
<\infty.
\]
Thus $DF\in L^q(P_Z)$, and so $F\in W^{1,q}(P_Z;\bbR^M)$.
\end{proof}

\begin{lemma}[One-step recurrence in the derivative order]
\label{lem:one_step_recurrence}
Fix $p\geq2$ and $m\geq0$.  Suppose $h:\bbR^n\to\bbR^n$ has continuous partial
derivatives through order $m+2$ and
$h\in W^{m+2,r}(P_Z;\bbR^n)$ for all $r\in[2,\infty)$.  Then
\[
\|\delta(D^m h)\|_p^2
\leq
(p-1)\|D^m h\|_p^2
+
(p-1)\|D^m h\|_p\,\|\delta(D^{m+1}h)\|_p.
\]
\end{lemma}

\begin{proof}
If $\|\delta(D^m h)\|_p=0$, the displayed inequality is immediate, so assume
$\|\delta(D^m h)\|_p>0$.  For each derivative label $\ell$ and each input
coordinate $i$, differentiating the definition of $\delta(D^m h)$ gives
\[
\begin{aligned}
\frac{\partial}{\partial z_i}\{\delta(D^m h)\}_{\ell}
&=
\frac{\partial}{\partial z_i}
\left[
\sum_{j=1}^n z_j (D^m h)_{\ell,j}
-
\sum_{j=1}^n
\frac{\partial}{\partial z_j}(D^m h)_{\ell,j}
\right]  \\
&=
(D^m h)_{\ell,i}
+
\sum_{j=1}^n z_j
\frac{\partial}{\partial z_i}(D^m h)_{\ell,j}
-
\sum_{j=1}^n
\frac{\partial}{\partial z_j}
\frac{\partial}{\partial z_i}(D^m h)_{\ell,j}  \\
&=
(D^m h)_{\ell,i}
+
\{\delta(D^{m+1}h)\}_{i,i_1,\ldots,i_m}.
\end{aligned}
\]
The last equality uses the assumed continuous mixed partial derivatives.
When $m=0$, the derivative label on the last term is just $i$.
Stacking these identities over all $i$ and all derivative labels $\ell$,
and applying the triangle inequality for the Euclidean norm, gives
\begin{equation}
\label{eq:dm_derivative_bound}
|D\{\delta(D^m h)\}|
\leq
|D^m h|
+
|\delta(D^{m+1}h)|.
\end{equation}
Set $A:=\delta(D^m h)$, so $DA=D\{\delta(D^m h)\}$.
Set $F:=A|A|^{p-2}$.
Before applying Gaussian integration by parts with $F$, we check that $F$
is admissible.  Lemma~\ref{lem:test_function_admissible} says it is enough to
show $A\in L^p(P_Z)$ and $DA\in L^p(P_Z)$.
Lemma~\ref{lem:first_order_divergence_continuity} with $k=m+1$ and
derivative order $m$ gives $A=\delta(D^m h)\in L^p(P_Z)$.  The same lemma
with derivative order $m+1$ gives $\delta(D^{m+1}h)\in L^p(P_Z)$.
Together with
\[
|DA|=|D\{\delta(D^m h)\}|
\leq
|D^m h|+|\delta(D^{m+1}h)|,
\]
this gives $DA\in L^p(P_Z)$.  Hence
Lemma~\ref{lem:test_function_admissible} gives $F\in W^{1,q}(P_Z)$, where
$q=p/(p-1)$.  Now
\begin{equation}
\label{eq:dm_expanded_bound}
\begin{aligned}
\|A\|_p^p
&=
\E[|A|^p]  \\
&=
\E\langle F,A\rangle  \\
&=
\E\langle F,\delta(D^m h)\rangle  \\
&=
\E\langle DF,D^m h\rangle
&&\text{(divergence-gradient duality)}  \\
&\leq
\E\{|DF|\,|D^m h|\}  \\
&\leq
(p-1)
\E\{|A|^{p-2}|DA|\,|D^m h|\}
&&\text{(Lemma~\ref{lem:test_function_admissible})}  \\
&\leq
(p-1)
\E\!\left[
|A|^{p-2}
\left\{
|D^m h|+|\delta(D^{m+1}h)|
\right\}
|D^m h|
\right]
&&\text{(by \eqref{eq:dm_derivative_bound})}  \\
&=
(p-1)
\E\!\left[
|A|^{p-2}|D^m h|^2
\right]  \\
&\quad+
(p-1)
\E\!\left[
|A|^{p-2}|D^m h|\,|\delta(D^{m+1}h)|
\right].
\end{aligned}
\end{equation}

When $p=2$, \eqref{eq:dm_expanded_bound} gives
\[
\begin{aligned}
\|A\|_2^2
&\leq
\E[|D^m h|^2]
+
\E\!\left[
|D^m h|\,|\delta(D^{m+1}h)|
\right]  \\
&=
\|D^m h\|_2^2
+
\E\!\left[
|D^m h|\,|\delta(D^{m+1}h)|
\right]  \\
&\leq
\|D^m h\|_2^2
+
\|D^m h\|_2\,\|\delta(D^{m+1}h)\|_2.
\end{aligned}
\]
Since $A=\delta(D^m h)$, this is the recurrence at $p=2$.

Now suppose $p>2$.  For the first term in
\eqref{eq:dm_expanded_bound}, H\"older's inequality with exponents
$p/(p-2)$ and $p/2$ gives
\[
\E\!\left[
|A|^{p-2}|D^m h|^2
\right]
\leq
\|A\|_p^{p-2}\|D^m h\|_p^2.
\]
For the second term, H\"older's inequality with the same exponents $p/(p-2)$ and $p/2$ gives
\[
\E\!\left[
|A|^{p-2}|D^m h|\,|\delta(D^{m+1}h)|
\right]
\leq
\|A\|_p^{p-2}
\left\|
|D^m h|\,|\delta(D^{m+1}h)|
\right\|_{p/2}.
\]
The $L^{p/2}(P_Z)$ factor in this display is bounded as follows:
\begin{align*}
\left\|
|D^m h|\,|\delta(D^{m+1}h)|
\right\|_{p/2}
&=
\left(
\E\left[
|D^m h|^{p/2}|\delta(D^{m+1}h)|^{p/2}
\right]
\right)^{2/p}  \\
&\leq
\left[
\left(\E[|D^m h|^p]\right)^{1/2}
\left(\E[|\delta(D^{m+1}h)|^p]\right)^{1/2}
\right]^{2/p}  \\
&=
\|D^m h\|_p\,\|\delta(D^{m+1}h)\|_p.
\end{align*}
Using these two inequalities in \eqref{eq:dm_expanded_bound} gives
\[
\|A\|_p^p
\leq
(p-1)\|A\|_p^{p-2}\|D^m h\|_p^2
+
(p-1)\|A\|_p^{p-2}
\|D^m h\|_p\,\|\delta(D^{m+1}h)\|_p.
\]
Since $\|A\|_p>0$, division by $\|A\|_p^{p-2}$ proves the recurrence for
$p>2$, and hence for all $p\geq2$.
\end{proof}

\begin{lemma}[Iterated Gaussian divergence-continuity bound]
\label{lem:divergence_continuity}
Fix $k\geq0$ and $p\geq2$.  Suppose $h:\bbR^n\to\bbR^n$ belongs to the
Gaussian Sobolev space with $k+1$ derivatives in $L^r(P_Z)$ for every
finite $r\geq2$, written
\[
h\in W^{k+1,r}(P_Z;\bbR^n)
\quad\text{for all }r\in[2,\infty),
\]
and has continuous partial derivatives through order $k+1$.
Then
\[
\|\delta(h)\|_{L^p(P_Z)}
\leq
c_k p^{1+3\cdot 2^{-k}}\|h\|_{W^{k+1,p}(P_Z)}.
\]
The constant $c_k$ depends only on $k$, not on $n$.
\end{lemma}

\begin{proof}
The lemma only requires a bound for $\delta(h)$, but the proof controls it
by iterating a one-step recurrence over higher derivative arrays.  One use of
the recurrence bounds $\delta(D^m h)$ in terms of $D^m h$ and the next
divergence term $\delta(D^{m+1}h)$.  Therefore, starting from $m=0$, the
argument must keep track of the sequence
\[
\delta(h),\quad \delta(Dh),\quad \delta(D^2h),\quad \ldots .
\]
The index $m$ records the current derivative order, and $r$ records
how many further derivative levels are used to bound that term.
The condition $m+r+1\leq k+1$ ensures that the right-hand side never uses
derivatives of $h$ beyond order $k+1$.  For each $r=0,\ldots,k$ and each
$m$ satisfying $m+r+1\leq k+1$, we prove
\begin{equation}
\label{eq:modular_auxiliary_claim}
\|\delta(D^m h)\|_p
\leq
c_r p^{1+3\cdot 2^{-r}}\|D^m h\|_{W^{r+1,p}(P_Z)}.
\end{equation}
Taking $r=k$ and $m=0$ in \eqref{eq:modular_auxiliary_claim} gives the
bound in the lemma statement.

First take $r=0$, and fix $m$ satisfying $m+1\leq k+1$.
Lemma~\ref{lem:first_order_divergence_continuity} gives
\[
\begin{aligned}
\|\delta(D^m h)\|_{L^p(P_Z)}
&\leq
c_0p^4
\left\{
\|D^m h\|_{L^p(P_Z)}
+
\|D^{m+1}h\|_{L^p(P_Z)}
\right\}  \\
&=
c_0p^4\|D^m h\|_{W^{1,p}(P_Z)}  \\
&=
c_0p^{1+3\cdot 2^{-0}}\|D^m h\|_{W^{1,p}(P_Z)}.
\end{aligned}
\]
This proves \eqref{eq:modular_auxiliary_claim} when $r=0$.

For $r\geq1$, the induction is on $r$, not on $m$.  The integer $r$
records how many additional derivative levels are being used beyond $D^m h$.
Thus the $r=0$ case uses only $D^m h$ and $D^{m+1}h$, while the
$r$-case uses derivatives from $D^m h$ through $D^{m+r+1}h$.  For each
fixed $r$, the claim is proved for every derivative order $m$ for which
these derivatives exist, namely $m+r+1\leq k+1$.  The one-step recurrence
applied at derivative order $m$ produces the next divergence term
$\delta(D^{m+1}h)$.  That term is handled by the already-proved
$(r-1)$-case, applied with $m$ replaced by $m+1$.

Now fix $r$ with $1\leq r\leq k$.  Assume that
\eqref{eq:modular_auxiliary_claim} has already been proved with $r$ replaced
by $r-1$.  Thus, for every integer $m$ satisfying $m+r\leq k+1$,
\[
\|\delta(D^m h)\|_p
\leq
c_{r-1}p^{1+3\cdot 2^{-(r-1)}}\|D^m h\|_{W^{r,p}(P_Z)}.
\]
Fix $m$ satisfying $m+r+1\leq k+1$.  We need to prove
\[
\|\delta(D^m h)\|_p
\leq
c_rp^{1+3\cdot 2^{-r}}\|D^m h\|_{W^{r+1,p}(P_Z)}.
\]
We prove this by bounding $\|\delta(D^m h)\|_p^2$ and then taking square
roots.  Since $r\geq1$, the condition $m+r+1\leq k+1$ implies
$m+2\leq k+1$, so Lemma~\ref{lem:one_step_recurrence} applies at derivative
order $m$.  Since $(m+1)+r\leq k+1$, the induction hypothesis applies to
derivative order $m+1$.
\[
\begin{aligned}
\|\delta(D^m h)\|_p^2
&\leq
(p-1)\|D^m h\|_p^2
+
(p-1)\|D^m h\|_p\,\|\delta(D^{m+1}h)\|_p
&&\text{(a)}  \\
&\leq
(p-1)\|D^m h\|_p^2
+
(p-1)c_{r-1}p^{1+3\cdot 2^{-(r-1)}}
\|D^m h\|_p\|D^{m+1}h\|_{W^{r,p}(P_Z)}
&&\text{(b)}  \\
&\leq
(p-1)\|D^m h\|_{W^{r+1,p}(P_Z)}^2
+
(p-1)c_{r-1}p^{1+3\cdot 2^{-(r-1)}}
\|D^m h\|_{W^{r+1,p}(P_Z)}^2
&&\text{(c)}  \\
&\leq
p\|D^m h\|_{W^{r+1,p}(P_Z)}^2
+
c_{r-1}p^{2+3\cdot 2^{-(r-1)}}
\|D^m h\|_{W^{r+1,p}(P_Z)}^2
&&\text{(d)}  \\
&\leq
p^{2+3\cdot 2^{-(r-1)}}
\|D^m h\|_{W^{r+1,p}(P_Z)}^2
+
c_{r-1}p^{2+3\cdot 2^{-(r-1)}}
\|D^m h\|_{W^{r+1,p}(P_Z)}^2
&&\text{(e)}  \\
&=
(1+c_{r-1})p^{2+3\cdot 2^{-(r-1)}}
\|D^m h\|_{W^{r+1,p}(P_Z)}^2 .
\end{aligned}
\]
In this display, (a) applies Lemma~\ref{lem:one_step_recurrence}, (b)
applies the induction hypothesis, (c) uses the Sobolev norm comparisons
stated below, (d) uses \(p-1\leq p\), and (e) uses
\(p\leq p^{2+3\cdot 2^{-(r-1)}}\) for \(p\geq2\).  The Sobolev norm
comparisons used in (c) are
\[
\|D^m h\|_p\leq \|D^m h\|_{W^{r+1,p}(P_Z)},
\qquad
\|D^{m+1}h\|_{W^{r,p}(P_Z)}
\leq
\|D^m h\|_{W^{r+1,p}(P_Z)}.
\]
Taking square roots and using
\[
\frac{2+3\cdot 2^{-(r-1)}}{2}=1+3\cdot 2^{-r}
\]
gives
\[
\|\delta(D^m h)\|_p
\leq
(1+c_{r-1})^{1/2}
p^{1+3\cdot 2^{-r}}\|D^m h\|_{W^{r+1,p}(P_Z)}.
\]
Set $c_r:=(1+c_{r-1})^{1/2}$.  This proves
\eqref{eq:modular_auxiliary_claim} for this value of $r$.  By induction,
\eqref{eq:modular_auxiliary_claim} holds for every $r=0,\ldots,k$.  Taking
$r=k$ and $m=0$ gives the stated bound.
\end{proof}

\noindent\phantomsection\label{proof:general_conc}
\begin{proof}[Proof of Theorem~\ref{thm:general_conc}]
Let
\[
G_n(f):=\sure_n(f)-L_n(f)
\]
denote the \sure{} error process whose uniform size must be controlled.

\medskip\noindent\textit{Step 1: Sub-Weibull increments.}
By Lemma~\ref{lem:representation},
\[
G_n(f_\gamma)-G_n(f_{\gamma'})
=
\frac{2}{n}\Psi(h),
\qquad
h:=g_\gamma-g_{\gamma'},
\]
because the $f$-independent term
$\{\tr(\Sigma)-\|\varepsilon\|_2^2\}/n$ cancels and $\Psi$ is linear in
the shrinkage adjustment.

We next rewrite the increment in standard-normal coordinates, because
Lemma~\ref{lem:divergence_continuity} is stated for functions of a standard
Gaussian vector.  Write
\[
Z:=\Sigma^{-1/2}(Y-\theta)\sim \N(0,I_n),
\qquad
Y=\theta+\Sigma^{1/2}Z.
\]
To express the shrinkage adjustment $h(Y)$ in these coordinates, define
\[
\tilde h(z):=\Sigma^{1/2}h(\theta+\Sigma^{1/2}z).
\]
With this definition, the linear part of the standard-normal divergence
corresponds to the noise inner product in the original coordinates:
\[
z^\top \tilde h(z)
=
(Y-\theta)^\top h(Y)
=
\langle \varepsilon,h(Y)\rangle.
\]
The derivative term transforms in the same coordinates.  By the chain rule,
\[
\tr\{D_z\tilde h(z)\}
=
\tr\{\Sigma^{1/2}D_y h(\theta+\Sigma^{1/2}z)\Sigma^{1/2}\}
=
\tr\{\Sigma Dh(Y)\}.
\]
Therefore, for the standard-normal divergence
$\delta(a)(z)=z^\top a(z)-\tr\{Da(z)\}$,
\[
\delta(\tilde h)(Z)
=
\langle\varepsilon,h(Y)\rangle-\tr\{\Sigma Dh(Y)\}
=
-\Psi(h).
\]
Applying Lemma~\ref{lem:divergence_continuity} to $\tilde h$ and then translating the resulting Sobolev norm back to $Y$-coordinates gives
\[
\|\Psi(h)\|_{L^p} = \|\delta(\tilde{h})\|_p \leq C_k\,p^{1+3\cdot 2^{-k}}\,\|\tilde{h}\|_{W^{k+1,p}(P_Z)} \leq C'_k\,p^{1+3\cdot 2^{-k}}\,\|h\|_{W^{k+1,p}(P_Y)},
\]
where the last inequality uses Assumption~\ref{asm:sampling_array}: each
$z$-derivative contributes factors of $\Sigma^{1/2}$, whose operator norm is
uniformly bounded.  The moment envelope converts this Sobolev bound into the
semi-metric $\rho_n^\beta(f_\gamma, f_{\gamma'}) := \frac{1}{n}M_{P_Y,k}^\beta(g_\gamma - g_{\gamma'})$.  By the moment envelope definition,
$\|h\|_{W^{k+1,p}} \leq p^\beta M_{P_Y,k}^\beta(h)$, so
\[
\|G_n(f_\gamma) - G_n(f_{\gamma'})\|_{L^p}
\lesssim
p^{1+3\cdot 2^{-k}+\beta}\,
\rho_n^\beta(f_\gamma, f_{\gamma'}).
\]
The sub-Weibull moment characterization
\citep[Theorem~2.1]{vladimirovaSubWeibullDistributionsGeneralizing2020}
implies, under the present $\psi_\alpha$ convention, that
$\|X\|_{L^p}\leq C s\,p^{1/\alpha}$ for all $p\geq2$ implies
$\|X\|_{\psi_\alpha}\leq C's$.  With
\[
\alpha:=\frac{1}{1+3\cdot 2^{-k}+\beta},
\]
the preceding display gives
\[
\|G_n(f_\gamma)-G_n(f_{\gamma'})\|_{\psi_\alpha}
\lesssim
\rho_n^\beta(f_\gamma,f_{\gamma'})
\leq
\frac{\nu_n}{n}\|\gamma-\gamma'\|_2.
\]
The final inequality is Assumption~\ref{asm:sobolev_reg} divided by $n$.
Thus $\{G_n(f_\gamma)\}$ is a $\psi_\alpha$-increment process over
$\Gamma$ with Euclidean Lipschitz scale $L\lesssim \nu_n/n$.

\medskip\noindent\textit{Step 2: Chaining and centering.}
Consider the centered process
\[
X_\gamma:=G_n(f_\gamma)-G_n(f_{\gamma_0}).
\]
By construction, $X_{\gamma_0}=0$, and Step~1 gives the required
$\psi_\alpha$ increment bound with Euclidean Lipschitz scale
$L\lesssim\nu_n/n$.  The separability hypothesis is preserved by centering:
subtracting the fixed random variable $G_n(f_{\gamma_0})$ from each
$G_n(f_\gamma)$ leaves the process's approximating sequences unchanged.
Write $\mathcal F_0=\{f_\gamma:\gamma\in\Gamma_0\}$ for a
countable $\Gamma_0\subseteq\Gamma$. Almost surely, any sequence in
$\Gamma_0$ along which $G_n(f_{\gamma_m})\to G_n(f_\gamma)$ also satisfies
$X_{\gamma_m}\to X_\gamma$, so $\Gamma_0\cup\{\gamma_0\}$ is a countable
separating set for $\{X_\gamma\}$ in the sense required by
Proposition~\ref{prop:chaining}.  Proposition~\ref{prop:chaining} therefore gives
\[
\E\!\left[\sup_{\gamma\in\Gamma}|X_\gamma|\right]
\leq
C_\alpha\,\frac{\nu_n}{n}\,\mathrm{diam}(\Gamma)\,\max\{d_\Gamma,1\}^{1/\alpha}.
\]
Since $1/\alpha=1+3\cdot 2^{-k}+\beta$ and $\mathrm{diam}(\Gamma)$ is absorbed
into the constant, this term is
\[
O\!\left(\frac{\nu_n \max\{d_\Gamma,1\}^{1+3\cdot 2^{-k}+\beta}}{n}\right).
\]

It remains to bound the \sure{} error process at the reference point:
\[
G_n(f_{\gamma_0})
=
\frac{1}{n}\{\tr(\Sigma)-\|\varepsilon\|_2^2\}
+\frac{2}{n}\Psi(g_{\gamma_0}).
\]
Because $\varepsilon\sim\N(0,\Sigma)$,
\[
\E[\|\varepsilon\|_2^2]=\tr(\Sigma),
\qquad
\mathrm{Var}(\|\varepsilon\|_2^2)=2\tr(\Sigma^2)
\]
by Isserlis' theorem.  Hence
\[
\E\left[\left|\frac{1}{n}\{\tr(\Sigma)-\|\varepsilon\|_2^2\}\right|\right]
\leq
\frac{\sqrt{2\tr(\Sigma^2)}}{n}
=O(n^{-1/2})
\]
under Assumption~\ref{asm:sampling_array}.  For the second term, the same divergence-continuity argument as in Step~1, applied to $g_{\gamma_0}$, gives
\[
\E[|\Psi(g_{\gamma_0})|]
\leq
\|\Psi(g_{\gamma_0})\|_{L^2}
\lesssim
M_{P_Y,k}^{\beta}(g_{\gamma_0})
\leq
\nu_n.
\]
Thus $\E[|G_n(f_{\gamma_0})|]\lesssim n^{-1/2}+\nu_n/n$. Combining the reference-point and increment terms,
\[
\E\left[\sup_{f\in \mathcal{F}}|\sure_n(f) - L_n(f)|\right] \lesssim \frac{1}{\sqrt{n}} + \frac{\nu_n\,\max\{d_\Gamma,1\}^{1+3\cdot 2^{-k}+\beta}}{n}.
\]

\medskip\noindent\textit{Step 3: Oracle inequality.}
The final step converts the uniform deviation bound into the oracle comparison.
Since $\hat f$ minimizes
$\sure_n$ over $\mathcal F$,
\[
\sure_n(\hat f)\leq \sure_n(f^*).
\]
Writing $G_n(f)=\sure_n(f)-L_n(f)$,
\begin{align*}
L_n(\hat f)-L_n(f^*)
&=
\{\sure_n(\hat f)-G_n(\hat f)\}
-\{\sure_n(f^*)-G_n(f^*)\} \\
&\leq
G_n(f^*)-G_n(\hat f).
\end{align*}
Using the reference map $f_{\gamma_0}$,
\[
G_n(f^*)-G_n(\hat f)
=
\{G_n(f^*)-G_n(f_{\gamma_0})\}
-\{G_n(\hat f)-G_n(f_{\gamma_0})\}
\leq
2\sup_{f\in\mathcal F}|G_n(f)-G_n(f_{\gamma_0})|.
\]
Step~2 bounds the expectation of this centered supremum by
$C\nu_n \max\{d_\Gamma,1\}^{1+3\cdot 2^{-k}+\beta}/n$.  The $n^{-1/2}$ centering term does not
appear in the oracle inequality because it is common to all $f$ and cancels.
\end{proof}

\subsubsection{Proof of Theorem~\ref{thm:uniform_conc} (via Theorem~\ref{thm:general_conc})}

\begin{proof}[Proof of Theorem~\ref{thm:uniform_conc}]
By Lemma~\ref{lem:pointwise_regular_to_sobolev}, Assumptions~\ref{asm:sampling_array}
and~\ref{asm:regularity} imply Assumption~\ref{asm:sobolev_reg} at $k=0$,
up to constants depending only on fixed sampling and envelope constants.  Theorem~\ref{thm:general_conc}
with $k=0$ gives exponent $1+3\cdot 2^{-0}+\beta=4+\beta$, yielding
\[
\E\left[\sup_{f\in \mathcal{F}}|\sure_n(f) - L_n(f)|\right] \lesssim \frac{1}{\sqrt{n}} + \frac{\nu_n\,\max\{d_\Gamma,1\}^{4+\beta}}{n}.
\]
The oracle-inequality part of Theorem~\ref{thm:general_conc} gives the
matching theorem conclusion
\[
\E[L_n(\hat f)-L_n(f^*)]
\lesssim
\frac{\nu_n\,\max\{d_\Gamma,1\}^{4+\beta}}{n}.
\]
\end{proof}

\begin{remark}[Role of the higher-order cases]
The main concentration theorem uses only the $k=0$ case of Theorem~\ref{thm:general_conc}. The higher-order cases are included for completeness: the exponent $1+3\cdot 2^{-k}+\beta$ interpolates between $4+\beta$ (when $k=0$, so only first derivatives are controlled) and $1+\beta$ (as arbitrarily many derivative levels are controlled). The improvement with $k$ follows by iterating the Meyer-type divergence bound through Lemma~\ref{lem:one_step_recurrence}: for $g$ satisfying the $k$th Sobolev moment condition, the divergence-continuity constant improves from $c_p=O(p^4)$ to $c_p=O(p^{1+3\cdot 2^{-k}})$.
\end{remark}

\begin{remark}[Bounded-envelope benchmark]
The case $\beta=0$ in Theorem~\ref{thm:uniform_conc} covers families whose shrinkage adjustments and parameter increments are uniformly bounded in $y$. A globally Lipschitz smoother can still have a linearly growing adjustment, for example $g(y)=(S-I)y$, and such cases are handled by the polynomial-envelope condition, typically with $\beta=1/2$. We do not pursue a separate Lipschitz sharpening here. Section~\ref{sec:regularity} explains how this continuum-class result relates to the finite-collection oracle inequalities of \citet{bellecSecondOrderStein2021}.
\end{remark}

\subsection{Proof of the \sure{} Model-Averaging Oracle Inequality}\label{app:model_averaging_proof}

\subsubsection{Proof of Proposition~\ref{prop:averaging}}

\begin{proof}[Proof of Proposition~\ref{prop:averaging}]
\textit{Step 1 (Fixed-weight oracle comparison).}\;
Fix any $v\in\Delta^{K-1}$, and recall that
$\tilde f(Y)=f_{\hat w(Y)}(Y)$. In this proof,
$f_w$ denotes the map $y\mapsto\sum_k w_k f_k(y)$ with $w$ held fixed,
and $\sure_n^{\mathrm{fix}}(w;Y):=\sure_n(f_w)$ differentiates this
fixed-weight map. When $w$ is random, $L_n(f_w)$ denotes the realized loss $n^{-1}\|\sum_k w_k(Y)f_k(Y)-\theta\|_2^2$, whereas $\sure_n^{\mathrm{fix}}(w;Y)$ treats the supplied weight vector as fixed when taking derivatives. Since $\hat{w}$ minimizes this fixed-weight criterion over
the simplex,
\begin{align*}
L_n(\tilde f) - L_n(f_v)
&= [L_n(f_{\hat{w}}) - \sure_n^{\mathrm{fix}}(\hat w;Y)]
+ [\sure_n^{\mathrm{fix}}(\hat w;Y) - \sure_n^{\mathrm{fix}}(v;Y)] \\
&\qquad + [\sure_n^{\mathrm{fix}}(v;Y) - L_n(f_v)] \\
&\leq E^{\mathrm{fix}}(v) - E^{\mathrm{fix}}(\hat{w}),
\end{align*}
where $E^{\mathrm{fix}}(w) := \sure_n^{\mathrm{fix}}(w;Y) - L_n(f_w)$.  This is a pointwise oracle comparison for the fixed-weight family. It is distinct from \sure{} evaluation for the \sure{}-weighted average $Y\mapsto f_{\hat w(Y)}(Y)$.

\textit{Step 2 (Candidate-specific \sure{} error cancellation).}\;
We write the fixed-weight \sure{} error
$E^{\mathrm{fix}}(w)=\sure_n^{\mathrm{fix}}(w;Y)-L_n(f_w)$ as a common noise term plus a weighted average of candidate-specific terms. For any differentiable adjustment $g$, define
\[
\Psi(g):=\tr\{\Sigma Dg(Y)\}-\langle Y-\theta,g(Y)\rangle .
\]
For fixed $w \in \Delta^{K-1}$, write $g_w(y):=f_w(y)-y$, so $g_w=\sum_{k=1}^K w_k g_k$. Expanding the definitions of $L_n(f_w)$ and the fixed-weight \sure{} criterion gives the algebraic identity
\[
\sure_n(f_w)-L_n(f_w)
=
\frac{1}{n}\bigl[\tr(\Sigma) - \|\varepsilon\|_2^2\bigr]
+\frac{2}{n}\Psi(g_w).
\]
The fixed-weight criterion $\sure_n^{\mathrm{fix}}(\hat w;Y)$ is obtained by substituting the realized vector $\hat w(Y)$ into the displayed formula for the fixed map $f_w$; no derivative of the weight map $Y\mapsto\hat w(Y)$ appears in this substitution. By the linearity of $\Psi(g)$ in $g$, the fixed-weight expansion applies both to the deterministic comparison weight $v$ and to the realized selected weight vector $\hat w(Y)$:
\[
E^{\mathrm{fix}}(w) = \frac{1}{n}\bigl[\tr(\Sigma) - \|\varepsilon\|_2^2\bigr] + \frac{2}{n}\Psi(g_w)
= E_n + \frac{2}{n}\sum_{k=1}^K w_k\Psi(g_k),
\]
where $E_n := n^{-1}\{\tr(\Sigma)-\|\varepsilon\|_2^2\}$ is common across
the fixed-weight family.
Hence the common term cancels in the regret difference:
\begin{align*}
E^{\mathrm{fix}}(v) - E^{\mathrm{fix}}(\hat w)
&= \frac{2}{n}\sum_{k=1}^K (v_k - \hat w_k)\Psi(g_k) \\
&\leq
\frac{2}{n}\left\{
\sum_{k=1}^K v_k \max_{1\leq j\leq K}\Psi(g_j)
-\sum_{k=1}^K \hat w_k \min_{1\leq j\leq K}\Psi(g_j)
\right\} \\
&=
\frac{2}{n}\left\{
\max_{1 \leq k \leq K}\Psi(g_k)
- \min_{1 \leq k \leq K}\Psi(g_k)
\right\}
&& \text{since } v,\hat w\in\Delta^{K-1} \\
&\leq \frac{4}{n} \max_{1 \leq k \leq K} \left|\Psi(g_k)\right|.
\end{align*}
Substituting this bound into the Step~1 inequality gives
\[
L_n(\tilde f) - L_n(f_v)
\leq \frac{4}{n} \max_{1 \leq k \leq K} \left|\Psi(g_k)\right|.
\]
Because the display holds for every $v\in\Delta^{K-1}$,
\[
L_n(\tilde f)-\inf_{v\in\Delta^{K-1}}L_n(f_v)
\leq
\frac{4}{n} \max_{1 \leq k \leq K} \left|\Psi(g_k)\right|.
\]
The infimum is the realized-loss benchmark for the best fixed convex average
$f_v(y)=\sum_k v_k f_k(y)$ of the trained candidate maps.

\textit{Step 3 (Per-candidate sub-Weibull bounds).}\;
Step~2 reduces the realized-loss gap to the maximum of the candidate-specific
terms $\Psi(g_k)$. To prove the proposition, it remains to
take expectations and control this maximum over the $K$ trained candidates.
We first bound each candidate-specific term in $L^p$.  Step~4 then
turns those one-dimensional tail bounds into a bound for the maximum over
candidates. For each candidate $f_k$, the order-zero, first-derivative case of
the divergence-continuity bound used
in Step~1 of the proof of Theorem~\ref{thm:general_conc} gives
\[
\|\Psi(g_k)\|_{L^p} \lesssim p^4\, \|g_k\|_{W^{1,p}(P_Y)},
\]
where the hidden constant absorbs the eigenvalue bound on $\Sigma$ from
Assumption~\ref{asm:sampling_array}. Assumption~\ref{asm:averaging_regularity}
bounds the Sobolev norm of $g_k$ directly:
\[
\|g_k\|_{W^{1,p}(P_Y)}
=
\|g_k\|_{L^p(P_Y)}+\|Dg_k\|_{L^p(P_Y)}
\leq
2\big\|\|g_k(\cdot)\|_W\big\|_{L^p(P_Y)}
\leq 2\nu_n^{(k)} p^{\beta_k}.
\]
Combining the divergence-continuity bound with the Sobolev-envelope bound,
\[
\|\Psi(g_k)\|_{L^p} \leq C\, p^{4+\beta_k}\, \nu_n^{(k)},
\qquad p\geq2.
\]
Since $\beta_k\leq\bar\beta$ and $\nu_n^{(k)}\leq\bar\nu_n$,
\[
\|\Psi(g_k)\|_{L^p}\leq C\,p^{4+\bar\beta}\bar\nu_n,
\qquad k=1,\ldots,K,\quad p\geq2.
\]
The moment-to-$\psi_{\bar\alpha}$ implication in
\citet[Theorem~2.1]{vladimirovaSubWeibullDistributionsGeneralizing2020}, applied with
$\bar\alpha:=1/(4+\bar\beta)$, gives
\[
\|\Psi(g_k)\|_{\psi_{\bar\alpha}}\leq C\bar\nu_n,
\qquad k=1,\ldots,K .
\]

\textit{Step 4 (Union bound over the finite candidate collection).}\;
The resulting $\psi_{\bar\alpha}$-norm bound from Step~3 gives, for some constant $C$,
\[
\Pr\!\left(\left|\Psi(g_k)\right|\geq t\right)
\leq
2\exp\!\left[-\left(\frac{t}{C\bar\nu_n}\right)^{\bar\alpha}\right],
\qquad t\geq0.
\]
Let
\[
M:=\max_{1\leq k\leq K}\left|\Psi(g_k)\right|,
\qquad
s:=C\bar\nu_n.
\]
Because
$\{M\geq t\}=\bigcup_{k=1}^K\{|\Psi(g_k)|\geq t\}$, the union bound gives
\[
\Pr(M\geq t)
\leq
2K\exp\!\left[-\left(\frac{t}{s}\right)^{\bar\alpha}\right].
\]
Combining this inequality with the trivial bound $\Pr(M\geq t)\leq1$,
\[
\Pr(M\geq t)
\leq
\min\left\{1,\;2K\exp\!\left[-\left(\frac{t}{s}\right)^{\bar\alpha}\right]\right\}.
\]
Set $t_0:=s\{\log(2K)\}^{1/\bar\alpha}$.  Splitting the tail integral at
$t_0$ gives
\begin{align*}
\E[M]
&=
\int_0^\infty \Pr(M\geq t)\,dt \\
&\leq
\int_0^{t_0}1\,dt
+2K\int_{t_0}^{\infty}
\exp\!\left[-\left(\frac{t}{s}\right)^{\bar\alpha}\right]\,dt \\
&=
t_0
+\frac{2Ks}{\bar\alpha}
\int_{\log(2K)}^\infty
e^{-r}r^{1/\bar\alpha-1}\,dr \\
&\leq
s\{\log(2K)\}^{1/\bar\alpha}
+C_B\,s\,2K e^{-\log(2K)}
\{1+\log(2K)\}^{1/\bar\alpha-1} \\
&\leq
C_B\,s\{\log(eK)\}^{1/\bar\alpha},
\end{align*}
where the equality uses the same change of variables as
Lemma~\ref{lem:tail_to_expectation}, namely
$r=(t/s)^{\bar\alpha}$. The term
$s\{\log(2K)\}^{1/\bar\alpha}$ is the cutoff contribution $t_0$. The
tail-integral bound uses
$\int_x^\infty e^{-r}r^q\,dr\leq C_q e^{-x}(1+x)^q$ with
$x=\log(2K)$ and $q=1/\bar\alpha-1$, so the factor
$2K e^{-\log(2K)}$ equals one. The convention
$\log(eK)=1+\log K$ keeps the logarithmic factor bounded away from zero when
$K=1$; for $K\geq1$, both $\log(2K)$ and $1+\log(2K)$ are bounded by constants
times $\log(eK)$, and $\bar\alpha\in(0,1]$ gives
$1/\bar\alpha-1\leq1/\bar\alpha$. Here $C_B$ depends only on the proposition's
upper bound $\bar\beta\leq B$.
Using $s=C\bar\nu_n$ and $\bar\alpha=1/(4+\bar\beta)$, the preceding display gives
\[
\E[M]
\lesssim
s\{\log(eK)\}^{1/\bar\alpha}
\lesssim
(\log(eK))^{4+\bar\beta}\,\bar\nu_n.
\]
Multiplying this maximal bound by the $4/n$ factor in the pointwise comparison
from Step~2 proves the stated oracle inequality against the fixed convex
combination with the smallest realized loss:
\[
\E\!\left[
L_n(\tilde f)-\min_{w\in\Delta^{K-1}}L_n(f_w)
\right]
\lesssim
\frac{\bar\nu_n(\log(eK))^{4+\bar\beta}}{n}.
\]
The final comparison to the best individual candidate follows because the
simplex contains the vertices $e_1,\ldots,e_K$:
\[
\min_{w\in\Delta^{K-1}}L_n(f_w)
\leq
\min_{1\leq k\leq K}L_n(f_k).
\]
The final display gives the comparison to the best individual candidate recorded in the remark following Proposition~\ref{prop:averaging}.
\end{proof}

\subsection{\sure{} under Noise-Covariance Misspecification}\label{app:noise_covariance_misspecification}

The main text assumes that \sure{} is computed with the true sampling covariance $\Sigma$.  In applications, researchers may instead compute \sure{} with an approximate noise covariance $\widehat\Sigma$, most commonly the diagonal matrix of marginal variances.  The following calculation isolates the resulting bias when $\widehat\Sigma$ does not depend on $Y$.

\begin{proposition}[\sure{} under noise-covariance misspecification]\label{prop:noise_covariance_misspecification}
Let $Y=\theta+\varepsilon$, $\varepsilon\sim\N(0,\Sigma)$, let $\widehat\Sigma$ be fixed with respect to $Y$, and let $f(y)=y+g(y)$ be weakly differentiable with $g(Y)\in L^2(P_Y)$ and
\[
\E\left[\sum_{i,j}\bigl(|\Sigma_{ij}|+|\widehat\Sigma_{ij}|\bigr)
|\partial_j g_i(Y)|\right]<\infty ,
\]
where $\partial_j g_i$ denotes the weak derivative.  Define
\[
\sure_n(f;\widehat\Sigma)
:=
\frac{1}{n}\|Y-f(Y)\|_2^2
-\frac{1}{n}\tr(\widehat\Sigma)
+\frac{2}{n}\tr\!\left(\widehat\Sigma Df(Y)\right).
\]
Then
\[
\E\!\left[\sure_n(f;\widehat\Sigma)-L_n(f)\right]
=
\frac{1}{n}\tr(\widehat\Sigma-\Sigma)
+\frac{2}{n}\E\!\left[
\tr\!\left\{(\widehat\Sigma-\Sigma)Dg(Y)\right\}
\right].
\]
In particular, if $\widehat\Sigma=\diag(\Sigma)$, then
$n^{-1}\tr(\widehat\Sigma-\Sigma)=0$ and
\[
\E\!\left[\sure_n(f;\widehat\Sigma)-L_n(f)\right]
=
-\frac{2}{n}\E\!\left[
\tr\!\left\{\Sigma_{\mathrm{off}}Dg(Y)\right\}
\right],
\qquad
\Sigma_{\mathrm{off}}:=\Sigma-\diag(\Sigma).
\]
\end{proposition}

\begin{proof}[Proof of Proposition~\ref{prop:noise_covariance_misspecification}]
Since $f(y)=y+g(y)$,
\[
n\,\sure_n(f;\widehat\Sigma)
=
\|g(Y)\|_2^2+\tr(\widehat\Sigma)+2\tr(\widehat\Sigma Dg(Y)).
\]
Also
\[
n\,L_n(f)=\|\varepsilon+g(Y)\|_2^2
=\|\varepsilon\|_2^2+2\langle\varepsilon,g(Y)\rangle+\|g(Y)\|_2^2.
\]
Subtracting and taking expectations gives
\[
\E[n\{\sure_n(f;\widehat\Sigma)-L_n(f)\}]
=
\tr(\widehat\Sigma)-\tr(\Sigma)
+2\E\!\left[\tr(\widehat\Sigma Dg(Y))-\langle\varepsilon,g(Y)\rangle\right].
\]
By Stein's lemma, the integrability assumption lets us write $\E[\langle\varepsilon,g(Y)\rangle]=\E[\tr\{\Sigma Dg(Y)\}]$, which gives the first display.  The diagonal noise-covariance approximation has $\tr\{\diag(\Sigma)-\Sigma\}=0$ and $\widehat\Sigma-\Sigma=-\Sigma_{\mathrm{off}}$.
\end{proof}

\begin{corollary}[Diagonal covariance and finite-candidate comparisons]\label{cor:diagonal_covariance}
Consider fixed linear candidates $f_k(Y)=S_kY$, $k=1,\ldots,K$, and write
\[
G_k:=\sure_n(f_k;\Sigma)-L_n(f_k),
\qquad
B_k:=-\frac{2}{n}\tr(\Sigma_{\mathrm{off}}S_k).
\]
Let
\[
\hat k\in\argmin_{1\leq k\leq K}\sure_n(f_k;\diag(\Sigma)),
\qquad
k^\star\in\argmin_{1\leq k\leq K}L_n(f_k).
\]
Then
\[
\E\!\left[L_n(f_{\hat k})-L_n(f_{k^\star})\right]
\leq
2\E\max_{1\leq k\leq K}|G_k|
+
\left\{\max_{1\leq k\leq K}B_k-\min_{1\leq k\leq K}B_k\right\}.
\]
Equivalently, the additional comparison distortion from using the diagonal
noise-covariance approximation is
\[
\max_{k,\ell}
\left|
\frac{2}{n}\tr\{\Sigma_{\mathrm{off}}(S_k-S_\ell)\}
\right|.
\]
\end{corollary}

\begin{proof}[Proof of Corollary~\ref{cor:diagonal_covariance}]
For a fixed linear candidate, $g_k(Y)=(S_k-I)Y$, so
$Dg_k=S_k-I$.  Proposition~\ref{prop:noise_covariance_misspecification}
therefore gives
\[
\E\!\left[\sure_n(f_k;\diag(\Sigma))-L_n(f_k)\right]
=
-\frac{2}{n}\tr\{\Sigma_{\mathrm{off}}(S_k-I)\}
=
B_k,
\]
because $\Sigma_{\mathrm{off}}$ has zero diagonal.  Moreover, for fixed linear
candidates,
\[
\sure_n(f_k;\diag(\Sigma))-\sure_n(f_k;\Sigma)
=
\frac{2}{n}\tr\{(\diag(\Sigma)-\Sigma)S_k\}
=
B_k,
\]
so the diagonal-covariance correction is deterministic, and hence
\[
\sure_n(f_k;\diag(\Sigma))-L_n(f_k)=G_k+B_k.
\]
Since $\hat k$ minimizes $\sure_n(f_k;\diag(\Sigma))$,
\[
L_n(f_{\hat k})-L_n(f_{k^\star})
\leq
\{\sure_n(f_{k^\star};\diag(\Sigma))-L_n(f_{k^\star})\}
-
\{\sure_n(f_{\hat k};\diag(\Sigma))-L_n(f_{\hat k})\}.
\]
Substituting the preceding decomposition gives
\[
L_n(f_{\hat k})-L_n(f_{k^\star})
\leq
(G_{k^\star}-G_{\hat k})+(B_{k^\star}-B_{\hat k}).
\]
The first difference is bounded by $2\max_k |G_k|$, and the second is bounded
by $\max_k B_k-\min_k B_k$.  Taking expectations proves the first display.
The second display follows from
\[
B_k-B_\ell
=
-\frac{2}{n}\tr\{\Sigma_{\mathrm{off}}(S_k-S_\ell)\}.
\]
\end{proof}

\begin{remark}[Interpretation for spatial shrinkage]
Diagonal-covariance \sure{} can be biased in levels and still compare
estimators well if the omitted-covariance correction is nearly common across
candidates.  It is most concerning when candidates differ sharply in cross-unit
smoothing on pairs whose sampling errors are correlated.  For example, among the fixed linear
candidates of Corollary~\ref{cor:diagonal_covariance}, a componentwise rule has
no off-diagonal smoothing contribution, while a spatial smoother can have a
nonzero correction.  Comparisons across those types
are therefore more sensitive to omitted off-diagonal covariance than comparisons
among smoothers with similar cross-unit weighting patterns.
\end{remark}

\subsection{\sure{} Evaluation of the Weighted Average}\label{app:averaging_adaptive_weights}

Proposition~\ref{prop:averaging} is a fixed-weight oracle comparison.  The
estimator reported in the application is instead the \sure{}-weighted average
\[
\tilde f(Y)=\sum_{k=1}^K \hat w_k(Y) f_k(Y),
\]
where the weights are selected by the simplex-constrained quadratic program (QP) of Section~\ref{sec:model_averaging}.  Valid \sure{} evaluation
of $\tilde f$ requires differentiating the full map, including the
dependence of $\hat w(Y)$ on $Y$.  The next proposition records a
sufficient condition under which $\tilde f$ has the derivative
needed for valid \sure{} evaluation.  The continuous-differentiability
hypothesis on the weight map is restrictive for QP solution paths.  See the
remark following the proof.

\begin{proposition}[Regularity of the weighted average]\label{prop:adaptive_averaging_validity}
Let $Y=\theta+\varepsilon$, $\varepsilon\sim\N(0,\Sigma)$, with
$\|\Sigma\|_{\mathrm{op}}\leq\bar\sigma^2<\infty$, and let
$f_k(y)=y+g_k(y)$, $k=1,\ldots,K$, for finite $K$, where each
$g_k:\bbR^n\to\bbR^n$ is continuously
differentiable.  Let $\hat w_k:\bbR^n\to[0,1]$ be continuously
differentiable weights with $\sum_{k=1}^K\hat w_k(y)=1$ for every $y$.
Define
\[
\tilde f(y)=\sum_{k=1}^K \hat w_k(y)f_k(y).
\]
Then
\[
\tilde f(y)=y+\tilde g(y),
\qquad
\tilde g(y)=\sum_{k=1}^K \hat w_k(y)g_k(y),
\]
and $\tilde g$ is continuously differentiable.
Assume also that, for each $k$,
\[
g_k(Y)\in L^2(P_Y),
\qquad
\E\!\left[\sum_{i,j}|\Sigma_{ij}\partial_j g_{k,i}(Y)|\right]<\infty .
\]
Finally, assume the weight maps are uniformly Lipschitz in $y$: there is a
finite constant $L_w$ such that, for every $k$,
\[
|\hat w_k(y)-\hat w_k(y')|
\leq
L_w\|y-y'\|_2,
\qquad y,y'\in\bbR^n .
\]
Then \sure{} is unbiased for the risk of $\tilde f$:
\[
\E\{\sure_n(\tilde f;\Sigma)(Y)\}
=
\E\{L_n(\tilde f)\}.
\]
\end{proposition}

\begin{proof}[Proof of Proposition~\ref{prop:adaptive_averaging_validity}]
Since the weights sum to one,
\[
\tilde f(y)
=
\sum_{k=1}^K \hat w_k(y)\{y+g_k(y)\}
=
y+\sum_{k=1}^K \hat w_k(y)g_k(y).
\]
Thus $\tilde g(y)=\sum_k\hat w_k(y)g_k(y)$.  Here $D\hat w_k(y)$ is the
derivative of the scalar weight map $\hat w_k$.  Each $g_k$ is continuously
differentiable and each $\hat w_k$ is continuously differentiable, so the
product rule gives, for every $y$,
\[
D\{\hat w_k(y)g_k(y)\}
=
\hat w_k(y)Dg_k(y)+g_k(y)D\hat w_k(y).
\]
Therefore
\[
D\tilde g(y)
=
\sum_{k=1}^K
\left\{
\hat w_k(y)Dg_k(y)
+g_k(y)D\hat w_k(y)
\right\}.
\]
Since $0\leq\hat w_k\leq1$ and $K<\infty$, Minkowski's inequality gives
\[
\left(\E[\|\tilde g(Y)\|_2^2]\right)^{1/2}
\leq
\sum_{k=1}^K
\left(\E[\|g_k(Y)\|_2^2]\right)^{1/2}
<\infty .
\]
Thus $\tilde g(Y)\in L^2(P_Y)$.

For the derivative integrability condition, the formula for $D\tilde g$
and the triangle inequality give
\[
\begin{aligned}
\sum_{i,j}|\Sigma_{ij}\partial_j\tilde g_i(Y)|
&\leq
\sum_{k=1}^K
\sum_{i,j}|\Sigma_{ij}\hat w_k(Y)\partial_j g_{k,i}(Y)|  \\
&\quad+
\sum_{k=1}^K
\sum_{i,j}|\Sigma_{ij}|
|g_{k,i}(Y)\partial_j\hat w_k(Y)|.
\end{aligned}
\]
The first term is bounded by
\[
\sum_{k=1}^K
\sum_{i,j}|\Sigma_{ij}\partial_j g_{k,i}(Y)|,
\]
because $0\leq \hat w_k\leq1$.  Since $\hat w_k$ is continuously
differentiable and uniformly Lipschitz with constant $L_w$,
$|\partial_j\hat w_k(y)|\leq L_w$ for every $j$ and $y$.  Therefore the
second term is bounded by
\[
L_w
\sum_{k=1}^K
\sum_{i,j}|\Sigma_{ij}|\,|g_{k,i}(Y)|.
\]
The first bound has finite expectation by the assumed derivative-integrability condition on
each $g_k$.  For the second bound, $\|\Sigma\|_{\mathrm{op}}\leq
\bar\sigma^2$ implies $|\Sigma_{ij}|\leq\bar\sigma^2$, so
\[
L_w
\sum_{k=1}^K
\sum_{i,j}|\Sigma_{ij}|\,|g_{k,i}(Y)|
\leq
L_w\bar\sigma^2 n
\sum_{k=1}^K
\sum_{i=1}^n |g_{k,i}(Y)|.
\]
Taking expectations in the preceding bound gives a finite bound: $K$
and $n$ are finite, and $g_k(Y)\in L^2(P_Y)$ implies
$g_{k,i}(Y)\in L^1(P_Y)$ for each coordinate $i$.  Thus $\tilde g$
satisfies the same derivative-integrability condition.  Since
$\tilde f(y)=y+\tilde g(y)$, the \sure{} unbiasedness identity of
Section~\ref{sec:sure} applies under the stated integrability conditions,
which gives the displayed unbiasedness statement.
\end{proof}

\begin{remark}[Smoothness of the QP weight map]\label{rem:qp_weight_smoothness}
Proposition~\ref{prop:adaptive_averaging_validity} establishes \sure{}
unbiasedness for the risk of $\tilde f$ under continuous differentiability
and a uniform Lipschitz bound on the selected weights.  The selected
weights can fail to be differentiable in $Y$ where the set of candidates
receiving positive weight changes.  The estimator evaluated by \sure{} is
the \sure{}-weighted average $\tilde f(Y)=f_{\hat w(Y)}(Y)$, where
$f_w(Y)=\sum_{k=1}^K w_k f_k(Y)$.  When different weights $w$ on the
simplex give different values of $f_w(Y)$, the selected weights
$\hat w(Y)$ are unique, and so is $\tilde f(Y)$.  The weights can fail to
be unique only when two different weight vectors give the same value of
$f_w(Y)$, and all minimizing weights then give that same value
(Section~\ref{sec:averaging_guarantee}).  The \sure{}-weighted average
$\tilde f(Y)$ is therefore insensitive to which minimizing weights are
selected.
\end{remark}

\begin{remark}[Per-candidate vs.\ within-class bounds]
When each $f_k = f_{\hat\gamma_k(Y)}^{(k)}$ is trained from a family satisfying Assumption~\ref{asm:regularity} with parameters $(\nu_n^{(k)}, \beta_k, d_k)$, the within-class uniform result gives a sufficient but typically conservative bound on the centered \sure{}-error variation
\[
\bar E_k
:=
\{\sure_n(f_k)-L_n(f_k)\}
-
n^{-1}\{\tr(\Sigma)-\|\varepsilon\|_2^2\}
=
(2/n)\Psi(g_k)
\]
with scale $\nu_n^{(k)}\,\max\{d_k,1\}^{4+\beta_k}/n$, where the $\max\{d_k,1\}^{4+\beta_k}$ factor is the covering cost for the $d_k$-dimensional parameter space, including the singleton case. The omitted centering term is common across candidates and cancels in the oracle and averaging comparisons. The per-candidate condition $\big\|\|g_k\|_W\big\|_{L^p} \leq \nu_n^{(k)}\,p^{\beta_k}$ evaluates regularity at the single composite function $y \mapsto f_{\hat\gamma_k(y)}^{(k)}(y)$ (including implicit differentiation through $\hat\gamma_k$), giving moment envelope $M_{P_Y,0}^{\beta_k}(g_k) \leq 2\nu_n^{(k)}$ and sub-Weibull scale $\nu_n^{(k)}/n$ for the centered increment $(2/n)\Psi(g_k)$ without the covering cost. The within-class result is the route used here for the within-class oracle inequality (comparing $f_{\hat\gamma_k}$ to $f_{\gamma_k^*}$), but the averaging step does not require this uniform control.
\end{remark}

\section{Regularity Verification for Estimators}\label{app:regularity_verification}

This appendix gives primitive sufficient conditions under which the estimator forms used in the empirical library satisfy the per-candidate regularity condition used by Proposition~\ref{prop:averaging} (Assumption~\ref{asm:averaging_regularity}).  The verification is procedural: first record training-stability and closure tools, then check the nontrivial estimator forms, and finally verify the fixed value-similarity building block used in the Cook County comparison.  For the value-similarity building block the verification holds under a bounded row-sum condition on the fixed geographic factor of the kernel, with envelope scale $\nu_n=O(\sqrt n)$ matching the other building blocks.

\subsection{Trained-Parameter Regularity from Training Stability}\label{app:learned_parameter_regularity}

We begin with trained candidates: primitive conditions on the base map and the training rule under which the composite $y\mapsto f_{\hat\gamma(y)}(y)$ satisfies Assumption~\ref{asm:averaging_regularity}.  This does not replace the within-class uniform bounds needed for Theorem~\ref{thm:uniform_conc}.  Those require separate control of $g_\gamma-g_{\gamma'}$ uniformly over $\gamma,\gamma'\in\Gamma$.  Write a trained candidate as
\[
f(y)=f_{\hat\gamma(y)}(y),
\qquad
g(y)=f_{\hat\gamma(y)}(y)-y=g_{\hat\gamma(y)}(y),
\]
where $g_\gamma(y)=f_\gamma(y)-y$ is the base shrinkage adjustment.  The target is the moment bound in Assumption~\ref{asm:averaging_regularity},
\begin{equation}
\left(\E[\|g(Y)\|_W^p]\right)^{1/p}\leq \nu_n p^\beta,
\qquad p\geq2.
\label{eq:per_estimator_regularity}
\end{equation}
Lemma~\ref{lem:learned_parameter_regularity} separates this target into two transparent pieces: regularity of the base map $(y,\gamma)\mapsto g_\gamma(y)$ along the trained path, and stability of the training rule $Y\mapsto\hat\gamma(Y)$ as measured by $D_y\hat\gamma(Y)$.
Thus, for any trained candidate in the empirical library, verification can be done in two steps: first verify the base bounds along the trained path $\{(Y,\hat\gamma(Y))\}$, then verify either the exact-training sensitivity bound in Proposition~\ref{prop:learned_parameter_ift} or the finite-step sensitivity bound in Proposition~\ref{prop:learned_parameter_unrolling}.

Throughout this subsection, Jacobians use output coordinates as rows:
\[
D_y g_\gamma\in\bbR^{n\times n},\qquad
D_\gamma g_\gamma\in\bbR^{n\times d},\qquad
D_y\hat\gamma\in\bbR^{d\times n}.
\]
Here $D_y g_{\hat\gamma(y)}(y)$ denotes the partial derivative of $(y,\gamma)\mapsto g_\gamma(y)$ with respect to $y$, holding $\gamma$ fixed and then evaluating at $\gamma=\hat\gamma(y)$.  The full derivative of the trained composite $y\mapsto g_{\hat\gamma(y)}(y)$ is denoted by $Dg(y)$.

\begin{lemma}[Trained-parameter regularity from sensitivity]\label{lem:learned_parameter_regularity}
Let $\Gamma\subseteq\bbR^d$.  Suppose $\hat\gamma:\bbR^n\to\Gamma$ is continuously differentiable and $(y,\gamma)\mapsto g_\gamma(y)$ is continuously differentiable on an open neighborhood of the trained path $\{(y,\hat\gamma(y)):y\in\bbR^n\}$.  Suppose there are finite constants $K_0,K_1,K_2$ and exponents $\beta_0,\beta_1,\beta_2\geq0$ such that, for every $p\geq2$,
\begin{align}
\left\|
  \|g_{\hat\gamma(Y)}(Y)\|_2
  +\|D_y g_{\hat\gamma(Y)}(Y)\|_F
\right\|_{L^p}
&\leq K_0 p^{\beta_0}, \label{eq:learned_base_y_bound}\\
\left(\E[\|D_\gamma g_{\hat\gamma(Y)}(Y)\|_{\mathrm{op}}^{2p}]\right)^{1/(2p)}
&\leq K_1 p^{\beta_1}, \label{eq:learned_base_gamma_bound}\\
\left(\E[\|D_y\hat\gamma(Y)\|_F^{2p}]\right)^{1/(2p)}
&\leq K_2 p^{\beta_2}. \label{eq:learned_training_sensitivity_bound}
\end{align}
Then the trained adjustment $g(y)=g_{\hat\gamma(y)}(y)$ satisfies
\[
\left(\E[\|g(Y)\|_W^p]\right)^{1/p}
\leq K_0p^{\beta_0}+K_1K_2p^{\beta_1+\beta_2}
\leq (K_0+K_1K_2)p^{\max\{\beta_0,\beta_1+\beta_2\}}.
\]
Thus the per-candidate moment condition \eqref{eq:per_estimator_regularity} holds with $\nu_n=K_0+K_1K_2$ and $\beta=\max\{\beta_0,\beta_1+\beta_2\}$; since the differentiability hypotheses make $y\mapsto g_{\hat\gamma(y)}(y)$ continuously differentiable by the chain rule, the trained candidate satisfies Assumption~\ref{asm:averaging_regularity}.
\end{lemma}

\begin{proof}[Proof of Lemma~\ref{lem:learned_parameter_regularity}]
By the chain rule,
\[
Dg(y)
=
D_y g_{\hat\gamma(y)}(y)
+D_\gamma g_{\hat\gamma(y)}(y)D_y\hat\gamma(y).
\]
By the definition of the $W$-norm for the composite map,
\[
\|g(Y)\|_W
=
\|g_{\hat\gamma(Y)}(Y)\|_2+\|Dg(Y)\|_F .
\]
Combining this identity with the chain-rule display, the triangle inequality,
and $\|AB\|_F\leq\|A\|_{\mathrm{op}}\|B\|_F$ gives
\[
\|g(Y)\|_W
\leq
\|g_{\hat\gamma(Y)}(Y)\|_2
+\|D_y g_{\hat\gamma(Y)}(Y)\|_F
+\|D_\gamma g_{\hat\gamma(Y)}(Y)\|_{\mathrm{op}}\|D_y\hat\gamma(Y)\|_F.
\]
Taking $L^p$ norms and applying Minkowski's inequality gives
\begin{align*}
\left(\E[\|g(Y)\|_W^p]\right)^{1/p}
&\leq
\left\|
  \|g_{\hat\gamma(Y)}(Y)\|_2
  +\|D_y g_{\hat\gamma(Y)}(Y)\|_F
\right\|_{L^p}\\
&\quad+
\left\|
\|D_\gamma g_{\hat\gamma(Y)}(Y)\|_{\mathrm{op}}\|D_y\hat\gamma(Y)\|_F
\right\|_{L^p}.
\end{align*}
The first term is bounded by \eqref{eq:learned_base_y_bound}.  H\"older's inequality with conjugate exponents $2$ and $2$ bounds the product term by
\[
\left(\E[\|D_\gamma g_{\hat\gamma(Y)}(Y)\|_{\mathrm{op}}^{2p}]\right)^{1/(2p)}
\left(\E[\|D_y\hat\gamma(Y)\|_F^{2p}]\right)^{1/(2p)}
\leq K_1K_2p^{\beta_1+\beta_2}.
\]
Combining the two bounds gives the first displayed inequality in the lemma.
Since $p\geq2$ and the exponents are nonnegative, the second displayed
inequality follows.  This proves the stated moment condition.
\end{proof}

\begin{remark}
A Frobenius-norm bound on $D_\gamma g_{\hat\gamma(Y)}(Y)$ is also sufficient for \eqref{eq:learned_base_gamma_bound}, since $\|D_\gamma g_{\hat\gamma(Y)}(Y)\|_{\mathrm{op}}\leq \|D_\gamma g_{\hat\gamma(Y)}(Y)\|_F$.
\end{remark}

\begin{proposition}[Exact smooth training]\label{prop:learned_parameter_ift}
Let $\Gamma\subseteq\bbR^d$, and let $\mathcal L(\gamma,y)$ be the sample objective used to choose the parameter $\gamma$ from the data vector $y$.  Suppose $\hat\gamma:\bbR^n\to\Gamma$ selects interior first-order solutions in the following local sense: for every $y_0\in\bbR^n$, $\hat\gamma(y_0)\in\operatorname{int}(\Gamma)$, $\mathcal L$ is twice continuously differentiable on an open neighborhood of $(\hat\gamma(y_0),y_0)$ contained in $\operatorname{int}(\Gamma)\times\bbR^n$, and there are neighborhoods $G_{y_0}$ of $\hat\gamma(y_0)$ and $V_{y_0}$ of $y_0$ such that, for every $y\in V_{y_0}$, $\hat\gamma(y)$ is the unique $\gamma\in G_{y_0}$ satisfying the first-order condition
\[
\nabla_\gamma\mathcal L(\gamma,y)=0.
\]
Define
\[
H(y):=\nabla_{\gamma\gamma}^2\mathcal L(\hat\gamma(y),y)\in\bbR^{d\times d},
\qquad
R(y):=D_y[\nabla_\gamma\mathcal L](\hat\gamma(y),y)\in\bbR^{d\times n}.
\]
Suppose $H(y)$ is invertible on $\bbR^n$, $\|H(y)^{-1}\|_{\mathrm{op}}\leq K_H$ for a finite constant $K_H$, and there are finite constants $K_R$ and $\beta_2\geq0$ such that, for every moment order $q\geq2$,
\[
\left(\E[\|R(Y)\|_F^q]\right)^{1/q}\leq K_Rq^{\beta_2}.
\]
Then
\[
D_y\hat\gamma(y)=-H(y)^{-1}R(y)
\]
and
\[
\left(\E[\|D_y\hat\gamma(Y)\|_F^q]\right)^{1/q}\leq K_HK_Rq^{\beta_2},\qquad q\geq2.
\]
If, in addition, $(y,\gamma)\mapsto g_\gamma(y)$ is continuously differentiable on an open neighborhood of $\{(y,\hat\gamma(y)):y\in\bbR^n\}$ and there exist finite constants $K_0,K_1$ and exponents $\beta_0,\beta_1\geq0$ such that the base-map bounds \eqref{eq:learned_base_y_bound}--\eqref{eq:learned_base_gamma_bound} hold along the selected map $y\mapsto\hat\gamma(y)$, then $g(y)=g_{\hat\gamma(y)}(y)$ satisfies the per-candidate moment condition \eqref{eq:per_estimator_regularity} with $\nu_n=K_0+2^{\beta_2}K_1K_HK_R$ and $\beta=\max\{\beta_0,\beta_1+\beta_2\}$.
Moreover, a sufficient curvature condition for the inverse-Hessian bound is $H(y)\succeq mI_d$ for all $y$ and some $m>0$, in which case $K_H=m^{-1}$.
\end{proposition}

\begin{proof}[Proof of Proposition~\ref{prop:learned_parameter_ift}]
Let $F(\gamma,y):=\nabla_\gamma\mathcal L(\gamma,y)$.  Fix $y_0\in\bbR^n$.  Since $F$ is continuously differentiable near $(\hat\gamma(y_0),y_0)$ and $D_\gamma F(\hat\gamma(y_0),y_0)=H(y_0)$ is invertible, the implicit function theorem gives a neighborhood $\widetilde V_{y_0}$ of $y_0$ and a continuously differentiable map $\widetilde\gamma_{y_0}:\widetilde V_{y_0}\to\bbR^d$ such that $\widetilde\gamma_{y_0}(y_0)=\hat\gamma(y_0)$ and
\[
F(\widetilde\gamma_{y_0}(y),y)=0,\qquad y\in\widetilde V_{y_0}.
\]
After shrinking $\widetilde V_{y_0}$ if necessary, $\widetilde\gamma_{y_0}(y)$ lies in $G_{y_0}$ for every $y\in\widetilde V_{y_0}$.  The local uniqueness assumption therefore implies $\widetilde\gamma_{y_0}(y)=\hat\gamma(y)$ on $\widetilde V_{y_0}$.  Because $y_0$ was arbitrary, $\hat\gamma$ is continuously differentiable on $\bbR^n$.
Differentiating the first-order condition $F(\hat\gamma(y),y)=0$ in $y$ by the chain rule separates the change through $\hat\gamma(y)$ from the direct change in $y$:
\[
D_\gamma F(\hat\gamma(y),y)D_y\hat\gamma(y)
+D_yF(\hat\gamma(y),y)
=0.
\]
Using the definitions of $H(y)$ and $R(y)$, this derivative identity becomes
\[
H(y)D_y\hat\gamma(y)+R(y)=0.
\]
Since $H(y)$ is invertible by assumption, we can left-multiply by $H(y)^{-1}$ and solve for the derivative of the trained parameter:
\[
D_y\hat\gamma(y)=-H(y)^{-1}R(y).
\]
Taking Frobenius norms in the last display and using $\|AB\|_F\leq \|A\|_{\mathrm{op}}\|B\|_F$,
\[
\|D_y\hat\gamma(y)\|_F
\leq
\|H(y)^{-1}\|_{\mathrm{op}}\|R(y)\|_F
\leq
K_H\|R(y)\|_F.
\]
The assumed moment bound on $R(Y)$ gives
\[
\left(\E[\|D_y\hat\gamma(Y)\|_F^q]\right)^{1/q}\leq K_HK_Rq^{\beta_2}.
\]
To connect this sensitivity bound to Lemma~\ref{lem:learned_parameter_regularity}, take $q=2p$.  Then, for every $p\geq2$,
\[
\left(\E[\|D_y\hat\gamma(Y)\|_F^{2p}]\right)^{1/(2p)}
\leq K_HK_R(2p)^{\beta_2}
=2^{\beta_2}K_HK_Rp^{\beta_2}.
\]
Thus the training-sensitivity condition \eqref{eq:learned_training_sensitivity_bound} in Lemma~\ref{lem:learned_parameter_regularity} holds with $K_2=2^{\beta_2}K_HK_R$.  Under the additional differentiability and base-map assumptions in the statement, the base-map bounds \eqref{eq:learned_base_y_bound}--\eqref{eq:learned_base_gamma_bound} and the continuous differentiability of $(y,\gamma)\mapsto g_\gamma(y)$, the remaining hypotheses of Lemma~\ref{lem:learned_parameter_regularity}, also hold.  Applying Lemma~\ref{lem:learned_parameter_regularity} gives the per-candidate moment condition \eqref{eq:per_estimator_regularity} with $\nu_n=K_0+2^{\beta_2}K_1K_HK_R$ and $\beta=\max\{\beta_0,\beta_1+\beta_2\}$.\end{proof}

\begin{remark}[Convexity as a sufficient condition]
Proposition~\ref{prop:learned_parameter_ift} verifies training sensitivity from local information about the selected parameter map.  For each data vector $y$, the selected value $\hat\gamma(y)$ must be an interior solution of $\nabla_\gamma\mathcal L(\gamma,y)=0$ that is locally unique, with invertible Hessian $H(y)$ and the stated moment control.  Strong local convexity is a convenient way to check the local uniqueness and inverse-Hessian parts of these hypotheses: if, in a neighborhood of each data vector $y$, the map $\gamma\mapsto\mathcal L(\gamma,y)$ has a unique interior minimizer and satisfies $H(y)=\nabla_{\gamma\gamma}^2\mathcal L(\hat\gamma(y),y)\succeq mI_d$ for some $m>0$, then $\|H(y)^{-1}\|_{\mathrm{op}}\leq m^{-1}$, so the inverse-Hessian condition in Proposition~\ref{prop:learned_parameter_ift} holds with $K_H=m^{-1}$.  More generally, convexity can be replaced by direct verification of local uniqueness, nonsingularity of $H(y)$, and enough moment control on $\|H(Y)^{-1}R(Y)\|_F$ to verify the training-sensitivity condition \eqref{eq:learned_training_sensitivity_bound}.
\end{remark}

Proposition~\ref{prop:learned_parameter_ift} verifies the training-sensitivity requirement \eqref{eq:learned_training_sensitivity_bound} when the trained parameter is locally characterized by the first-order condition $\nabla_\gamma\mathcal L(\hat\gamma(y),y)=0$.  Proposition~\ref{prop:learned_parameter_unrolling} treats the case in which training stops after a fixed number of differentiable iterative updates.  It writes the optimization state as $s_t(y)=(\gamma_t(y),a_t(y))$ and sets $\hat\gamma(y)=\gamma_T(y)$, so the proof can track how a perturbation of $y$ propagates through the finite path $s_0(y),\ldots,s_T(y)$.  Fixed runs of stochastic-gradient and related randomized methods fit this formulation after conditioning on the random seed, the random subsets of observations used to form the stepwise objectives, and any other algorithmic randomness.  Conditional on those draws, the realized optimization path is deterministic in $y$.

\begin{proposition}[Finite-step training]\label{prop:learned_parameter_unrolling}
Let $\Gamma\subseteq\bbR^d$.  Fix a realization of any algorithmic randomness used by the iterative optimization procedure, such as random seeds or sampled data subsets used to form stochastic-gradient updates; all maps and constants below are conditional on that realization, and expectations are over $Y$.  Let
\[
s_t(y)=(\gamma_t(y),a_t(y))\in\bbR^d\times\bbR^{m-d}
\]
denote the optimization state after step $t$, for some integer $m\geq d$, where $\gamma_t(y)$ is the parameter vector and $a_t(y)$ collects any auxiliary optimization quantities, such as momentum terms or Adam first and second moments.  Suppose a fixed, non-data-adaptive number $1\leq T<\infty$ of updates satisfies
\[
s_{t+1}(y)=U_t(s_t(y),y),\qquad t=0,\ldots,T-1,
\]
and set $\hat\gamma(y)=\gamma_T(y)\in\Gamma$.  Let $s_0:\bbR^n\to\bbR^m$ be continuously differentiable, and suppose that, for each $t$, the update map $U_t:\bbR^m\times\bbR^n\to\bbR^m$ is continuously differentiable on an open neighborhood of the state-data path $\{(s_t(y),y):y\in\bbR^n\}$.  Define the data-derivative matrices
\[
\begin{aligned}
J_t(y)&:=D_y s_t(y)\in\bbR^{m\times n},\\
A_t(y)&:=D_sU_t(s,y)\big|_{s=s_t(y)}\in\bbR^{m\times m},\\
B_t(y)&:=D_yU_t(s,y)\big|_{s=s_t(y)}\in\bbR^{m\times n}.
\end{aligned}
\]
Here $D_s$ and $D_y$ denote partial derivatives of $U_t$; in the definition of $B_t(y)$, the state argument is held fixed and then evaluated at $s_t(y)$.  Assume that $\|A_t(y)\|_{\mathrm{op}}\leq M$ for all $t=0,\ldots,T-1$ and all $y\in\bbR^n$, for a finite constant $M\geq0$, and that there are finite constants $K_U$ and $\beta_2\geq0$ such that, for every $q\geq2$,
\[
\left(\E[\|J_0(Y)\|_F^q]\right)^{1/q}
+\sum_{t=0}^{T-1}
\left(\E[\|B_t(Y)\|_F^q]\right)^{1/q}
\leq K_Uq^{\beta_2}.
\]
Then, with
\[
\kappa_{T,M}:=M^T+\sum_{r=0}^{T-1}M^{T-1-r},
\]
where $M^0=1$ by convention, the training rule satisfies
\[
\left(\E[\|D_y\hat\gamma(Y)\|_F^q]\right)^{1/q}
\leq
\kappa_{T,M}K_Uq^{\beta_2},\qquad q\geq2.
\]
If, in addition, $(y,\gamma)\mapsto g_\gamma(y)$ is continuously differentiable on an open neighborhood of the selected parameter path $\{(y,\hat\gamma(y)):y\in\bbR^n\}$ and there exist finite constants $K_0,K_1$ and exponents $\beta_0,\beta_1\geq0$ such that the base-map bounds \eqref{eq:learned_base_y_bound}--\eqref{eq:learned_base_gamma_bound} hold along this path, then $g(y)=g_{\hat\gamma(y)}(y)$ satisfies the per-candidate moment condition \eqref{eq:per_estimator_regularity} with $\nu_n=K_0+2^{\beta_2}K_1\kappa_{T,M}K_U$ and $\beta=\max\{\beta_0,\beta_1+\beta_2\}$.
\end{proposition}

\begin{proof}[Proof of Proposition~\ref{prop:learned_parameter_unrolling}]
Because $s_0$ and the update maps are continuously differentiable on the relevant neighborhoods, induction on $t$ and the chain rule imply that $s_t$ is continuously differentiable for each $t=0,\ldots,T$.  In particular, $\hat\gamma(y)=\gamma_T(y)$ is continuously differentiable.  Differentiating $s_{t+1}(y)=U_t(s_t(y),y)$ with respect to $y$ gives the recursive formula
\[
J_{t+1}(y)=A_t(y)J_t(y)+B_t(y).
\]
Applying this formula successively from $t=0$ to $t=T-1$ gives
\[
J_T(y)
=
A_{T-1}(y)\cdots A_0(y)J_0(y)
+\sum_{r=0}^{T-1}A_{T-1}(y)\cdots A_{r+1}(y)B_r(y),
\]
where, for the term $r=T-1$, the product multiplying $B_{T-1}(y)$ is the identity.  For the first term,
\[
\|A_{T-1}(Y)\cdots A_0(Y)J_0(Y)\|_F
\leq
M^T\|J_0(Y)\|_F,
\]
and for each $r=0,\ldots,T-1$,
\[
\|A_{T-1}(Y)\cdots A_{r+1}(Y)B_r(Y)\|_F
\leq
M^{T-1-r}\|B_r(Y)\|_F,
\]
using the assumed bound $\|A_t(y)\|_{\mathrm{op}}\leq M$ and $\|AB\|_F\leq\|A\|_{\mathrm{op}}\|B\|_F$.  Therefore the triangle inequality gives
\[
\|J_T(Y)\|_F
\leq
M^T\|J_0(Y)\|_F
+\sum_{r=0}^{T-1}M^{T-1-r}\|B_r(Y)\|_F.
\]
Taking $L^q$ norms and applying Minkowski's inequality,
\begin{align*}
\left(\E[\|J_T(Y)\|_F^q]\right)^{1/q}
&\leq
M^T
\left(\E[\|J_0(Y)\|_F^q]\right)^{1/q}
+\sum_{r=0}^{T-1}
M^{T-1-r}
\left(\E[\|B_r(Y)\|_F^q]\right)^{1/q}\\
&\leq
\kappa_{T,M}
\left\{
\left(\E[\|J_0(Y)\|_F^q]\right)^{1/q}
+\sum_{r=0}^{T-1}
\left(\E[\|B_r(Y)\|_F^q]\right)^{1/q}
\right\}\\
&\leq \kappa_{T,M}K_Uq^{\beta_2},
\end{align*}
where the second inequality uses that each nonnegative coefficient is bounded by $\kappa_{T,M}$, and the last inequality is the displayed moment assumption in the proposition.  Since $\hat\gamma(y)=\gamma_T(y)$, the matrix $D_y\hat\gamma(y)$ is the block of rows of $J_T(y)$ corresponding to the parameter component of the optimization state.  Hence $\|D_y\hat\gamma(y)\|_F\leq\|J_T(y)\|_F$, and
\[
\left(\E[\|D_y\hat\gamma(Y)\|_F^q]\right)^{1/q}
\leq \kappa_{T,M}K_Uq^{\beta_2}.
\]
To connect this bound to Lemma~\ref{lem:learned_parameter_regularity}, take $q=2p$.  Then, for every $p\geq2$,
\[
\left(\E[\|D_y\hat\gamma(Y)\|_F^{2p}]\right)^{1/(2p)}
\leq
\kappa_{T,M}K_U(2p)^{\beta_2}
=
2^{\beta_2}\kappa_{T,M}K_U p^{\beta_2}.
\]
Thus the training-sensitivity condition \eqref{eq:learned_training_sensitivity_bound} holds with $K_2=2^{\beta_2}\kappa_{T,M}K_U$.  Under the additional differentiability and base-map assumptions in the statement, the base-map bounds \eqref{eq:learned_base_y_bound}--\eqref{eq:learned_base_gamma_bound} and the continuous differentiability of $(y,\gamma)\mapsto g_\gamma(y)$, the remaining hypotheses of Lemma~\ref{lem:learned_parameter_regularity}, also hold.  Applying Lemma~\ref{lem:learned_parameter_regularity} gives the per-candidate moment condition \eqref{eq:per_estimator_regularity} with $\nu_n=K_0+2^{\beta_2}K_1\kappa_{T,M}K_U$ and $\beta=\max\{\beta_0,\beta_1+\beta_2\}$.
\end{proof}

\begin{remark}[Applying the result to gradient descent]
For a plain gradient step with no auxiliary optimization state,
\[
U_t(\gamma,y)=\gamma-\eta_t\nabla_\gamma \ell_t(\gamma,y).
\]
Here $\ell_t(\gamma,y)$ denotes the objective used in update $t$.  For
full-sample gradient descent, $\ell_t$ may be the same sample criterion at every
step; for stochastic gradient descent, $\ell_t$ is the criterion formed from the
observations sampled at step $t$, after conditioning on that sampling.  The subscript
$t$ allows the objective, any penalty term, the sampled data subset, and the
step size $\eta_t$ to vary across updates.  In this case, the matrices in Proposition~\ref{prop:learned_parameter_unrolling} are
\[
A_t(y)=I-\eta_t\nabla_{\gamma\gamma}^2\ell_t(\gamma_t(y),y),
\qquad
B_t(y)
=
-\eta_t
D_y\{\nabla_\gamma\ell_t(\gamma,y)\}\big|_{\gamma=\gamma_t(y)},
\]
where the derivative in $B_t(y)$ is taken with respect to $y$, holding $\gamma$
fixed, and then evaluated at $\gamma=\gamma_t(y)$.  Thus $A_t(y)$ propagates
existing data sensitivity in $\gamma_t(y)$, while $B_t(y)$ is the direct
derivative of the step with respect to $y$.

For momentum or Adam-style optimization, Proposition~\ref{prop:learned_parameter_unrolling}
should be applied to the full state $s_t(y)=(\gamma_t(y),a_t(y))$.  The
auxiliary block $a_t(y)$ contains
the deterministic quantities carried from one update to the next, such as
momentum terms, running first and second moments, or learning-rate state.  After
conditioning on seeds and the sampled observations or data subsets used in the
stochastic-gradient steps, these quantities are deterministic
functions of $y$, and the realized state-data path is
$\{(s_t(y),y):y\in\bbR^n,\ t=0,\ldots,T-1\}$.  The moment assumption in
Proposition~\ref{prop:learned_parameter_unrolling} requires the state derivative
matrices $A_t(y)=D_sU_t(s_t(y),y)$ to have bounded operator norm on this path and
the data-derivative matrices $B_t(y)=D_yU_t(s_t(y),y)$ to have controlled
moments.  These are derivatives of the full state update $U_t(s,y)$, not only
of the parameter update for $\gamma_t(y)$.  Nonsmooth operations such as
clipping or projection require smoothing or a separate nonsmooth argument.
Likewise, if the number of updates is data-adaptive, $T=T(Y)$, the resulting
stopping rule requires its own smoothness and moment verification rather than a
direct application of Proposition~\ref{prop:learned_parameter_unrolling}.
\end{remark}

\begin{remark}[Interpreting $\kappa_{T,M}$]
Propositions~\ref{prop:learned_parameter_ift} and
\ref{prop:learned_parameter_unrolling} both verify the same
training-sensitivity requirement by bounding $D_y\hat\gamma(Y)$.  In the exact
first-order-condition case, differentiating
$\nabla_\gamma\mathcal L(\hat\gamma(y),y)=0$ gives
\[
D_y\hat\gamma(y)=-H(y)^{-1}R(y).
\]
In the finite-step case, differentiating the update path gives
\[
J_{t+1}(y)=A_t(y)J_t(y)+B_t(y),
\]
so the derivative of the trained parameter is controlled by repeated
multiplication by the matrices $A_t(y)$ and by the direct data derivatives
$B_t(y)$.  The constant $\kappa_{T,M}=M^T+\sum_{r=0}^{T-1}M^{T-1-r}$ is the resulting
bound on the accumulated products of the matrices $A_t(y)$ under the assumption
$\|A_t(y)\|_{\mathrm{op}}\leq M$.  If $M<1$, $\kappa_{T,M}$ remains bounded as $T$
increases; if $M=1$, $\kappa_{T,M}=T+1$; and if $M>1$, the term $M^T$ grows
geometrically in $T$.  Thus the finite-step sufficient condition yields a large
bound for $D_y\hat\gamma(Y)$ when $M$ is close to or above one and $T$ is large,
or when the quantities controlled by $K_U$, namely $J_0(Y)$ and $B_t(Y)$, are
large.  For gradient descent,
\[
A_t(y)=I-\eta_t\nabla_{\gamma\gamma}^2\ell_t(\gamma_t(y),y),
\]
so curvature enters the finite-step sensitivity bound through the Hessian inside
$A_t(y)$, just as curvature enters the exact first-order-condition calculation
through the inverse Hessian $H(y)^{-1}$.
\end{remark}

\subsection{Regularity Closure for Candidate Constructions}\label{app:regularity_closure}

Throughout this subsection, for $u,v\in\bbR^n$, $u\odot v\in\bbR^n$
denotes componentwise multiplication.

The estimator-form verifications in Section~\ref{sec:verification} use
shrinkage adjustments built from simpler maps of the data.  This subsection
records the closure facts used in those verifications: polynomial
bounds on a map and its Jacobian are preserved under the fixed affine maps,
componentwise products, compositions, and standardization maps used to construct
shrinkage adjustments.

\begin{definition}[Pointwise polynomial regularity]\label{def:pointwise_polynomial_regularity}
A continuously differentiable map $g:\bbR^n\to\bbR^m$ satisfies the pointwise polynomial
regularity condition if, for some constants $C<\infty$ and $0\leq a<\infty$,
\[
\|g(y)\|_2+\|Dg(y)\|_F\leq C(1+\|y\|_2)^a,\qquad y\in\bbR^n .
\]
\end{definition}

\begin{lemma}[Affine maps and finite sums]\label{lem:closure_affine_sum}
Let $g_1,\ldots,g_J:\bbR^n\to\bbR^n$ satisfy the pointwise polynomial
regularity condition in Definition~\ref{def:pointwise_polynomial_regularity}.
Let $A_1,\ldots,A_J\in\bbR^{n\times n}$ and $b\in\bbR^n$ be fixed.  Then
\[
g(y)=b+\sum_{j=1}^J A_j g_j(y)
\]
satisfies the condition in Definition~\ref{def:pointwise_polynomial_regularity}.
In particular, any fixed affine preprocessing map $y\mapsto Ay+b$, including
residualization by a fixed projection matrix, also satisfies the condition in
Definition~\ref{def:pointwise_polynomial_regularity}.
\end{lemma}

\begin{proof}[Proof of Lemma~\ref{lem:closure_affine_sum}]
For each $j$, choose $C_j<\infty$ and $0\leq a_j<\infty$ such that
\[
\|g_j(y)\|_2+\|Dg_j(y)\|_F\leq C_j(1+\|y\|_2)^{a_j}.
\]
Let $a:=\max_j a_j$.  By linearity of differentiation,
\[
Dg(y)=\sum_{j=1}^J A_jDg_j(y),
\]
where the matrices $A_j$ are fixed.  The triangle inequality and
submultiplicativity give separate bounds for the level of $g$ and for its
Jacobian:
\[
\begin{aligned}
\|g(y)\|_2+\|Dg(y)\|_F
&\leq
\|b\|_2+\sum_{j=1}^J\|A_jg_j(y)\|_2
+\sum_{j=1}^J\|A_jDg_j(y)\|_F\\
&\leq
\|b\|_2+\sum_{j=1}^J \|A_j\|_{\mathrm{op}}
\{\|g_j(y)\|_2+\|Dg_j(y)\|_F\}\\
&\leq
\|b\|_2+
\sum_{j=1}^J \|A_j\|_{\mathrm{op}}C_j(1+\|y\|_2)^{a_j}\\
&\leq
\left\{\|b\|_2+\sum_{j=1}^J \|A_j\|_{\mathrm{op}}C_j\right\}
(1+\|y\|_2)^a .
\end{aligned}
\]
The penultimate inequality applies the assumed pointwise polynomial bounds for
the maps $g_j$.  The final inequality uses $a\geq a_j$ for each $j$, so that
$(1+\|y\|_2)^{a_j}\leq(1+\|y\|_2)^a$, and uses $a\geq0$, so that
$\|b\|_2\leq \|b\|_2(1+\|y\|_2)^a$.
It remains to verify the statement that any fixed affine preprocessing map
$y\mapsto Ay+b$ satisfies the condition.
For $y\mapsto Ay+b$, take $J=1$, $g_1(y)=y$, and $A_1=A$ in the finite-sum
part of the lemma.  The required input map $g_1(y)=y$ satisfies
Definition~\ref{def:pointwise_polynomial_regularity} because
$Dg_1(y)=I_n$ and
$\|y\|_2+\|I_n\|_F\leq (1+\sqrt n)(1+\|y\|_2)$.
\end{proof}

\begin{lemma}[Products and compositions]\label{lem:closure_products_compositions}
Let $g_1,g_2:\bbR^n\to\bbR^n$ satisfy
Definition~\ref{def:pointwise_polynomial_regularity}.  Then their componentwise
product $g_1\odot g_2$ satisfies the condition in
Definition~\ref{def:pointwise_polynomial_regularity}.  More generally, if
$g:\bbR^n\to\bbR^m$ satisfies the condition in
Definition~\ref{def:pointwise_polynomial_regularity} and
$\phi:\bbR^m\to\bbR^r$ is continuously differentiable with
\[
\|\phi(z)\|_2+\|D\phi(z)\|_{\mathrm{op}}
\leq C_\phi(1+\|z\|_2)^q ,
\]
then $\phi\{g(y)\}$ satisfies the condition in
Definition~\ref{def:pointwise_polynomial_regularity}.
\end{lemma}

\begin{proof}[Proof of Lemma~\ref{lem:closure_products_compositions}]
We first verify the componentwise product.  The product rule applied coordinate
by coordinate gives
\[
D(g_1\odot g_2)(y)=\diag\{g_2(y)\}Dg_1(y)
+\diag\{g_1(y)\}Dg_2(y).
\]
The Euclidean norm of the product is bounded by the product of the Euclidean
norms:
\[
\|g_1(y)\odot g_2(y)\|_2
\leq \|g_1(y)\|_2\|g_2(y)\|_2 .
\]
For the Jacobian, the triangle inequality and submultiplicativity give
\[
\begin{aligned}
\|D(g_1\odot g_2)(y)\|_F
&\leq
\|\diag\{g_2(y)\}\|_{\mathrm{op}}\|Dg_1(y)\|_F
+\|\diag\{g_1(y)\}\|_{\mathrm{op}}\|Dg_2(y)\|_F\\
&\leq \|g_2(y)\|_2\|Dg_1(y)\|_F
+\|g_1(y)\|_2\|Dg_2(y)\|_F .
\end{aligned}
\]
The last inequality uses
$\|\diag(u)\|_{\mathrm{op}}=\max_i |u_i|\leq\|u\|_2$.  By
Definition~\ref{def:pointwise_polynomial_regularity}, each factor on the right
side is bounded by a polynomial in $1+\|y\|_2$.  A product of two such bounds,
$C_1(1+\|y\|_2)^{a_1}\cdot C_2(1+\|y\|_2)^{a_2}=C_1C_2(1+\|y\|_2)^{a_1+a_2}$, is
a polynomial bound of degree $a_1+a_2$, and a finite sum of polynomial bounds is
bounded by a single polynomial of the maximum degree, so
$g_1\odot g_2$ satisfies
Definition~\ref{def:pointwise_polynomial_regularity}.

We next verify the composition statement.  Let $C_g<\infty$ and
$0\leq a_g<\infty$ be constants such that
\[
\|g(y)\|_2+\|Dg(y)\|_F\leq C_g(1+\|y\|_2)^{a_g}.
\]
Set $q_+:=\max\{q,0\}$.  The assumed bound on $\phi$ and $D\phi$ is also valid
with $q_+$ in place of $q$, since
$(1+\|z\|_2)^q\leq(1+\|z\|_2)^{q_+}$.  The chain rule gives
\[
D\{\phi(g(y))\}=D\phi(g(y))Dg(y).
\]
Taking norms, using submultiplicativity, and applying the assumed bound on
$\phi$ and $D\phi$ gives
\[
\begin{aligned}
&\|\phi(g(y))\|_2+\|D\{\phi(g(y))\}\|_F\\
&\quad\leq
C_\phi(1+\|g(y)\|_2)^{q_+}
+C_\phi(1+\|g(y)\|_2)^{q_+}\|Dg(y)\|_F\\
&\quad=
C_\phi(1+\|g(y)\|_2)^{q_+}\{1+\|Dg(y)\|_F\}.
\end{aligned}
\]
Since $\|g(y)\|_2\leq C_g(1+\|y\|_2)^{a_g}$, the factor
$(1+\|g(y)\|_2)^{q_+}$ is bounded by a constant times $(1+\|y\|_2)^{q_+a_g}$ and
$1+\|Dg(y)\|_F$ by a constant times $(1+\|y\|_2)^{a_g}$, so the right side is
bounded by a polynomial in $1+\|y\|_2$ of degree $(q_++1)a_g$.  Thus
$\phi\{g(y)\}$ satisfies
Definition~\ref{def:pointwise_polynomial_regularity}.
\end{proof}

\begin{lemma}[Standardization maps]\label{lem:nw_standardization_maps}
Let $m,s:\bbR^n\to\bbR^n$ satisfy the pointwise polynomial regularity condition
in Definition~\ref{def:pointwise_polynomial_regularity}.  Suppose the scale map
is bounded away from zero:
\[
s_i(y)\geq s_{\min}>0,\qquad y\in\bbR^n,\quad i=1,\ldots,n.
\]
Define the standardized map
\[
g(y)=\{y-m(y)\}\odot \{1/s(y)\},
\]
where $1/s(y)$ denotes componentwise reciprocals.  Then $g$ satisfies the
pointwise polynomial regularity condition in
Definition~\ref{def:pointwise_polynomial_regularity}.
\end{lemma}

\begin{proof}[Proof of Lemma~\ref{lem:nw_standardization_maps}]
Because $s$ satisfies Definition~\ref{def:pointwise_polynomial_regularity},
there are constants $C_s<\infty$ and $0\leq a_s<\infty$ such that
\[
\|s(y)\|_2+\|Ds(y)\|_F\leq C_s(1+\|y\|_2)^{a_s},\qquad y\in\bbR^n .
\]
The lower bound $s_i(y)\geq s_{\min}$ gives the level bound
\[
\|1/s(y)\|_2\leq \sqrt n\,s_{\min}^{-1}.
\]
Let $Ds_i(y)$ denote the $i$th row of $Ds(y)$.  Differentiating each coordinate
of the reciprocal scale map gives
\[
D\{1/s_i(y)\}=-s_i(y)^{-2}Ds_i(y),\qquad i=1,\ldots,n.
\]
Therefore
\[
\|D\{1/s(y)\}\|_F^2
=\sum_{i=1}^n s_i(y)^{-4}\|Ds_i(y)\|_2^2
\leq s_{\min}^{-4}\|Ds(y)\|_F^2,
\]
so
\[
\|D\{1/s(y)\}\|_F
\leq s_{\min}^{-2}C_s(1+\|y\|_2)^{a_s}.
\]
This is the derivative bound for the reciprocal scale map.  Combining the level
and derivative bounds, and using $a_s\geq0$, gives
\[
\begin{aligned}
\|1/s(y)\|_2+\|D\{1/s(y)\}\|_F
&\leq \sqrt n\,s_{\min}^{-1}
+s_{\min}^{-2}C_s(1+\|y\|_2)^{a_s}\\
&\leq
\{\sqrt n\,s_{\min}^{-1}+s_{\min}^{-2}C_s\}
(1+\|y\|_2)^{a_s}.
\end{aligned}
\]
Thus $1/s(y)$ satisfies
Definition~\ref{def:pointwise_polynomial_regularity}.

The residual map $y-m(y)$ satisfies
Definition~\ref{def:pointwise_polynomial_regularity} by
Lemma~\ref{lem:closure_affine_sum}: the identity map satisfies the condition,
and $m(y)$ satisfies the condition by assumption.  Since
\[
g(y)=\{y-m(y)\}\odot \{1/s(y)\},
\]
the componentwise product closure in
Lemma~\ref{lem:closure_products_compositions} gives
Definition~\ref{def:pointwise_polynomial_regularity} for $g$.
\end{proof}

\subsection{Estimator-Form Verifications}\label{sec:verification}

This subsection verifies pointwise polynomial regularity for estimator-form
building blocks used in the empirical library and records why the
value-similarity rule requires a separate envelope argument.
Remark~\ref{rem:normal_normal_close_regular} gives the closure argument for
normal--normal EB shrinkage and for the \close{}-style construction of
\citet{chenEmpiricalBayesWhen2024}, which applies Nadaraya--Watson centering and
scaling before the same adjustment.  Lemma~\ref{lem:fixed_kernel_gp_regularity}
covers fixed-kernel \gp{} smoothing, where the adjustment is linear in the
outcome vector.  Proposition~\ref{prop:basic_bilateral_not_lipschitz} shows
that a basic value-similarity map need not be globally Lipschitz.  The
fixed-candidate pointwise envelope is verified in
Section~\ref{sec:bilateral_verification}, and
Lemma~\ref{lem:pointwise_regular_to_sobolev} gives the corresponding
$k=0$ Sobolev moment envelope.

\begin{remark}[Normal--normal EB and \close{}-style shrinkage]\label{rem:normal_normal_close_regular}
Let $\sigma_1^2,\ldots,\sigma_n^2$ be fixed
known sampling variances, and let $\mu:\bbR^n\to\bbR$ and
$\tau^2:\bbR^n\to\bbR$ satisfy
Definition~\ref{def:pointwise_polynomial_regularity}, with
$\tau^2(y)\geq\tau_{\min}^2>0$.  The adjustment has coordinates
\[
g_i(y)
=
\frac{\sigma_i^2}{\tau^2(y)+\sigma_i^2}\{\mu(y)-y_i\},
\qquad i=1,\ldots,n .
\]
The denominator $\tau^2(y)+\sigma_i^2$ is regular and bounded away from zero,
so the reciprocal-scale argument in
Lemma~\ref{lem:nw_standardization_maps} applies.  The remaining operations are
fixed multiplication by $\sigma_i^2$, affine formation of $\mu(y){\bf 1}-y$,
and componentwise multiplication, so
Definition~\ref{def:pointwise_polynomial_regularity} follows from
Lemmas~\ref{lem:closure_affine_sum} and
\ref{lem:closure_products_compositions}.  For the \close{}-style construction,
Nadaraya--Watson centering and scaling first standardize the observations;
Lemma~\ref{lem:nw_standardization_maps}, together with closure under products
and compositions, gives pointwise polynomial regularity for the resulting
adjustment.
Fixed affine residualization or deresidualization, when present, is covered by
Lemma~\ref{lem:closure_affine_sum}.
\end{remark}

\begin{lemma}[Fixed-kernel \gp{} smoother regularity]\label{lem:fixed_kernel_gp_regularity}
Let $K\in\bbR^{n\times n}$ and $\Sigma\succ0$ be fixed matrices with
$K+\Sigma$ invertible, and define
\[
S=K(K+\Sigma)^{-1},\qquad f(y)=Sy,\qquad g(y)=f(y)-y.
\]
Then $g$ satisfies Definition~\ref{def:pointwise_polynomial_regularity}.  In
particular,
\[
\|g(y)\|_2+\|Dg(y)\|_F
\leq
\|S-I\|_{\mathrm{op}}\|y\|_2+\|S-I\|_F,\qquad y\in\bbR^n .
\]
\end{lemma}

\begin{proof}[Proof of Lemma~\ref{lem:fixed_kernel_gp_regularity}]
Because $K$ and $\Sigma$ are fixed, the smoothing matrix $S$ is fixed and does
not depend on $y$.  Hence $g$ is continuously differentiable, and for every $y\in\bbR^n$
\[
g(y)=(S-I)y,\qquad Dg(y)=S-I .
\]
From these two identities, $\|g(y)\|_2=\|(S-I)y\|_2\leq\|S-I\|_{\mathrm{op}}\|y\|_2$ and $\|Dg(y)\|_F=\|S-I\|_F$.  Setting $C:=\|S-I\|_{\mathrm{op}}+\|S-I\|_F$, which does not depend on $y$, gives $\|g(y)\|_2+\|Dg(y)\|_F\leq C(1+\|y\|_2)$, so Definition~\ref{def:pointwise_polynomial_regularity}
holds with $a=1$.
\end{proof}

\citet{bellecSecondOrderStein2021} prove \sure{} selection guarantees for
finite collections of globally Lipschitz estimators.  Proposition~\ref{prop:basic_bilateral_not_lipschitz} shows
that even a simple value-similarity smoother fails that global Lipschitz
condition.  The condition used here is instead the Sobolev moment envelope in
Definition~\ref{def:sobolev_moment_envelope}.  Section~\ref{sec:bilateral_verification}
verifies a pointwise polynomial envelope for the fixed product-kernel
value-similarity rule, with envelope scale $\nu_n=O(\sqrt n)$ under a
bounded row-sum condition on its fixed geographic factor.

\begin{proposition}[A two-dimensional value-similarity map is not globally Lipschitz]\label{prop:basic_bilateral_not_lipschitz}
Let $n=2$, $\lambda,\sigma^2>0$, and define
\[
K_{ij}(y):=\exp\{-\lambda(y_i-y_j)^2\}.
\]
Let
\[
f(y)=K(y)\{K(y)+\sigma^2I_2\}^{-1}y,
\qquad
g(y)=f(y)-y.
\]
Then $f$ and $g$ are not globally Lipschitz on $\bbR^2$.
\end{proposition}

\begin{proof}[Proof of Proposition~\ref{prop:basic_bilateral_not_lipschitz}]
For every $y$, the matrix $K(y)$ has $K_{ii}(y)=1$ and
$K_{12}(y)=K_{21}(y)=\exp\{-\lambda(y_1-y_2)^2\}$.  Its eigenvalues are therefore
$1\pm\exp\{-\lambda(y_1-y_2)^2\}$, both nonnegative.  Hence $K(y)\succeq0$, and
since $\sigma^2>0$, $K(y)+\sigma^2I_2$ is invertible.

The adjustment equals
\[
g(y)
=
\left[K(y)\{K(y)+\sigma^2I_2\}^{-1}-I_2\right]y
=
-\sigma^2\{K(y)+\sigma^2I_2\}^{-1}y .
\]
We now compare two input vectors that remain a fixed distance apart.  Choose
fixed $a,b\in\bbR$ with $a\ne b$ and $a^2\ne b^2$.  For
$c\in\{a,b\}$ and $t\geq0$, set
\[
y^{(c)}(t)=(t,t+c),\qquad
r_c=\exp\{-\lambda c^2\},\qquad
d_c=1+\sigma^2+r_c .
\]
The distance between the two inputs is constant:
\[
\|y^{(a)}(t)-y^{(b)}(t)\|_2=|a-b|
\]
for every $t$.  Along the path $y^{(c)}(t)$, the matrix
\[
B_c:=K\{y^{(c)}(t)\}+\sigma^2I_2
=
\begin{pmatrix}
1+\sigma^2&r_c\\
r_c&1+\sigma^2
\end{pmatrix}
\]
does not depend on $t$.  Since ${\bf 1}$ is an eigenvector of $B_c$ with
eigenvalue $d_c$, $B_c^{-1}{\bf 1}=d_c^{-1}{\bf 1}$.  To isolate the part of
the adjustment that grows with $t$, we write
$y^{(c)}(t)=t{\bf 1}+(0,c)^\top$.  Then
\[
B_c^{-1}y^{(c)}(t)
=
\frac{t}{d_c}{\bf 1}+B_c^{-1}(0,c)^\top .
\]
Substituting this decomposition into $g(y)=-\sigma^2\{K(y)+\sigma^2I_2\}^{-1}y$ gives, for
$c\in\{a,b\}$,
\[
\begin{aligned}
g\{y^{(c)}(t)\}
&=
-\sigma^2B_c^{-1}y^{(c)}(t)\\
&=
-\sigma^2
\left\{\frac{t}{d_c}{\bf 1}+B_c^{-1}(0,c)^\top\right\}\\
&=
-\frac{\sigma^2 t}{d_c}{\bf 1}
-\sigma^2B_c^{-1}(0,c)^\top .
\end{aligned}
\]
Taking the difference between the two adjustments gives
\[
\begin{aligned}
g\{y^{(a)}(t)\}-g\{y^{(b)}(t)\}
&=
-\sigma^2t\left(\frac{1}{d_a}-\frac{1}{d_b}\right){\bf 1}\\
&\quad
-\sigma^2\left\{B_a^{-1}(0,a)^\top-B_b^{-1}(0,b)^\top\right\}.
\end{aligned}
\]
Because $\lambda>0$ and $a^2\ne b^2$, $r_a\ne r_b$, so $d_a\ne d_b$.
The matrices $B_a$ and $B_b$ are fixed invertible matrices, and $a,b$ are
fixed scalars.  Hence the part of the difference that does not depend on $t$
has finite norm; write
\[
C:=\sigma^2
\left\|B_a^{-1}(0,a)^\top-B_b^{-1}(0,b)^\top\right\|_2<\infty .
\]
Applying the reverse triangle inequality to the $t$-dependent term and this
bounded remainder gives
\[
\begin{aligned}
\left\|g\{y^{(a)}(t)\}-g\{y^{(b)}(t)\}\right\|_2
&=
\left\|
-\sigma^2t\left(\frac{1}{d_a}-\frac{1}{d_b}\right){\bf 1}
-\sigma^2\left\{B_a^{-1}(0,a)^\top-B_b^{-1}(0,b)^\top\right\}
\right\|_2\\
&\geq
\left\|
-\sigma^2t\left(\frac{1}{d_a}-\frac{1}{d_b}\right){\bf 1}
\right\|_2
-\sigma^2
\left\|B_a^{-1}(0,a)^\top-B_b^{-1}(0,b)^\top\right\|_2\\
&=
\sqrt 2\,\sigma^2
\left|\frac{1}{d_a}-\frac{1}{d_b}\right|t-C .
\end{aligned}
\]
Because $d_a\ne d_b$, the coefficient of $t$ is nonzero, so this lower bound
diverges as $t\to\infty$.  If $g$ were globally Lipschitz with constant
$L<\infty$, then, since $\|y^{(a)}(t)-y^{(b)}(t)\|_2=|a-b|$ for every $t$,
these inputs would have to satisfy
\[
\left\|g\{y^{(a)}(t)\}-g\{y^{(b)}(t)\}\right\|_2
\leq
L\|y^{(a)}(t)-y^{(b)}(t)\|_2
=
L|a-b|
\]
for every $t$.  The right-hand side is constant in $t$, while the preceding
display shows that the left-hand side diverges.  Therefore $g$ is not globally
Lipschitz.  Finally, $g(y)=f(y)-y$, so if $f$ were globally Lipschitz, then
$g$ would also be globally Lipschitz.  Hence $f$ is not globally Lipschitz either.
\end{proof}

Proposition~\ref{prop:basic_bilateral_not_lipschitz} shows that global
Lipschitz regularity can fail even in a two-dimensional value-similarity
example.  Section~\ref{sec:bilateral_verification} verifies instead a
pointwise polynomial envelope for the general fixed-parameter
value-similarity rule.

\subsection{Fixed-Parameter Value-Similarity Regularity}\label{sec:bilateral_verification}

This subsection verifies the per-candidate regularity condition used for
averaging (Assumption~\ref{asm:averaging_regularity}) for one fixed
value-similarity map: the shrinkage map of
Example~\ref{ex:bilateral_shrinkage}, with the geographic parameter
$\gamma$ and the value-similarity parameter $\lambda$ held fixed.  The
result is a single-candidate verification for the averaging guarantee of
Proposition~\ref{prop:averaging} rather than a uniform regularity
statement over a trained class as in Theorem~\ref{thm:uniform_conc}.
Under a bounded row-sum condition on the fixed geographic factor of the
kernel, the adjustment satisfies the pointwise envelope
$\nu_n(1+\|y\|_2/\sqrt n)$ with envelope scale $\nu_n=C\sqrt n$, for a
constant $C$ that does not depend on $n$
(Lemma~\ref{lem:bilateral_product_envelope}).  Through the sufficient condition
stated after Assumption~\ref{asm:averaging_regularity}, this envelope
gives the moment condition of that assumption with envelope scale
$O(\sqrt n)$, a scale at which the regret bound of
Proposition~\ref{prop:averaging} vanishes for fixed $K$.

Fix $n$ and $\lambda\geq0$, and assume that
\[
0<\underline\sigma^2 I\preceq \Sigma\preceq \bar\sigma^2 I,
\]
with constants independent of $n$.  Write
$\varphi(r):=e^{-\lambda r^2}$.  Throughout this subsection, constants may
depend on the fixed values of $\lambda$, $\underline\sigma^2$,
$\bar\sigma^2$, and the geographic row-sum bound $C_G$ of
Assumption~\ref{asm:geo_row_sums}, but not on $y$.  Dependence on $n$ is displayed
explicitly.  Let $G\in\bbR^{n\times n}$ be a fixed matrix, not depending
on $y$, playing the role of the geographic covariance matrix
$K^{\mathrm{geo}}_\gamma$ of Example~\ref{ex:bilateral_shrinkage} at the
fixed parameter value.  For $y\in\bbR^n$, define the \emph{product
kernel} $K(y)$, the entrywise product of $G$ with the value-similarity
factor of the example,
\[
K_{ij}(y):=G_{ij}\,\varphi(y_i-y_j),
\qquad
A(y):=K(y)+\Sigma,
\qquad
w(y):=A(y)^{-1}y .
\]
The shrinkage map and its adjustment are
\[
f(y):=K(y)A(y)^{-1}y,
\qquad
g(y):=f(y)-y .
\]
The matrix whose $(i,j)$th entry is $\varphi(y_i-y_j)=e^{-\lambda(y_i-y_j)^2}$ is positive
semidefinite for every $y\in\bbR^n$, since the Gaussian kernel
$(a,b)\mapsto e^{-\lambda(a-b)^2}$ is a positive-definite function on
$\bbR$ (with $\lambda=0$ giving the all-ones matrix).

\begin{assumption}[Bounded row sums of the geographic covariance matrix]\label{asm:geo_row_sums}
The matrix $G$ is fixed, symmetric, and positive semidefinite, with
nonnegative entries and, for a constant $C_G<\infty$ not depending on
$n$,
\[
\max_{1\leq i\leq n}\sum_{j=1}^nG_{ij}\leq C_G .
\]
\end{assumption}

The verification proceeds in two steps.  The first lemma computes the
Jacobian of the adjustment in closed form, and the second bounds that
Jacobian entry by entry to obtain the envelope.

\begin{lemma}[Derivative of the product-kernel adjustment]
\label{lem:bilateral_product_jacobian}
Under the standing setup of this subsection and Assumption~\ref{asm:geo_row_sums}, the matrix $A(y)$ is
invertible for every $y\in\bbR^n$, with
\[
\|A(y)^{-1}\|_{\mathrm{op}}\leq\frac{1}{\underline\sigma^2},
\qquad
g(y)=-\Sigma A(y)^{-1}y .
\]
Moreover, $g$ is continuously differentiable on $\bbR^n$, with
\[
Dg(y)
=
-\Sigma A(y)^{-1}
+\Sigma A(y)^{-1}
\bigl[\diag\{Q(y)w(y)\}
-Q(y)\,\diag\{w(y)\}\bigr],
\]
where $Q_{ij}(y):=G_{ij}\,\varphi'(y_i-y_j)$.
\end{lemma}

\begin{proof}[Proof of Lemma~\ref{lem:bilateral_product_jacobian}]
The matrix whose $(i,j)$th entry is $\varphi(y_i-y_j)$ is positive
semidefinite for every $y\in\bbR^n$, as recorded above, and by the Schur
product theorem the entrywise product of two positive semidefinite
matrices is positive semidefinite, so $K(y)\succeq0$ for every
$y\in\bbR^n$.  Therefore
$A(y)\succeq\Sigma\succeq\underline\sigma^2I$, which gives invertibility
of $A(y)$ and $\|A(y)^{-1}\|_{\mathrm{op}}\leq1/\underline\sigma^2$.
Since $K(y)-A(y)=-\Sigma$,
\[
g(y)
=\{K(y)-A(y)\}A(y)^{-1}y
=-\Sigma A(y)^{-1}y .
\]

The map $y\mapsto g(y)$ is continuously differentiable because the
entries of $K(y)$ are continuously differentiable in $y$ and $A(y)$ is
invertible for every $y$.  Because $G$ is fixed, only the factor
$\varphi(y_i-y_j)$ in $K_{ij}(y)$ varies with $y$, so by the scalar chain
rule the directional derivative of $K$ at $y$ in the direction
$v\in\bbR^n$ has entries
\[
\bigl\{DK(y)[v]\bigr\}_{ij}
=G_{ij}\,\varphi'(y_i-y_j)\,(v_i-v_j)
=Q_{ij}(y)(v_i-v_j) .
\]

We now differentiate $g(y)=-\Sigma A(y)^{-1}y$ in the direction $v$.
The product rule gives one term from the inverse and one term from the
final factor $y$:
\[
Dg(y)[v]
=-\Sigma\bigl(D\{A(y)^{-1}\}[v]\bigr)y
-\Sigma A(y)^{-1}v .
\]
The first term requires the directional derivative of the matrix
inverse.  Differentiating both sides of the identity $A(y)^{-1}A(y)=I$
in the direction $v$ --- the right-hand side does not depend on $y$, so
its derivative is zero --- the product rule gives
\[
\bigl(D\{A(y)^{-1}\}[v]\bigr)A(y)+A(y)^{-1}\bigl(DA(y)[v]\bigr)=0 .
\]
Multiplying on the right by $A(y)^{-1}$ and rearranging,
\[
D\{A(y)^{-1}\}[v]
=-A(y)^{-1}\bigl(DA(y)[v]\bigr)A(y)^{-1},
\]
and $DA(y)[v]=DK(y)[v]$ because $\Sigma$ does not depend on $y$.
Substituting this into the expression for $Dg(y)[v]$ and using
$A(y)^{-1}y=w(y)$,
\[
Dg(y)[v]
=\Sigma A(y)^{-1}\bigl\{DK(y)[v]\bigr\}w(y)
-\Sigma A(y)^{-1}v .
\]
The second term matches the leading term $-\Sigma A(y)^{-1}$ of the
formula in the lemma statement.  The first term is where the kernel's
dependence on the data enters.  To identify its contribution to
$Dg(y)$, we write the vector $\{DK(y)[v]\}w(y)$ as a matrix acting on
$v$.  Its $i$th entry splits into a $v_i$ part and a $v_j$ part:
\begin{align*}
\bigl[\{DK(y)[v]\}w(y)\bigr]_i
&=\sum_jQ_{ij}(v_i-v_j)w_j\\
&=v_i\sum_jQ_{ij}w_j-\sum_jQ_{ij}w_jv_j\\
&=\bigl[\diag\{Q(y)w(y)\}\,v\bigr]_i-\bigl[Q(y)\,\diag\{w(y)\}\,v\bigr]_i .
\end{align*}
Hence
$\{DK(y)[v]\}w(y)=\bigl[\diag\{Q(y)w(y)\}-Q(y)\,\diag\{w(y)\}\bigr]v$,
and substituting this into the first term of the expression for
$Dg(y)[v]$,
\[
Dg(y)[v]
=\Bigl(-\Sigma A(y)^{-1}
+\Sigma A(y)^{-1}\bigl[\diag\{Q(y)w(y)\}
-Q(y)\,\diag\{w(y)\}\bigr]\Bigr)v .
\]
Since this holds for every direction $v$, the matrix in parentheses is
$Dg(y)$, as stated in the lemma.
\end{proof}

\begin{lemma}[Product-kernel envelope]\label{lem:bilateral_product_envelope}
Under the standing setup of this subsection and Assumption~\ref{asm:geo_row_sums}, the singleton class
$\mathcal F=\{f\}$ satisfies Assumption~\ref{asm:regularity} with
$\beta=1/2$ and $\nu_n=C\sqrt n$ for a constant $C<\infty$ not depending
on $n$; that is,
\[
\|g(y)\|_2+\|Dg(y)\|_F
\leq
\nu_n\left(1+\frac{\|y\|_2}{\sqrt n}\right)
\qquad\text{for all }y\in\bbR^n .
\]
\end{lemma}

\begin{proof}[Proof of Lemma~\ref{lem:bilateral_product_envelope}]
Lemma~\ref{lem:bilateral_product_jacobian} gives the continuous differentiability of $g$.
It remains to verify the displayed envelope.  Write
$\rho_\sigma:=\bar\sigma^2/\underline\sigma^2$.  We first bound
the Euclidean norm of the shrinkage adjustment, the $\|g(y)\|_2$ term in
the displayed envelope.  By Lemma~\ref{lem:bilateral_product_jacobian},
$g(y)=-\Sigma A(y)^{-1}y$ with
$\|A(y)^{-1}\|_{\mathrm{op}}\leq1/\underline\sigma^2$, and
$\|\Sigma\|_{\mathrm{op}}\leq\bar\sigma^2$, so
\[
\|g(y)\|_2\leq\rho_\sigma\|y\|_2,
\qquad
\|w(y)\|_2\leq\frac{\|y\|_2}{\underline\sigma^2} .
\]

Next we bound the Frobenius norm of the derivative.  Write $A:=A(y)$,
$w:=w(y)$, $Q:=Q(y)$.  Applying the triangle inequality to the Jacobian
formula of Lemma~\ref{lem:bilateral_product_jacobian}, and then the bound
$\|MX\|_F\leq\|M\|_{\mathrm{op}}\|X\|_F$ with $M=\Sigma A^{-1}$ and
$\|\Sigma A^{-1}\|_{\mathrm{op}}\leq\rho_\sigma$,
\begin{align*}
\|Dg(y)\|_F
&\leq
\|\Sigma A^{-1}\|_F
+\|\Sigma A^{-1}\diag(Qw)\|_F
+\|\Sigma A^{-1}Q\,\diag(w)\|_F\\
&\leq
\|\Sigma A^{-1}\|_F
+\rho_\sigma\|\diag(Qw)\|_F
+\rho_\sigma\|Q\,\diag(w)\|_F .
\end{align*}
For the middle term, $\|\diag(Qw)\|_F=\|Qw\|_2$.  The Frobenius
norm of an $n\times n$ matrix is at most $\sqrt n$ times its operator
norm, so the first term satisfies
\[
\|\Sigma A^{-1}\|_F
\leq\sqrt n\,\|\Sigma A^{-1}\|_{\mathrm{op}}
\leq\rho_\sigma\sqrt n .
\]
It remains to bound $\|Qw\|_2$ and $\|Q\,\diag(w)\|_F$.

Recall $Q_{ij}(y)=G_{ij}\,\varphi'(y_i-y_j)$, with
$\varphi(r)=e^{-\lambda r^2}$ from the setup, so
$\varphi'(r)=-2\lambda re^{-\lambda r^2}$.  For $\lambda>0$, the function $|\varphi'|$
attains its supremum at $|r|=1/\sqrt{2\lambda}$; together with the trivial case
$\varphi'\equiv0$ at $\lambda=0$, this gives
$L_\lambda:=\sup_r|\varphi'(r)|=\sqrt{2\lambda/e}$ for all $\lambda\geq0$.  Since
the entries of $G$ are nonnegative, $|Q_{ij}|\leq L_\lambda G_{ij}$.
Because $G$ is symmetric, its column sums equal its row sums, so every
row sum and every column sum of $G$ is at most $C_G$, and the same holds
for the entrywise absolute values of $Q$ up to the factor $L_\lambda$:
\[
\max_i\sum_j|Q_{ij}|\leq L_\lambda C_G,
\qquad
\max_j\sum_i|Q_{ij}|\leq L_\lambda C_G .
\]

For $\|Qw\|_2$ we bound the operator norm of $Q$, which is the supremum
of $|u^\top Qv|$ over unit vectors $u,v\in\bbR^n$.  For unit vectors $u$
and $v$, the Cauchy--Schwarz inequality gives
\begin{align*}
|u^\top Qv|
&\leq
\sum_{i,j}\bigl(|Q_{ij}|^{1/2}|u_i|\bigr)\bigl(|Q_{ij}|^{1/2}|v_j|\bigr)\\
&\leq
\Bigl(\sum_{i,j}|Q_{ij}|\,u_i^2\Bigr)^{1/2}
\Bigl(\sum_{i,j}|Q_{ij}|\,v_j^2\Bigr)^{1/2}\\
&\leq
\Bigl(\max_{i}\sum_j|Q_{ij}|\Bigr)^{1/2}
\Bigl(\max_{j}\sum_i|Q_{ij}|\Bigr)^{1/2}\\
&\leq
L_\lambda C_G .
\end{align*}
Taking the supremum over unit vectors $u$ and $v$ gives
$\|Q\|_{\mathrm{op}}\leq L_\lambda C_G$.  Combining this operator-norm
bound with $\|w\|_2\leq\|y\|_2/\underline\sigma^2$ from the first
display,
\[
\|Qw\|_2
\leq\|Q\|_{\mathrm{op}}\|w\|_2
\leq\frac{L_\lambda C_G}{\underline\sigma^2}\,\|y\|_2 .
\]

For $\|Q\,\diag(w)\|_F$, multiplying $Q$ by $\diag(w)$ on the right
rescales its columns: the $j$th column of $Q\,\diag(w)$ is
$w_jQ_{\cdot j}$, so
\[
\|Q\,\diag(w)\|_F^2
=\sum_{j=1}^nw_j^2\sum_{i=1}^nQ_{ij}^2 .
\]
Each column of $Q$ satisfies
\[
\sum_iQ_{ij}^2
\leq L_\lambda^2\sum_iG_{ij}^2
\leq L_\lambda^2\Bigl(\max_iG_{ij}\Bigr)\sum_iG_{ij}
\leq L_\lambda^2C_G^2 ,
\]
using $|Q_{ij}|\leq L_\lambda G_{ij}$ for the first inequality,
$G_{ij}^2\leq(\max_iG_{ij})G_{ij}$ for the second, and
$\max_iG_{ij}\leq\sum_iG_{ij}\leq C_G$ for the third.  Combining the two
displays with $\|w\|_2\leq\|y\|_2/\underline\sigma^2$,
\[
\|Q\,\diag(w)\|_F^2
\leq L_\lambda^2C_G^2\|w\|_2^2
\leq\frac{L_\lambda^2C_G^2}{\underline\sigma^4}\,\|y\|_2^2 .
\]

Collecting the bound on $\|g(y)\|_2$ and the three derivative bounds,
\[
\|g(y)\|_2+\|Dg(y)\|_F
\leq
\rho_\sigma\sqrt n
+\rho_\sigma\left(1+\frac{2L_\lambda C_G}{\underline\sigma^2}\right)
\|y\|_2 .
\]
Set $C:=\rho_\sigma(1+2L_\lambda C_G/\underline\sigma^2)$, which does
not depend on $n$ or $y$.  The first term is at most $C\sqrt n$ because
$\rho_\sigma\leq C$, and the second term equals
$C\sqrt n\,(\|y\|_2/\sqrt n)$, so
\[
\|g(y)\|_2+\|Dg(y)\|_F
\leq
C\sqrt n\left(1+\frac{\|y\|_2}{\sqrt n}\right),
\]
the pointwise envelope with $\beta=1/2$ and $\nu_n=C\sqrt n$.
\end{proof}

\section{Implementation and Computation}\label{app:implementation_computation}

\subsection{Opportunity Atlas Implementation Details}\label{app:implementation}

The main comparison across CZs for the pooled outcome uses the seven-candidate library summarized in Table~\ref{tab:oa_candidates}, which compares non-spatial baselines to geographic and contiguity \gp{} shrinkage candidates.  The additional-outcome and targeting tables use the four-candidate geographic-distance \gp{} comparison stated in Appendix~\ref{app:oa_robustness}.  The value-similarity rule from Example~\ref{ex:bilateral_shrinkage} appears in the Cook County comparison in Section~\ref{sec:oa_value_similarity_ladder} and Figure~\ref{fig:oa_geo_bilateral_ladder}, not as part of the main average across CZs.  All reported \sure{} values use the diagonal sampling covariance formed from the Opportunity Atlas marginal standard errors.  The paragraphs below record how these \sure{} values are computed for the implemented candidates, including randomized trace estimation and the training correction.

The reported Opportunity Atlas computations were run on a single workstation GPU (an NVIDIA RTX 4090).

The trainable \gp{} candidates are trained by minimizing the proxy \sure{} criterion via AdamW optimization \citep{kingmaAdamMethodStochastic2017,loshchilovDecoupledWeightDecay2019} with cosine learning-rate schedule, with the weight-decay coefficient set to zero for the \gp{} parameters, and gradient clipping at norm $1.0$.  The \gp{} candidates in Table~\ref{tab:oa_candidates} use learning rate $0.02$ for 100 epochs; the value-similarity \gpbilat{} candidates used in the Cook County comparison use learning rate $0.10$ for 100 epochs.

For the \gpbilat{} candidates, the value-similarity component is computed after standardizing by the reported standard error and applying an arcsinh transformation: $z_i=\operatorname{arcsinh}(Y_i/\sigma_i)$.  The implementation then uses squared differences $(z_i-z_j)^2$ in this component.  This transformation compresses extreme standardized estimates and led to more stable training in the empirical runs.

Gradient clipping is used only as a numerical safeguard.  Appendix~\ref{app:learned_parameter_regularity} covers the corresponding differentiable AdamW update map with clipping omitted.  The reported \sure{} calculations apply automatic differentiation to the implemented training path used in the empirical runs.

To compute \sure{} in a way that accounts for trained-parameter dependence on $Y$, we propagate derivative information (tangent vectors) through the AdamW optimizer state during training---that is, we differentiate through the training iterations themselves---following the approach described in Section~\ref{sec:learned_params}.  The randomized trace terms use Hutchinson probe vectors---the random vectors $v$ in the trace estimate of Section~\ref{sec:learned_params}---scaled coordinatewise by the reported standard errors: $v_i=\sigma_i z_i$ with $z_i\in\{-1,1\}$, so that $\E[vv^\top]=\Sigma$ \citep{hutchinson_stochastic_1990}.  During optimization, the proxy \sure{} criterion uses 5 probes.  After training, 10 Hutchinson probes are used to evaluate \sure{} for each final estimator.  All probe draws and optimization seeds are treated as external algorithmic randomness.  Conditional on them, each reported estimator is a deterministic map of $Y$, and automatic differentiation differentiates that implemented map.  With finitely many probes, the reported divergence is a Hutchinson estimate of the full divergence trace: it is unbiased over the probe randomness, but noisier than full trace computation.  Thus the reported trainable-candidate values should be read as Hutchinson-estimated \sure{} values for the implemented trained maps.

Geographic distance matrices record Euclidean coordinate distance between tract centroids in longitude--latitude coordinates, and contiguity distance matrices record shortest-path distance on the queen-contiguity tract graph (tracts sharing a boundary or vertex are adjacent).  Both distance types are median-normalized before kernel evaluation, so the kernels operate on normalized distances rather than squared distances, consistent with the kernel definitions in Section~\ref{sec:estimators}.

The \nneb{} candidate is the global normal--normal EB rule.  In the implementation, this closed-form rule is evaluated after a global centering and scaling normalization and then mapped back to the original scale.  Apart from the scale floors, this normalization is algebraically equivalent to applying the same rule directly to $Y$, so it is not listed as preprocessing in Table~\ref{tab:oa_candidates}.  Local Nadaraya--Watson preprocessing is used for \closegauss{} and the \gp{} candidates, following the precision-based standardization in \citet{chenEmpiricalBayesWhen2024}.  Nadaraya--Watson weights based on log reported variance define a local conditional mean $\hat\mu_i$ and standard deviation $\hat s_i$, the outcome is standardized to $\tilde Y_i=(Y_i-\hat\mu_i)/\hat s_i$ before prediction, and the prediction is transformed back afterward.  The log reported variance, the conditioning variable of the Nadaraya--Watson regression, is standardized to unit standard deviation, so the bandwidths are deterministic functions of the tract count: the local-mean bandwidth uses Silverman's rule of thumb, $1.06\,n^{-1/5}$, and the local-scale bandwidth uses the oversmoothed rule $2.0\,n^{-1/6}$.  Neither bandwidth is trained by \sure{} or otherwise depends on $Y$.  Candidates labeled ``OLS'' in Table~\ref{tab:oa_candidates} additionally residualize $Y$ on \texttt{pct\_white}, \texttt{pct\_black}, \texttt{pct\_hispanic}, and \texttt{median\_age} before applying local Nadaraya--Watson preprocessing and the \gp{} smoother.  The fitted demographic component is added back afterward.  Automatic differentiation is applied to the implemented prediction wrapper, including these preprocessing and inverse-preprocessing steps conditional on the precomputed Nadaraya--Watson weights and scale floors.

The weights of the \sure{}-weighted average are computed by solving the simplex-constrained quadratic program
\[
\min_{w \in \Delta^{K-1}} \sure_n(f_w),
\]
using the trained candidate predictions and the per-candidate divergence estimates, which together determine the \sure{} objective for each convex combination.  The oracle comparison treats the selected weights as fixed after selection.  The reported \sure{} value for the \sure{}-weighted average instead differentiates the full map $Y\mapsto f_{\hat w(Y)}(Y)$.  In the appendix comparisons, we also report a fixed-weight proxy that holds $\hat w(Y)$ fixed in the \sure{} calculation.  This proxy isolates the fixed-weight averaging object used in the model-averaging oracle comparison; Appendix~\ref{app:averaging_adaptive_weights} states the condition used to interpret the reported evaluation of the \sure{}-weighted average.

\section{Supplementary Empirical Analyses and Additional Outcomes}\label{app:empirical_diagnostics}

\subsection{Fixed-Parameter Proxy Comparisons for the Main Candidate Library}\label{app:proxy_exact_case}

The proxy comparison is supplementary to the main empirical claim, but it is useful for understanding trained-parameter optimism.  Table~\ref{tab:oa_proxy_exact} reports the tract-weighted pooled comparison for the main pooled-outcome candidate library.  Figure~\ref{fig:oa_proxy_exact_case} displays the correction for the same candidate library and pooled outcome as Table~\ref{tab:oa_proxy_exact}, with the \sure{}-minus-proxy gap on the vertical axis.  For \nneb{} and \closegauss{} the fixed-parameter proxy holds their moment-based estimates fixed---the global mean and prior variance for \nneb{}, the local mean and variance for \closegauss{}---so the gap comes entirely from the data-dependence of those estimates.

\input{tables/auto_oa_proxy_exact_table}

\begin{figure}[H]
  \centering
  \includegraphics[width=0.85\textwidth]{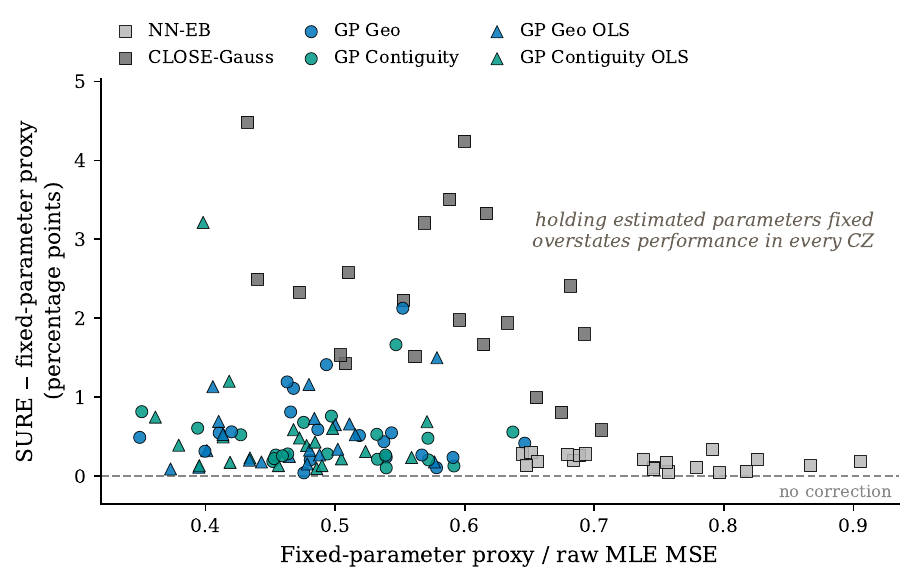}
  \caption{Trained-parameter optimism across commuting zones.  Each point is one candidate in one CZ (\oaProxyGapPoints{} points: the non-\mle{} candidates of the main library in all \oaProxyGapNCzs{} CZs, pooled outcome).  The vertical axis is \sure{} minus the fixed-parameter proxy in percentage points; the horizontal axis is the proxy value divided by the raw-\mle{} MSE benchmark.  For \nneb{} and \closegauss{} the proxy holds their moment-based mean and variance estimates fixed.  The dashed line marks a zero correction.  The gap is positive for every candidate in every CZ.  The size of the correction tracks how much each tract's own observation feeds back into the estimated structure of its fit: \nneb{} estimates two global moments, so its corrections stay small; \closegauss{} estimates the mean and variance locally, giving each tract more influence over its own shrinkage and the largest corrections in the library.  Marker shapes and colors follow Figure~\ref{fig:oa_heterogeneity}.}
  \label{fig:oa_proxy_exact_case}
\end{figure}

\subsection{Additional Opportunity Atlas KFR Outcomes}\label{app:oa_robustness}

Table~\ref{tab:oa_robustness} reports a four-candidate comparison for the pooled outcome and three additional KFR outcomes.  Each column uses the \mle{}, \nneb{}, \closegauss{}, and the geographic-distance \gp{} candidate, which uses local Nadaraya--Watson preprocessing and no OLS residualization.  The table asks whether spatial shrinkage's advantage over the non-spatial EB baselines on the pooled outcome also appears in related outcomes, while keeping the larger seven-candidate library reserved for the main pooled-outcome result.

\input{tables/auto_oa_robustness_table}

The targeting exercise in Table~\ref{tab:oa_targeting_welfare} uses the same four-candidate comparison with the geographic-distance \gp{}.  For each coupled-bootstrap draw, estimators are trained on one perturbed sample, tracts are ranked by the resulting predictions, and the top third are evaluated on the paired perturbed sample.  The table reports tract-weighted averages across CZs for the pooled outcome and the same subgroup outcomes shown above. We use $B=20$ coupled-bootstrap draws per CZ. Each draw retrains all four candidates in every CZ, and each reported cell averages over all \oaNCzs{} CZs, so it pools $400$ ($=20\times\oaNCzs{}$) trained coupled-bootstrap evaluations.

\subsection{Coupled Bootstrap Comparison}\label{app:cb}

As a supplementary check, we compare \sure{} to the coupled bootstrap (CB) of \citet{oliveiraUnbiasedRiskEstimation2024}.  For each bootstrap replicate $b = 1, \ldots, B$:
\begin{enumerate}[nosep]
  \item Draw $\xi^{(b)} \sim \N(0, \Sigma)$ and form the paired samples
  \[
  Y^{*(b)} = Y + \sqrt{\alpha}\,\xi^{(b)},
  \qquad
  Y^{\dagger(b)} = Y - \alpha^{-1/2}\xi^{(b)}.
  \]
  \item Train the estimator on $Y^{*(b)}$ to obtain predictions $\hat\theta^{*(b)}$.
  \item Compute the normalized debiased MSE:
  $n^{-1}\{\|\hat\theta^{*(b)} - Y^{\dagger(b)}\|^2 - \alpha^{-1}\|\xi^{(b)}\|^2 - \tr(\Sigma)\}$.
\end{enumerate}
The subtracted correction has expectation $(1+\alpha^{-1})\tr(\Sigma)$, which equals $\E\|Y^{\dagger}-\theta\|^2$, the evaluation sample's noise contribution.  The training and evaluation noises in $Y^{*(b)}$ and $Y^{\dagger(b)}$ are uncorrelated, and hence independent in the Gaussian case.  The average over $B$ replicates estimates the risk of the estimator applied to the variance-inflated training sample $Y^*$; as $\alpha\downarrow0$ this targets the original risk, while finite $\alpha$ leaves the variance inflation described below.

For a fixed linear smoother $f(Y) = SY$, write $R_n(f)=n^{-1}\E\|SY-\theta\|^2 = n^{-1}\|(S-I)\theta\|^2 + V_n(f)$ for its expected squared-error risk and $V_n(f)=n^{-1}\tr(S\Sigma S^\top)$ for the variance component of that risk.  In the homoskedastic case, this reduces to $V_n(f)=\sigma^2 \tr(S^\top S)/n$.  The expected CB risk satisfies $\E[\mathrm{CB}(f)] = R_n(f) + \alpha\,V_n(f)$.  Normalizing by the homoskedastic $\E[\mathrm{CB}(\mathrm{MLE})]=(1+\alpha)\sigma^2$ gives the ratio calculation below.  The same variance-inflation logic applies in the heteroskedastic case with $\E[\mathrm{CB}(\mathrm{MLE})]=(1+\alpha)n^{-1}\tr(\Sigma)$.  When $V_n(f)/R_n(f)<1$, the CB \emph{ratio} underestimates the \sure{} ratio by a factor $(1+\alpha)/(1+\alpha\,V_n(f)/R_n(f))$.  Figure~\ref{fig:sure_vs_cb} is consistent with this pattern: all points lie above the 45-degree line, with the gap increasing for candidates whose risk has a smaller variance share $V_n(f)/R_n(f)$ (heavier smoothing).

\begin{figure}[H]
  \centering
  \includegraphics[width=0.65\textwidth]{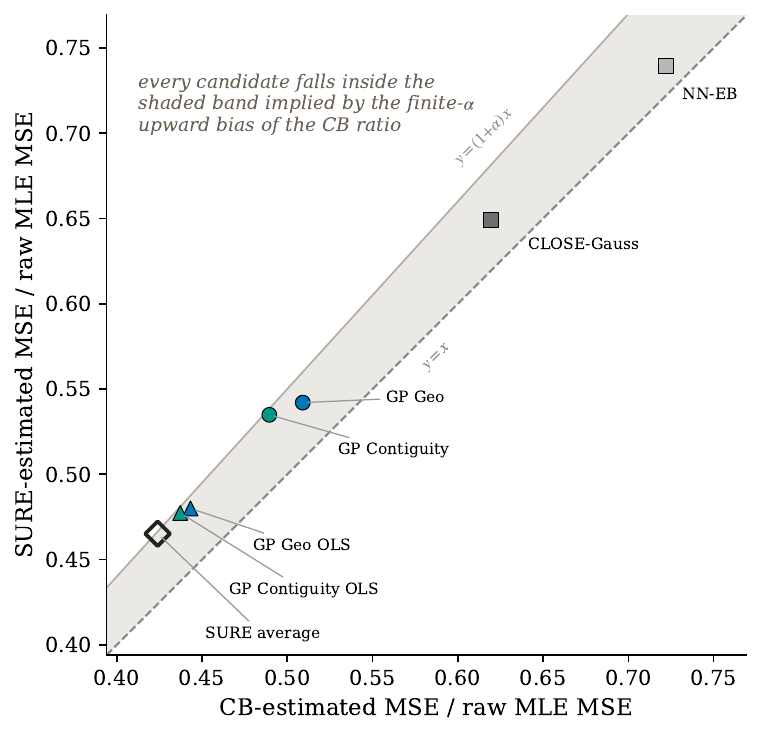}
  \caption{\sure{}-estimated MSE ratio versus coupled-bootstrap MSE estimate for Pittsburgh, the smallest CZ in the sample ($B = 100$, $\alpha = 0.1$).  Each coupled-bootstrap replicate retrains every candidate, so the comparison uses the CZ where the $B=100$ retrainings are least costly.  The plotted candidates are the non-\mle{} members of the main library of Table~\ref{tab:oa_candidates} (\mle{} is the normalizing benchmark), and the open diamond is the \sure{}-weighted average.  The vertical coordinate divides \sure{}-estimated MSE by the raw-\mle{} MSE benchmark.  The horizontal coordinate divides the coupled-bootstrap MSE estimate by its raw-\mle{} counterpart.  Points lie close to, and systematically above, the 45-degree line, showing that this derivative-free check closely tracks \sure{} in this example.  The upward shift reflects the finite-$\alpha$ bias of the CB ratio.  The CB denominator is inflated by the full factor $(1+\alpha)$ while each numerator is inflated only by $(1+\alpha V_n/R_n)$, so the CB ratio underestimates the \sure{} ratio by $(1+\alpha)/(1+\alpha V_n/R_n)$.  The shaded band marks the implied region between $y=x$ and $y=(1+\alpha)x$.  The band is exact in expectation for fixed linear smoothers.  For the trained candidates and the \sure{}-weighted average, whose training is rerun on each variance-inflated sample, it should be read as a reference band.  Individual points carry Monte Carlo error from the $B=100$ replicates.  Marker shapes and colors follow Figure~\ref{fig:oa_heterogeneity}.}
  \label{fig:sure_vs_cb}
\end{figure}

\subsection{Supplementary ASSURE Targeting Comparison}\label{app:assure}

\citet{chenCompoundSelectionDecisions2025} propose ASSURE, which directly optimizes a welfare criterion for compound selection decisions.  In a threshold targeting problem with cost $K$, a rule selects unit $i$ when its decision index exceeds $K$, and the latent welfare is proportional to $(\theta_i-K)$ for selected units.\footnote{In this appendix, $K$ denotes the selection cost threshold and $\lambda$ indexes the smoothed threshold rule. Neither is related to the candidate count $K$ or the value-similarity parameter $\lambda$ used elsewhere in the paper.}  This is a different target from \sure{}: ASSURE targets one decision problem, while \sure{} estimates squared-error risk for reusable estimation maps.  For that reason, the comparison here is appendix-only.

The ASSURE criterion smooths the discontinuous threshold rule.  Write $\sigma_i^2$ for the sampling variance of $Y_i$ and let
\[
u_i(\delta)=\frac{Y_i-\delta_i}{\sigma_i h},
\qquad
h=\{2\log(n)\}^{-1/2}.
\]
With
\[
\operatorname{sinc}(u)=\frac{\sin u}{\pi u},
\qquad
\operatorname{Csinc}(u)=\frac{1}{2}+\frac{\operatorname{Si}(u)}{\pi},
\]
where $\operatorname{Si}(u)=\int_0^u \sin(t)/t\,\mathrm{d}t$ is the sine integral, the implemented ASSURE objective for a threshold vector $\delta=(\delta_1,\ldots,\delta_n)$ is
\begin{equation}
\widehat W_h(\delta;Y)
=
\frac{1}{n}\sum_{i=1}^n
\left[
(Y_i-K)\operatorname{Csinc}\{u_i(\delta)\}
-\frac{\sigma_i}{h}\operatorname{sinc}\{u_i(\delta)\}
\right].
\label{eq:assure_welfare_objective}
\end{equation}
The raw-\mle{} rule corresponds to $\delta_i=K$, so it selects units with $Y_i>K$.

The figure compares three non-\mle{} rules.  Linear ASSURE uses the one-parameter threshold
\[
\delta_i(\beta)=K+\beta\sigma_i^2,
\qquad
\widehat\beta\in\argmax_{\beta\in\mathcal B}
\widehat W_h\{\delta(\beta);Y\},
\]
where $\mathcal B$ is the grid used in this comparison.  Spatial ASSURE is a grid search over geographic \gp{} parameters. It starts from a geographic \gp{} prediction map $\widehat\theta_\lambda(Y)$ indexed by kernel parameters $\lambda$.  For each $\lambda$ in a grid centered at the \sure{}-trained geographic \gp{} parameters, the rule selects units with $\widehat\theta_{\lambda,i}(Y)>K$.  To evaluate this selection rule with \eqref{eq:assure_welfare_objective}, the implementation uses the local threshold representation
\[
\delta_i(\lambda;Y)
=
Y_i-\frac{\widehat\theta_{\lambda,i}(Y)-K}{H_i},
\]
where $H_i>0$ is the estimated diagonal derivative $\partial \widehat\theta_i(Y)/\partial Y_i$ at the \sure{}-trained geographic \gp{} map, held fixed across the grid.  This is a local linearized threshold representation, not an exact global inversion of the prediction map.  Spatial ASSURE chooses the grid point maximizing $\widehat W_h\{\delta(\lambda;Y);Y\}$.  The \sure{} plug-in rule does not optimize \eqref{eq:assure_welfare_objective}.  It trains the geographic \gp{} by \sure{} and then selects units with $\widehat\theta_{\widehat\lambda_{\sure},i}(Y)>K$.

The plotted values are coupled-bootstrap estimates of the realized threshold welfare delivered by the trained rules, not the smoothed objective in \eqref{eq:assure_welfare_objective}.  For bootstrap draw $b$, draw $\xi^{(b)}\sim \N(0,\diag(\sigma_1^2,\ldots,\sigma_n^2))$ independently of $Y$ and form
\[
Y^{*(b)}=Y+\sqrt{\alpha}\,\xi^{(b)},
\qquad
Y^{\dagger(b)}=Y-\alpha^{-1/2}\xi^{(b)}.
\]
Each rule is trained on $Y^{*(b)}$, producing a selection vector
$S_m^{(b)}\in\{0,1\}^n$.  The coupled-bootstrap welfare estimate for method $m$ is
\begin{equation}
\widehat W_{\mathrm{CB}}(m;K)
=
\frac{1}{B}\sum_{b=1}^B
\frac{1}{n}\sum_{i=1}^n
\bigl(Y_i^{\dagger(b)}-K\bigr)S_{m,i}^{(b)}.
\label{eq:assure_cb_welfare}
\end{equation}
Figure~\ref{fig:assure_welfare} reports
$\widehat W_{\mathrm{CB}}(m;K)-\widehat W_{\mathrm{CB}}(\mle;K)$, computed from paired draw-level differences.
For Spatial ASSURE, the \sure{}-trained grid center and diagonal derivative $H_i$ are recomputed within each training draw $Y^{*(b)}$.  The proposition below applies to a rule whose training and selection operations are functions of the training draw.

\begin{proposition}[Coupled-bootstrap targeting evaluation]\label{prop:assure_cb_unbiased}
Fix a measurable selection rule $S:\bbR^n\to\{0,1\}^n$.  Suppose $Y=\theta+\varepsilon$, $\varepsilon\sim \N(0,\Sigma)$, and $\xi\sim \N(0,\Sigma)$ is independent of $Y$.  Let
$Y^*=Y+\sqrt{\alpha}\xi$ and
$Y^\dagger=Y-\alpha^{-1/2}\xi$.  Then
\[
\E\left[
\frac{1}{n}\sum_{i=1}^n (Y_i^\dagger-K)S_i(Y^*)
\right]
=
\E\left[
\frac{1}{n}\sum_{i=1}^n (\theta_i-K)S_i(Y^*)
\right].
\]
Thus the coupled-bootstrap evaluation is unbiased for the welfare of the rule trained on the variance-inflated sample
$Y^*\sim \N(\theta,(1+\alpha)\Sigma)$.  If the welfare of the rule is continuous in the training-noise variance at $\alpha=0$, this target converges to the welfare of the same rule trained on the original sampling distribution as $\alpha\downarrow0$.
\end{proposition}

\begin{proof}[Proof of Proposition~\ref{prop:assure_cb_unbiased}]
Because $\varepsilon$ and $\xi$ are jointly Gaussian and $(Y^*,Y^\dagger)$ is an affine function of $(\varepsilon,\xi)$, the pair $(Y^*,Y^\dagger)$ is jointly Gaussian with mean $(\theta,\theta)$ and covariance
\[
\Cov(Y^*,Y^\dagger)
=
\Cov(Y,Y)-\Cov(\sqrt{\alpha}\xi,\alpha^{-1/2}\xi)
=
\Sigma-\Sigma
=0.
\]
Zero cross-covariance between jointly Gaussian blocks implies independence, so $\E[Y^\dagger\mid Y^*]=\E[Y^\dagger]=\theta$.  Conditioning on $Y^*$ gives
\[
\E\left[
\frac{1}{n}\sum_{i=1}^n (Y_i^\dagger-K)S_i(Y^*)
\mid Y^*
\right]
=
\frac{1}{n}\sum_{i=1}^n
(\theta_i-K)S_i(Y^*),
\]
and taking expectations proves the displayed identity.  As $\alpha\downarrow0$, $Y^*\sim \N(\theta,(1+\alpha)\Sigma)$ converges in distribution to $\N(\theta,\Sigma)$, so the continuity condition of the proposition gives the final claim.
\end{proof}

Figure~\ref{fig:assure_welfare} displays this comparison for Black children with genders pooled, the outcome closest to the Opportunity Atlas application in \citet{chenCompoundSelectionDecisions2025}.  The Black-male subgroup gives the same ranking pattern.

\begin{figure}[H]
  \centering
  \includegraphics[width=\textwidth]{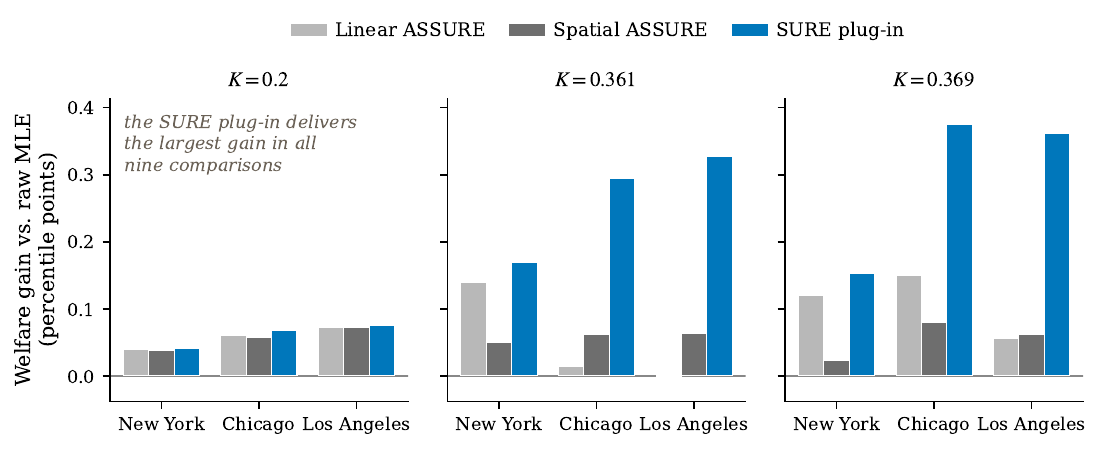}
  \caption{Targeting-welfare comparison with ASSURE in the three largest CZs, New York, Chicago, and Los Angeles.  Each bar is a coupled-bootstrap estimate of the welfare gain over the raw-\mle{} targeting rule, in percentile points of household income rank, for Black children with genders pooled ($B=\assureB{}$, $\alpha=\assureAlphaCb{}$).  Panels vary the cost threshold $K$, following the Opportunity Atlas application in \citet{chenCompoundSelectionDecisions2025}.  Linear ASSURE (a one-parameter variance-based threshold adjustment) and Spatial ASSURE (a grid search over geographic \gp{} kernel parameters) maximize the smoothed targeting-welfare objective within the restricted classes described above.  The \sure{} plug-in rule trains the geographic \gp{} by minimizing \sure{} with no welfare tuning and is evaluated on the same targeting task.}
  \label{fig:assure_welfare}
\end{figure}

\end{document}

%% file: tables/auto_oa_results.tex
\newcommand{\oaNCzs}{20}
\newcommand{\oaTotalTracts}{25,777}
\newcommand{\oaMinN}{723}
\newcommand{\oaMaxN}{3,859}
\newcommand{\oaMedianN}{1,016}

\newcommand{\oaReductionLsClose}{38}
\newcommand{\oaReductionGpGeo}{48}

\newcommand{\oaReductionGpContigOls}{53}

\newcommand{\oaGeoWins}{8}
\newcommand{\oaContigWins}{12}
\newcommand{\oaGeoPlainWins}{8}

\newcommand{\oaGeoOlsWins}{7}

\newcommand{\oaAggFullMinusFixed}{0.007}

\newcommand{\oaAggReportedSure}{0.45}
\newcommand{\oaReductionAggMle}{55}
\newcommand{\oaReductionAggVsClose}{27}
\newcommand{\oaAggBeatsBest}{16}
\newcommand{\oaAggTotal}{20}

\newcommand{\oaWeightGpGeo}{20.1}

\newcommand{\oaWeightGpGeoOls}{27.6}
\newcommand{\oaWeightGpContigOls}{39.0}

%% file: tables/auto_oa_proxy_gaps.tex
\newcommand{\oaProxyGapPoints}{120}
\newcommand{\oaProxyGapNCzs}{20}

%% file: tables/auto_oa_assure_welfare.tex
\newcommand{\assureB}{100}
\newcommand{\assureAlphaCb}{0.1}

%% file: tables/auto_oa_target_scatter.tex
\newcommand{\oaWedgeSharePct}{39}
\newcommand{\oaScatterNTracts}{1,304}

%% file: tables/auto_oa_ladder.tex
\newcommand{\oaLadderAggBeforeBilat}{0.519}
\newcommand{\oaLadderAggAfterBilat}{0.507}

\newcommand{\oaLadderBestBeforeBilat}{0.544}
\newcommand{\oaLadderBestAfterBilat}{0.522}

\newcommand{\oaLadderBilatWeight}{0.655}

%% file: tables/auto_oa_table.tex
\begin{table}[!htbp]
  \centering
  \begin{threeparttable}
  \caption{\sure{}-estimated MSE relative to raw \mle{} across \oaNCzs{} commuting zones ($n = \oaTotalTracts{}$ total tracts, pooled outcome).}\label{tab:oa_results}
  \begin{tabular}{lccc}
    \toprule
    Method & \sure{} MSE / \mle{} MSE & Reduction vs.\ \mle{} & Average weight \\
    \midrule
    \addlinespace
    \multicolumn{4}{@{}l}{\textit{Non-spatial baselines}} \\
    \addlinespace[2pt]
  \quad \mle{} & 1.000 & 0\% & 0.0\% \\
  \quad \nneb{} & 0.752 & 25\% & 1.0\% \\
  \quad \closegauss{} & 0.619 & 38\% & 3.6\% \\
    \addlinespace[6pt]
    \multicolumn{4}{@{}l}{\textit{Spatial (geographic distance)}} \\
    \addlinespace[2pt]
  \quad \gp{} Geo & 0.515 & 48\% & 20.1\% \\
  \quad \gp{} Geo OLS & 0.473 & 53\% & 27.6\% \\
    \addlinespace[6pt]
    \multicolumn{4}{@{}l}{\textit{Spatial (contiguity distance)}} \\
    \addlinespace[2pt]
  \quad \gp{} Contig & 0.511 & 49\% & 8.8\% \\
  \quad \gp{} Contig OLS & 0.466 & 53\% & 39.0\% \\
    \midrule
  \textbf{\sure{} Average} & \textbf{0.450} & \textbf{55\%} & --- \\
    \bottomrule
  \end{tabular}
  \begin{tablenotes}[flushleft]\footnotesize
  \item \textit{Notes:} Lower values in the \sure{} MSE / \mle{} MSE column are better. Values are tract-weighted averages across \oaNCzs{} commuting zones of the \sure{} estimate for each CZ divided by the corresponding raw-\mle{} MSE benchmark, using the diagonal matrix of reported tract variances.  Candidate \sure{} differentiates through trained parameters (Section~\ref{sec:learned_params}).  The average row reports \sure{} for the \sure{}-weighted average, including the derivative of the selected weights with respect to the data; this is an evaluation quantity for the final map, not the fixed-weight oracle criterion in Proposition~\ref{prop:averaging}.  Average weight is the tract-weighted average weight assigned by the candidate \sure{} QP; it is not applicable to the average row.
  \end{tablenotes}
  \end{threeparttable}
\end{table}

%% file: tables/auto_oa_targeting_welfare_table.tex
\begin{table}[!htbp]
  \centering
  \begin{threeparttable}
  \caption{Coupled-bootstrap estimates of top-third targeting gains relative to raw \mle{} across Opportunity Atlas outcomes.}\label{tab:oa_targeting_welfare}
  \setlength{\tabcolsep}{3pt}
  \begin{tabular}{@{}lcccc@{}}
    \toprule
    Method & Pooled & Pooled male & Black male & White male \\
    \midrule
  \nneb{} & -0.09 (-\$115) & -0.16 (-\$184) & -0.44 (-\$454) & -0.03 (-\$27) \\
  \closegauss{} & +0.08 (+\$102) & +0.12 (+\$136) & +1.14 (+\$1,196) & +0.19 (+\$230) \\
  \gp{} Geo & +0.23 (+\$283) & +0.37 (+\$435) & +1.12 (+\$1,173) & +0.67 (+\$834) \\
  \sure{} Average & +0.21 (+\$264) & +0.34 (+\$403) & +1.21 (+\$1,265) & +0.61 (+\$751) \\
    \bottomrule
  \end{tabular}
  \begin{tablenotes}[flushleft]\footnotesize
  \item \textit{Notes:} Cells report gains relative to \mle{}. The first number is the selected group's average rank gain in percentile-rank points; the value in parentheses is the 2015 dollar-equivalent conversion using the official Opportunity Atlas percentile-dollar crosswalk (\texttt{kid\_hh\_income} column). Dollar values are interpretive conversions of rank outcomes, not separately estimated dollar outcomes. Values are tract-weighted averages across commuting zones using $B=20$ coupled-bootstrap draws per CZ at perturbation level $\alpha=0.1$, selecting the top third of tracts. The \gp{} Geo row is the geographic-distance \gp{} with Nadaraya--Watson preprocessing and no OLS residualization. The targeting evaluation construction is described in Appendix~\ref{app:oa_robustness}. Outcome coverage: Pooled: 20 CZs, $n=25{,}777$; Pooled male: 20 CZs, $n=25{,}669$; Black male: 20 CZs, $n=10{,}026$; White male: 20 CZs, $n=21{,}643$.
  \end{tablenotes}
  \end{threeparttable}
\end{table}

%% file: tables/auto_oa_proxy_exact_table.tex
\begin{table}[H]
  \centering
  \begin{threeparttable}
  \small
  \caption{Fixed-parameter proxy comparisons across \oaNCzs{} commuting zones ($n = \oaTotalTracts{}$ total tracts, pooled outcome).}\label{tab:oa_proxy_exact}
  \begin{tabular}{@{}lcccc@{}}
    \toprule
    Method & \shortstack{Fixed-param.\ proxy\\ / \mle{} MSE} & \shortstack{\sure{}\\ / \mle{} MSE} & \sure{} $-$ proxy & Weight \\
    \midrule
    \addlinespace
    \multicolumn{5}{@{}l}{\textit{Non-spatial baselines}} \\
    \addlinespace[2pt]
  \quad \mle{} & 1.000 & 1.000 & 0.000 & 0.0\% \\
  \quad \nneb{} & 0.751 & 0.752 & +0.002 & 1.0\% \\
  \quad \closegauss{} & 0.601 & 0.619 & +0.019 & 3.6\% \\
    \addlinespace[6pt]
    \multicolumn{5}{@{}l}{\textit{Spatial (geographic distance)}} \\
    \addlinespace[2pt]
  \quad \gp{} Geo & 0.510 & 0.515 & +0.005 & 20.1\% \\
  \quad \gp{} Geo OLS & 0.469 & 0.473 & +0.004 & 27.6\% \\
    \addlinespace[6pt]
    \multicolumn{5}{@{}l}{\textit{Spatial (contiguity distance)}} \\
    \addlinespace[2pt]
  \quad \gp{} Contig & 0.508 & 0.511 & +0.004 & 8.8\% \\
  \quad \gp{} Contig OLS & 0.461 & 0.466 & +0.005 & 39.0\% \\
    \midrule
  \textbf{Fixed-Weight Proxy Average} & \textbf{0.443} & & & \\
  \textbf{\sure{} Average} & & \textbf{0.450} & \textbf{+0.007} & \\
    \bottomrule
  \end{tabular}
  \begin{tablenotes}[flushleft]\footnotesize
  \item \textit{Notes:} The fixed-parameter proxy treats trained parameters as constants; it is not a \sure{} value.  Both normalized MSE columns divide by the corresponding raw-\mle{} MSE benchmark.  The \sure{} column differentiates through trained parameters and is used in the main analysis (Section~\ref{sec:learned_params}).  The gap is \sure{} minus the fixed-parameter proxy; positive values indicate proxy optimism.  For \nneb{} and \closegauss{}, the proxy holds their moment-based mean and variance estimates fixed.  Weight is the tract-weighted average weight from the candidate \sure{} QP.  The average rows compare the fixed-weight proxy with \sure{} for the \sure{}-weighted average.  This table is supplementary and is not the main performance table.
  \end{tablenotes}
  \end{threeparttable}
\end{table}

%% file: tables/auto_oa_robustness_table.tex
\begin{table}[H]
  \centering
  \begin{threeparttable}
  \caption{Four-candidate \gp{} Geo comparison across Opportunity Atlas KFR outcomes.}\label{tab:oa_robustness}
  \setlength{\tabcolsep}{3pt}
  \scriptsize
  \begin{tabular*}{\textwidth}{@{\extracolsep{\fill}}lcccc@{}}
    \toprule
    Method & \multicolumn{1}{c}{Pooled} & \multicolumn{1}{c}{Pooled male} & \multicolumn{1}{c}{Black male} & \multicolumn{1}{c}{White male} \\
    \midrule
    \addlinespace[6pt]
    \multicolumn{5}{@{}l}{\textit{Panel A. \sure{} MSE / \mle{} MSE}} \\
    \addlinespace[2pt]
    \multicolumn{5}{@{}l}{\textit{Non-spatial baselines}} \\
    \addlinespace[2pt]
  \quad \mle{} & 1.000 & 1.000 & 1.000 & 1.000 \\
  \quad \nneb{} & 0.752 & 0.586 & 0.393 & 0.349 \\
  \quad \closegauss{} & 0.619 & 0.467 & 0.215 & 0.340 \\
    \addlinespace[6pt]
    \multicolumn{5}{@{}l}{\textit{Spatial candidate (geographic distance)}} \\
    \addlinespace[2pt]
  \quad \gp{} Geo & 0.515 & 0.381 & 0.188 & 0.291 \\
    \midrule
  \textbf{\sure{} Average} & \textbf{0.510} & \textbf{0.372} & \textbf{0.175} & \textbf{0.275} \\
    \midrule
    \addlinespace[6pt]
    \multicolumn{5}{@{}l}{\textit{Panel B. Tract-weighted average weight}} \\
    \addlinespace[2pt]
    \multicolumn{5}{@{}l}{\textit{Non-spatial baselines}} \\
    \addlinespace[2pt]
  \quad \mle{} & 0.000 & 0.000 & 0.000 & 0.000 \\
  \quad \nneb{} & 0.075 & 0.125 & 0.034 & 0.309 \\
  \quad \closegauss{} & 0.107 & 0.110 & 0.398 & 0.059 \\
    \addlinespace[6pt]
    \multicolumn{5}{@{}l}{\textit{Spatial candidate (geographic distance)}} \\
    \addlinespace[2pt]
  \quad \gp{} Geo & 0.818 & 0.764 & 0.568 & 0.631 \\
    \bottomrule
  \end{tabular*}
  \begin{tablenotes}[flushleft]\footnotesize
  \item \textit{Notes:} Lower values in Panel~A are better. Panel~A values are tract-weighted averages across commuting zones of the \sure{} estimate for each CZ divided by the corresponding raw-\mle{} MSE benchmark, with each CZ weighted by its tract count.  Panel~B reports the corresponding tract-weighted average weights from the \sure{} weight optimization; weights in each column sum to one up to rounding.  All columns use the same four-candidate \gp{} Geo comparison: \mle{}, \nneb{}, \closegauss{}, and \gp{} Geo.  The \sure{} Average row reports \sure{} for the \sure{}-weighted average (evaluated as in Table~\ref{tab:oa_results}).  Candidate \sure{} differentiates through trained parameters.  The main pooled-outcome table uses the larger seven-candidate library.  Outcome coverage: Pooled: 20 CZs, $n = 25,777$; Pooled male: 20 CZs, $n = 25,669$; Black male: 20 CZs, $n = 10,026$; White male: 20 CZs, $n = 21,643$.
  \end{tablenotes}
  \end{threeparttable}
\end{table}